\documentclass[11pt,a4paper]{article}

\usepackage[pagebackref,plainpages=false,colorlinks,linkcolor=blue,anchorcolor=teal,citecolor=blue]{hyperref}
\usepackage{amsmath} 
\usepackage{amsthm} 
\usepackage{amssymb}	
\usepackage{graphicx} 
\usepackage{multicol} 
\usepackage{multirow}
\usepackage[dvipsnames]{xcolor}
\usepackage[dvips,margin=1in,bottom=1in]{geometry}
\usepackage[capitalize,noabbrev]{cleveref}

\allowdisplaybreaks[1]

\usepackage[utf8]{inputenc}
\usepackage[english]{babel}

\usepackage{diagbox}
\usepackage{mathtools}

\usepackage{xspace}
\usepackage[T1]{fontenc}
\AtBeginDocument{%
  \DeclareFontShape{T1}{cmr}{m}{scit}{<->ssub*cmr/m/sc}{}%
}
\usepackage{comment}
\usepackage{authblk}

\usepackage{enumerate}
\usepackage{enumitem}

\usepackage{adjustbox}
\usepackage{footnotehyper} 
\makesavenoteenv{table}  

\usepackage{makecell}
\usepackage{float}
\usepackage{threeparttable}

\usepackage{thmtools}
\declaretheorem[style=plain,numberwithin=section]{theorem}
\declaretheorem[style=plain,numberlike=theorem]{lemma,corollary}
\declaretheorem[style=remark,numberlike=theorem]{remark}
\declaretheorem[style=plain,numberlike=theorem]{definition}
\declaretheorem[style=plain,numberlike=theorem]{proposition}

\declaretheorem[style=definition]{problem,conjecture}

\DeclarePairedDelimiter\rbra{\lparen}{\rparen}
\DeclarePairedDelimiter\sbra{\lbrack}{\rbrack}
\DeclarePairedDelimiter\cbra{\{}{\}}
\DeclarePairedDelimiter\abs{\lvert}{\rvert}
\DeclarePairedDelimiter\Abs{\lVert}{\rVert}
\DeclarePairedDelimiter\ceil{\lceil}{\rceil}
\DeclarePairedDelimiter\floor{\lfloor}{\rfloor}
\DeclarePairedDelimiter\ket{\lvert}{\rangle}
\DeclarePairedDelimiter\bra{\langle}{\rvert}
\DeclarePairedDelimiter\ave{\langle}{\rangle}

\newcommand{\ketbra}[2]{\ensuremath{\ket{#1}\!\bra{#2}}}

\newcommand{\tr} {\operatorname{tr}}
\newcommand{\poly} {\operatorname{poly}}

\newcommand{\rank} {\operatorname{rank}}
\newcommand{\supp} {\operatorname{supp}}

\newcommand{\sgn} {\operatorname{sgn}}

\newcommand{\negl} {\textrm{negl}}

\newcommand{\BQP}{\textnormal{\textsf{BQP}}\xspace}

\newcommand{\QMA}{\textnormal{\textsf{QMA}}\xspace}
\newcommand{\NQP}{\textnormal{\textsf{NQP}}\xspace}
\newcommand{\coNQP}{\textnormal{\textsf{coNQP}}\xspace}
\newcommand{\coCeP}{\mathsf{coC_{=}P}}
\newcommand{\CeP}{\mathsf{C_{=}P}}
\newcommand{\PP}{\textnormal{\textsf{PP}}\xspace}
\newcommand{\SZK}{\textnormal{\textsf{SZK}}\xspace}
\newcommand{\QSZK}{\textnormal{\textsf{QSZK}}\xspace}
\newcommand{\NIQSZK}{\textnormal{\textsf{NIQSZK}}\xspace}

\newcommand{\Real} {\operatorname{Re}}
\newcommand{\Imag} {\operatorname{Im}}

\newcommand{\Sq}{\mathrm{S}_q}
\newcommand{\D}{\mathrm{D}}

\newcommand{\F}{\mathrm{F}}

\newcommand{\td}{\mathrm{T}}
\newcommand{\Talpha}{{\mathrm{T}_\alpha}}
\newcommand{\Lalpha}{{\Lambda_\alpha}}

\newcommand{\TV}{\mathrm{TV}}

\newcommand{\QSD}{\textnormal{\textsc{QSD}}\xspace}
\newcommand{\PureQSD}{\textnormal{\textsc{PureQSD}}\xspace}
\newcommand{\QSDalpha}{\texorpdfstring{\textnormal{\textsc{QSD}\textsubscript{\ensuremath{\alpha}}}}\xspace}
\newcommand{\PoweredQSDalpha}{\texorpdfstring{\textnormal{\textsc{PoweredQSD}\textsubscript{\ensuremath{\alpha}}}}\xspace}
\newcommand{\PureQSDalpha}{\textnormal{\textsc{PureQSD}\textsubscript{\ensuremath{\alpha}}}\xspace}
\newcommand{\PurePoweredQSDalpha}{\textnormal{\textsc{PurePoweredQSD}\textsubscript{\ensuremath{\alpha}}}\xspace}

\newcommand{\PurePoweredQSDinfty}{\textnormal{\textsc{PurePoweredQSD}\textsubscript{\ensuremath{\infty}}}\xspace}
\newcommand{\QJSP}{\textnormal{\textsc{QJSP}}\xspace}
\newcommand{\measQTDP}{\textnormal{\textsc{measQTDP}}\xspace}

\newcommand{\binset}{\{0,1\}}
\newcommand{\innerprod}[2]{\left\langle #1 | #2 \right\rangle}
\newcommand{\innerprodF}[2]{\langle #1 , #2 \rangle}
\newcommand{\dx}{\mathrm{d}x}

\newcommand{\dd}{\mathrm{d}}

\newcommand{\bbN}{\mathbb{N}}
\newcommand{\bbR}{\mathbb{R}}

\newcommand{\sfF}{\mathsf{F}}
\newcommand{\sfO}{\mathsf{O}}
\newcommand{\ttU}{\mathtt{U}}

\newcommand{\calE}{\mathcal{E}}
\newcommand{\calF}{\mathcal{F}}
\newcommand{\calI}{\mathcal{I}}

\newcommand{\calP}{\mathcal{P}}
\newcommand{\calH}{\mathcal{H}}

\newcommand{\CNOT}{\textnormal{\textsc{CNOT}}\xspace}

\usepackage{enumerate}

\usepackage{tcolorbox}

\usepackage{algorithm}
\usepackage{algpseudocode}
\renewcommand{\algorithmicrequire}{\textbf{Input:}} 
\renewcommand{\algorithmicensure}{\textbf{Output:}} 

\usepackage{tabularx}
\usepackage{booktabs}

\usepackage{subfig}
\usepackage{tikz}
\usetikzlibrary{arrows}
\usepackage{quantikz}

\begin{document}
\setlength{\abovedisplayskip}{6pt}
\setlength{\belowdisplayskip}{6pt}

\title{On estimating Schatten norm and power distances between quantum states}

\author[1,2]{Yupan Liu\thanks{Email: \url{yupan.liu@epfl.ch}}}
\author[3,4]{Qisheng Wang\thanks{Email: \url{QishengWang1994@gmail.com}}}
\affil[1]{School of Computer and Communication Sciences, \'Ecole Polytechnique F\'ed\'erale de Lausanne}
\affil[2]{Graduate School of Mathematics, Nagoya University}
\affil[3]{School of Computer Science, Shanghai Jiao Tong University}
\affil[4]{School of Informatics, University of Edinburgh}
\date{}

\maketitle
\pagenumbering{roman}
\thispagestyle{empty}

\begin{abstract}
    We study the computational complexity of estimating the quantum Schatten $\alpha$-norm distance ${\mathrm{T}_\alpha}(\rho_0,\rho_1)$, defined via the Schatten $\alpha$-norm $\|A\|_{\alpha} \coloneqq \mathrm{tr}(|A|^{\alpha})^{1/\alpha}$, given $\operatorname{poly}(n)$-size state-preparation circuits of $n$-qubit quantum states $\rho_0$ and $\rho_1$. 
    This quantity serves as a lower bound on the trace distance and, for $\alpha > 1$, is interchangeable with its powered version ${\Lambda_\alpha}(\rho_0,\rho_1)$.
    For any constant $\alpha > 1$, we develop an efficient \textit{rank-independent} quantum estimator for ${\mathrm{T}_\alpha}(\rho_0,\rho_1)$, via ${\Lambda_\alpha}(\rho_0,\rho_1)$, with time complexity $\operatorname{poly}(n)$, achieving an \textit{exponential} speedup over the prior best results of $\exp(n)$ due to \hyperlink{cite.WGL+22}{Wang, Guan, Liu, Zhang, and Ying (TIT 2024)}. Our improvement leverages efficiently computable \textit{uniform} polynomial approximations of \textit{signed} positive power functions within quantum singular value transformation, thereby eliminating the dependence on the rank of the quantum states. 
    When $0<\alpha<1$ is a constant, ${\rm T}_\alpha(\rho_0,\rho_1)$ is no longer a distance metric, whereas $\Lambda_\alpha(\rho_0,\rho_1)$ remains one; we refer to the latter as the quantum Schatten $\alpha$-power distance. Accordingly, we provide a \emph{rank-efficient} quantum estimator for $\Lambda_\alpha (\rho_0,\rho_1)$, which serves as a natural substitute for the Schatten $\alpha$-norm distance below $\alpha=1$.
    
    Our quantum algorithm reveals a dichotomy in the computational complexity of the \textsc{Quantum State Distinguishability Problem with Schatten $\alpha$-norm} ($\textsc{QSD}_{\alpha}$), which involves deciding whether ${\mathrm{T}_\alpha}(\rho_0,\rho_1)$ is at least $2/5$ or at most $1/5$. This dichotomy arises between the cases of $\alpha > 1$ and $0 < \alpha\leq 1$:
    \begin{enumerate}[label={\upshape(\arabic*)}]
        \item For any constant $\alpha>1$, $\textsc{QSD}_{\alpha}$ is $\mathsf{BQP}$-complete. 
        \item For any $1 \leq \alpha(n) \leq 1+{\rm negl}(n)$, $\textsc{QSD}_{\alpha}$ is $\mathsf{QSZK}$-complete, implying that no efficient quantum estimator for $\mathrm{T}_\alpha(\rho_0,\rho_1)$ exists unless $\mathsf{BQP} = \mathsf{QSZK}$. This $\mathsf{QSZK}$-hardness result also extends to the promise problem defined by $\Lambda_\alpha(\rho_0,\rho_1)$ for constant $0<\alpha<1$.
    \end{enumerate}
    The hardness results follow from reductions based on new rank-dependent  inequalities for ${\rm T}_\alpha(\rho_0,\rho_1)$ when $1\leq \alpha \leq \infty$ and for $\Lambda_\alpha(\rho_0,\rho_1)$ when $0<\alpha<1$, which are of independent interest.
\end{abstract}

\newpage
\tableofcontents
\thispagestyle{empty}

\newpage
\pagenumbering{arabic}
\section{Introduction}

Closeness testing of quantum states is a central topic in quantum property testing~\cite{MdW16}, which aims to develop (efficient) quantum testers for properties of quantum objects. This problem is also closely related to verifying the functionality of quantum devices, such as $Q_0$ and $Q_1$, which are commonly designed to prepare the respective $n$-qubit (mixed) quantum states $\rho_0$ and $\rho_1$. The goal of (tolerant) quantum state testing is to design efficient quantum algorithms that test whether $\rho_0$ is $2/5$-far from or $1/5$-close to $\rho_1$ with respect to a given closeness measure. Notably, this problem generalizes classical (tolerant) distribution testing (see~\cite{Canonne20} and~\cite[Chapter 11]{Goldreich17}) from a non-commutative perspective. 

When the ``source codes'' of distribution- or state-preparation circuits are given, a surprising correspondence was established between such closeness testing problems -- measured by the $\ell_1$ norm distance~\cite{SV97,GV99} or entropy difference~\cite{GSV98} -- and interactive proof systems that admit statistical zero-knowledge (\SZK{}). This correspondence links closeness testing problems to both complexity theory and cryptography. 
A similar correspondence was later identified in the quantum world: closeness testing of quantum states with respect to the trace distance (given by Schatten $1$-norm)~\cite{Wat02,Wat09}, denoted by \QSD{}, or the von Neumann entropy difference~\cite{BASTS10} was shown to be \QSZK{}-complete.\footnote{The \QSZK{} containment of the closeness testing problem with respect to the trace distance, denoted by $\QSD[a(n),b(n)]$, holds only in the polarizing regime $a(n)^2-b(n) \geq 1/O(\log{n})$, as shown in~\cite{Wat02,Wat09}. A recent work~\cite{Liu23} slightly improved the parameter regime for this containment.\label{fnote:polarizing regime}} 

In contrast, when the closeness measure follows an $\ell_2$-norm-like definition, such as the Hilbert-Schmidt distance or the quantum linear entropy, the corresponding closeness testing problems are in \BQP{} using the SWAP test~\cite{BCWdW01,EAO+02}. Taken together, these results reveal a dichotomy in the complexity of closeness testing: when the measure is $\ell_1$-norm-like, the problems are \QSZK{}-hard and their query or sample complexities have \textit{polynomial} dependence on the dimension or rank of the states; whereas for $\ell_2$-norm-like measures, the problems are contained in \BQP{} and their query or sample complexities are \textit{rank-independent}. 

What about the closeness testing problems with respect to generalizations that approximate the trace distance or the von Neumann entropy? The quantum Schatten $\alpha$-norm distance $\Talpha(\cdot,\cdot)$, formulated in terms of the Schatten $\alpha$-norm, and its powered version $\Lalpha(\cdot,\cdot)$ are defined by
\[
\Talpha(\rho_0,\rho_1) \coloneqq \frac{1}{2} \tr\rbra*{|\rho_0-\rho_1|^\alpha}^{1/\alpha}
\quad\text{and}\quad
\Lalpha(\rho_0,\rho_1) \coloneqq \frac{1}{2} \tr\rbra*{|\rho_0-\rho_1|^\alpha}.
\]
Notably, the quantum Schatten $\alpha$-norm distance generalizes both the trace distance ($\alpha=1$) and the Hilbert-Schmidt distance ($\alpha=2$). 
Similarly, the quantum $q$-Tsallis entropy $\Sq(\rho)$ extends both von Neumann entropy ($q=1$) and quantum linear entropy ($q=2$). 

Interestingly, prior results show a divergence in behavior for closeness measures looser than the $\ell_2$ norm: The closeness testing problem with respect to $\Talpha(\rho_0,\rho_1)$, denoted by \QSDalpha{} (see \cref{def:QSDalpha}), is in \BQP{} only for \textit{even} integer $\alpha \geq 2$ via the Shift test~\cite{EAO+02};\footnote{Similarly, the closeness testing problem with respect to $\Lalpha(\rho_0,\rho_1)$, denoted by \PoweredQSDalpha{} (see \cref{def:PoweredQSDalpha-less-1}), is also in \BQP{} for \textit{even} integer $\alpha \geq 2$ using the same technique in~\cite{EAO+02}.} while for odd integers $\alpha \geq 3$, the query and sample complexities generally depend on the rank~\cite{WGL+22}. However, the techniques in~\cite{EAO+02} yield \BQP{} algorithms for estimating $\Sq(\rho)$ for \textit{all} integer $q \geq 2$. 
A recent work~\cite{LW25} further explored the closeness testing problem with respect to $\Sq(\rho_0) - \Sq(\rho_1)$, and extended the observed dichotomy from integers -- where the transition occurs between $q=1$ and $q\geq 2$ -- to a continuous setting, showing a sharp distinction between $q=1$ and any constant $q>1$. These results naturally lead to an intriguing question:

\begin{problem}
    \label{prob:quantum-Lp-distance}
    What is the computational complexity of the closeness testing problem with respect to $\Talpha(\rho_0,\rho_1)$? Does a similar dichotomy hold between $\alpha=1$ and constants $\alpha > 1$, or does the complexity vary largely depending on whether $\alpha$ is even or odd? 
\end{problem}

Why study $\ell_{\alpha}$ problems for possibly non-integer $\alpha>1$? The trace distance ($\alpha=1$) is a fundamental closeness measure of quantum states, capturing the maximum success probability of quantum state discrimination~\cite{Hol73,Hel67} and playing a key role in applications such as the security of quantum key distribution~\cite{BHLMO05,RK05}. For $\alpha>1$, such as $\alpha=1.001$, the quantum Schatten $\alpha$-norm distance provides a natural lower bound on the trace distance, and addressing \Cref{prob:quantum-Lp-distance} could make this bound efficiently computable. Moreover, insights from $\ell_{\alpha}$ problems have previously contributed to progress on well-studied $\ell_1$ problems, as seen in~\cite{LN04}. 

Beyond their connections to $\ell_1$ problems, $\ell_\alpha$ problems for $\alpha>1$ are of independent interest. In classical scenarios, they have applications in machine learning (e.g.,~\cite{KBSZ11}), as well as in streaming and sketching algorithms (e.g.~\cite{Indyk06}). In quantum scenarios, the Hilbert-Schmidt distance ($\alpha=2$) is widely used in quantum information theory (e.g.,~\cite{HRJJ04,PSW20}), and more recently, has been leveraged in designing near-term (variational) quantum algorithms (e.g., ~\cite{ACSZC19,EBS+23}). Consequently, positive answers to \Cref{prob:quantum-Lp-distance} may offer new opportunities to refine, extend, or develop techniques relevant to these areas.

A classical counterpart to \Cref{prob:quantum-Lp-distance} was investigated in~\cite{Waggoner15} nearly a decade ago. 
The main takeaway aligns with~\cite{LW25}: For constant $\alpha > 1$, the sample complexity for distinguishing whether $\TV_{\alpha}(D_0,D_1)$ is at least $2/5$ or at most $1/5$ is \textit{independent} of the dimension of the probability distributions $D_0$ and $D_1$, fewer samples are needed as $\alpha$ increases.\footnote{The closeness measure $\TV_{\alpha}(D_0,D_1)$ represents the classical $\ell_\alpha$ distance based on the $\ell_\alpha$ norm and generalizes the total variation distance, which is recovered at $\alpha=1$.} Classically, these upper bounds can be achieved by drawing a polynomial number of samples and computing the $\ell_{\alpha}$ norm distance between the resulting empirical distributions. However, this approach does not directly extend to the quantum world for two reasons: (1) quantum states $\rho_0$ and $\rho_1$ are generally not simultaneously diagonalizable; and (2) even when they are, estimating their eigenvalues associated with the unknown common eigenbasis remains challenging.\footnote{Consider a variant of the closeness testing problem \QSD{} in which $\rho_1$ is fixed to be the maximally mixed state $I/2^n$. Even though both $\rho_0$ and $\rho_1$ can be simultaneously diagonalized in the eigenbasis of $\rho_0$, this problem remains difficult to solve efficiently in general unless $\BQP = \NIQSZK$~\cite{Kobayashi03,CCKV08,BASTS10}. The complexity class \NIQSZK{} refers to non-interactive quantum statistical zero-knowledge~\cite{Kobayashi03}.} 
Addressing these challenges is central to resolving \Cref{prob:quantum-Lp-distance}, which is the main focus of our work.

In addition to the regime $\alpha\geq 1$, where the quantum Schatten $\alpha$-norm distance $\Talpha(\cdot,\cdot)$ serves as a distance metric, we also study the regime $0<\alpha<1$. In this regime, $\Talpha(\cdot,\cdot)$ no longer satisfies the triangle inequality, whereas its powered version $\Lalpha(\cdot,\cdot)$ does define a distance metric, a fact dating back to~\cite{Rotfel69} (see also the finite-dimensional formulations in~\cite{BL10,FP19}). This property makes $\Lalpha(\cdot,\cdot)$ a natural substitute for the Schatten $\alpha$-norm distance below $\alpha=1$, while retaining sensitivity to the singular-value profile of $\rho_0-\rho_1$ beyond what is captured by the trace distance alone. We refer to $\Lalpha(\cdot,\cdot)$ as the quantum Schatten $\alpha$-power distance, also known as the quantum Schatten $\alpha$-quasi-norm distance. The regime $0<\alpha<1$ is further motivated by the broader role of Schatten quasi-norm powers in matrix analysis and low-rank optimization~\cite{MMBS15,SLC16,GLRH20}, as well as by  quantum information quantities based on sublinear trace powers, such as quantum R\'enyi-type entropies~\cite{MLDSFT13,Tomamichel15,KFR26}. From this perspective, we also investigate quantum algorithms for estimating $\Lalpha(\cdot,\cdot)$ when $0 < \alpha < 1$ in our work.

\subsection{Main results}

We begin by stating our first main theorem about efficient quantum algorithms for estimating the quantum Schatten $\alpha$-norm distances for $\alpha > 1$ and the Schatten $\alpha$-power distances for $0 < \alpha < 1$, in both cases for \emph{constant} $\alpha$:
\begin{theorem}[Quantum estimators for $\Talpha$ and $\Lalpha$, informal] 
    \label{thm:QSDalpha-containment-informal}
    Given quantum query access to the state-preparation circuits of $n$-qubit quantum states $\rho_0$ and $\rho_1$, 
    \begin{itemize}
        \item For constant $\alpha > 1$, there is a quantum algorithm that estimates $\Talpha(\rho_0,\rho_1)$ to within additive error $1/5$ with query complexity $O(1)$.
        Furthermore, if the state-preparation circuits have $\poly(n)$-size descriptions, then the time complexity of the algorithm is $\poly(n)$. Consequently, \QSDalpha{} is in \BQP{}. 
        \item For constant $0 < \alpha < 1$, there is a quantum algorithm that estimates $\Lalpha(\rho_0, \rho_1)$ to within additive error $1/5$ with query complexity $\widetilde{O}(r^{\frac{2}{\alpha}-1})$, where $r$ is a known upper bound on the ranks of both $\rho_0$ and $\rho_1$.
    \end{itemize}
\end{theorem}

More precisely, for a given additive error $\epsilon$, the explicit complexities of estimating $\Talpha(\rho_0,\rho_1)$ and $\Lalpha(\rho_0, \rho_1)$ are presented below and compared to the results in the prior work.
A technical comparison of our approach with this result is provided in \Cref{subsec:techniques-upper-bounds}. 

\begin{itemize}
    \item For constant $\alpha > 1$, the query complexity of estimating $\Talpha(\rho_0,\rho_1)$ is $O\rbra{1/\epsilon^{\alpha+1+\frac{1}{\alpha-1}}}$ (see \cref{thm:Lalpha-query-algo}). 
    In combination with the samplizer \cite{WZ23,WZ24b}, estimating $\Talpha(\rho_0,\rho_1)$ can be done using $O\rbra{\log^2\rbra{1/\epsilon}/\epsilon^{3\alpha+2+\frac{2}{\alpha-1}}}$ samples of $\rho_0$ and $\rho_1$ and $O\rbra{n\log^2\rbra{1/\epsilon}/\epsilon^{3\alpha+2+\frac{2}{\alpha-1}}}$ elementary quantum gates (see \cref{thm:Lalpha-sample-algo}).
    In combination with the samplizer implied by \cite{TWZ25} (cf.\ \cite{CWZ25}), another approach to estimating $\Talpha(\rho_0,\rho_1)$ is obtained, which has a better sample complexity of $O\rbra{1/\epsilon^{2\alpha+2+\frac{2}{\alpha-1}}}$ while with a worse time complexity (i.e., the number of elementary quantum gates) of $\widetilde{O}\rbra{\poly\rbra{n}/\epsilon^{8\alpha+8+\frac{8}{\alpha-1}}}$ (see \cref{remark:sample-time-tradeoff}).\footnote{The two quantum algorithms have different sample complexities and time complexities: one has a better sample complexity and the other has a better time complexity. They are both built on the quantum query algorithms but use different samplization processes.} 
    All of these upper bounds can be expressed as $\poly(1/\epsilon)$ for constant $\alpha > 1$. 
    In addition, if the state-preparation circuits of $\rho_0$ and $\rho_1$ have size $\poly(n)$, then any of these three approaches implies a quantum algorithm with time complexity $\poly(n,1/\epsilon)$. 

    Previous quantum algorithms for estimating the quantum Schatten $\alpha$-norm distance $\Talpha(\rho_0,\rho_1)$ have all relied on its powered variant, specifically the quantum Schatten $\alpha$-power distance: 
    \[\Lalpha(\rho_0,\rho_1) \coloneqq \frac{1}{2} \tr\rbra*{\abs{\rho_0-\rho_1}^\alpha} = 2^{\alpha-1} \cdot \Talpha(\rho_0,\rho_1)^{\alpha}.\]
    Thus, the estimates of $\Talpha(\rho_0,\rho_1)$ and $\Lalpha(\rho_0,\rho_1)$ are polynomially related. 
    When $\alpha \geq 2$ is an even integer, estimating $\Lalpha(\rho_0,\rho_1)$ follows from a folklore result via the Shift test~\cite{EAO+02}, using $O\rbra{1/\epsilon}$ queries or $O\rbra{1/\epsilon^2}$ samples.\footnote{The sample complexity was noted in~\cite[Equations (83) and (84)]{QKW22}.} 
    However, for odd integer $\alpha > 1$, no efficient quantum algorithm is known in general. 
    The closeness testing of quantum states with respect to $\Talpha(\rho_0,\rho_1)$ for $\alpha=3$, with query complexity $O(1/\epsilon^{3/2})$, has been noted only in~\cite{GL20}.
    For general non-integer constants $\alpha>1$, a query complexity of $\widetilde{O}\rbra{r^{3+\frac{1}{\cbra{\alpha/2}}}/\epsilon^{4\alpha + \frac{\alpha}{\cbra{\alpha/2}}}}$ was given in \cite{WGL+22} for estimating $\Talpha(\rho_0,\rho_1)$, with polynomial dependence on the maximum rank $r$ of $\rho_0$ and $\rho_1$, where $\cbra{x} = x - \floor{x}$ denotes the fractional part of $x$. 
    In comparison, the quantum query algorithm for $\alpha > 1$ in \cref{thm:QSDalpha-containment-informal} exponentially improves the prior best result in \cite{WGL+22} in the query complexity and time complexity.

    \item For constant $0 < \alpha < 1$, the query complexity of estimating $\Lalpha(\rho_0,\rho_1)$ is $O\rbra{r^{\frac{2}{\alpha}-1}\log\rbra{r/\epsilon}/\epsilon^{\frac{2}{\alpha}}}$ (see \cref{thm:alpha-quasi-query}), improving the prior best query complexity $\widetilde{O}\rbra{r^{\frac{5}{\alpha}}/\epsilon^{\frac{5}{\alpha}+1}}$ given in \cite{WGL+22}. 
    In combination with the samplizer \cite{WZ23,WZ24b}, estimating $\Lalpha(\rho_0,\rho_1)$ can be done using $O\rbra{r^{\frac{5}{\alpha}-3}\log^4\rbra{r/\epsilon}/\epsilon^{\frac{5}{\alpha}}}$ samples of $\rho_0$ and $\rho_1$ and $O\rbra{nr^{\frac{5}{\alpha}-3}\log^4\rbra{r/\epsilon}/\epsilon^{\frac{5}{\alpha}}}$ elementary quantum gates (see \cref{thm:alpha-quasi-sample}).
    In combination with the samplizer implied by \cite{TWZ25} (cf.\ \cite{CWZ25}), another approach to estimating $\Lalpha(\rho_0,\rho_1)$ is obtained, which has a better sample complexity of $O\rbra{r^{\frac{4}{\alpha}-2}\log^2\rbra{r/\epsilon}/\epsilon^{\frac{4}{\alpha}}}$ while with a worse time complexity of $\widetilde{O}\rbra{\poly\rbra{n} r^{\frac{16}{\alpha}-8}\log^8\rbra{r/\epsilon}/\epsilon^{\frac{16}{\alpha}}}$ (see \cref{thm:alpha-le-1-better-samplizer}). 
    All of these upper bounds can be expressed as $\poly(r/\epsilon)$ for constant $0 < \alpha < 1$. 
    In addition, if the state-preparation circuits of $\rho_0$ and $\rho_1$ have size $\poly(n)$, then any of these three approaches implies a quantum algorithm with time complexity $\poly(n,r,1/\epsilon)$. 
\end{itemize}

\vspace{1em}
By combining our efficient quantum estimator for $\Talpha(\rho_0,\rho_1)$ in the regime of fixed real constant $\alpha>1$ (\Cref{thm:QSDalpha-containment-informal}), with our hardness results for \QSDalpha{} (\Cref{thm:QSDalpha-hardness-informal}), we identify a sharp phase transition between $\alpha=1$ and any fixed real constant $\alpha > 1$, addressing \Cref{prob:quantum-Lp-distance}. For clarity, we summarize our main theorems and the quantitative bounds on quantum query and sample complexities, derived from both our results and prior work, in \Cref{table:QSDalpha-complexities}. 

\begin{table}[H]
\centering
\adjustbox{max width=\textwidth}{
\begin{threeparttable}
\caption{Computational, query, and sample complexities of (\textsc{Powered})\QSDalpha{}.}
\label{table:QSDalpha-complexities}
\begin{tabular}{cccccc}
    \toprule
     & & $0<\alpha<1$ & $\alpha=1$ & $1 < \alpha(n) \leq 1\!+\!\negl(n)$ & $\alpha>1$\\
    \midrule
    \multicolumn{2}{c}{(\textsc{Powered})\QSDalpha{}} & \makecell{\QSZK{}-hard\\ \footnotesize{\Cref{thm:PoweredQSDalpha-QSZK-hard-less-1}}} & \makecell{\QSZK{}-complete\tnote{*}\\ \footnotesize{\cite{Wat02,Wat09}}} & \makecell{\QSZK{}-complete\tnote{*}\\ \footnotesize{\Cref{thm:QSDalpha-hardness-informal}\ref{thmitem:QSDalpha-hard-regime-informal}
}} & \makecell{\BQP{}-complete\\ \footnotesize{\Cref{thm:QSDalpha-containment-informal,thm:QSDalpha-hardness-informal}\ref{thmitem:QSDalpha-easy-regime-informal}}}\\
    \midrule
    \multirow{2}{*}{\makecell{Query\\\scriptsize{Complexity}}} & \makecell{\scriptsize{Upper}\\\scriptsize{Bound}} & \makecell{$\widetilde{O}\rbra*{ r^{\frac{2}{\alpha}-1} / \epsilon^{\frac{2}{\alpha}}}$ \\\footnotesize{\Cref{thm:alpha-quasi-query}}}
    & \makecell{$\widetilde{O}\rbra*{r/\epsilon^2}$\\\footnotesize{\cite{WZ24}}} & \makecell{$\widetilde{O}\rbra*{r^{3+\frac{2}{\alpha}}/\epsilon^{4\alpha+2}}$\\\footnotesize{\cite{WGL+22}}} & \makecell{$O\rbra*{1/\epsilon^{\alpha+1+\frac{1}{\alpha-1}}}$\\\footnotesize{\Cref{thm:Lalpha-query-algo}}}\\
    \cmidrule{2-6}
     & \makecell{\scriptsize{Lower}\\\scriptsize{Bound}} & \makecell{$\Omega\rbra[\big]{r^{\frac{1}{\alpha}-1}/\epsilon^{\frac{1}{\alpha}}}$\\\footnotesize{\Cref{thm:lower-bounds-PoweredQSDalpha}\ref{thmitem:queries-0leAle1}
}} & \makecell{$\Omega\rbra*{r^{1/2}/\epsilon}$\\\footnotesize{\cite{CWZ25}}} & \makecell{$\Omega\rbra*{r^{1/2}/\epsilon}$\\\footnotesize{\Cref{thm:query-lower-bounds-QSDalpha}}} & \makecell{$\Omega\rbra*{1/\epsilon}$\\\footnotesize{\Cref{thm:lower-bounds-PureQSDalpha}\ref{thmitem:easy-regime-queries}}}\\
     \midrule
    \multirow{2}{*}{\makecell{Sample\\\scriptsize{Complexity}}} & \makecell{\scriptsize{Upper}\\\scriptsize{Bound}} & \makecell{$\widetilde{O}\rbra*{ r^{\frac{4}{\alpha}-2} / \epsilon^{\frac{4}{\alpha}}}$\\\footnotesize{\Cref{thm:alpha-le-1-better-samplizer}}} & \makecell{$\widetilde{O}\rbra*{r^2/\epsilon^4}$\\\footnotesize{\cite{CWZ25}}} & \makecell{$\poly(r,1/\epsilon)$\\\footnotesize{Noted in \cite[Footnote 2]{WGL+22}}} & \makecell{$\widetilde{O}\rbra*{1/\epsilon^{3\alpha+2+\frac{2}{\alpha-1}}}$\\\footnotesize{\Cref{thm:Lalpha-sample-algo}}}\\
    \cmidrule{2-6}
     & \makecell{\scriptsize{Lower}\\\scriptsize{Bound}} & \makecell{$\Omega\rbra[\big]{r^{\frac{2}{\alpha}-1}/\epsilon^{\frac{2}{\alpha}}}$\\\footnotesize{\Cref{thm:lower-bounds-PoweredQSDalpha}\ref{thmitem:samples-0leAle1}}} & \makecell{$\Omega\rbra*{r/\epsilon^2}$\\\footnotesize{\cite{OW21}}} & \makecell{$\Omega\rbra*{r/\epsilon^2}$\\\footnotesize{\Cref{thm:sample-lower-bound-QSDalpha}}} & \makecell{$\Omega\rbra*{1/\epsilon^2}$\\\footnotesize{\Cref{thm:lower-bounds-PureQSDalpha}\ref{thmitem:easy-regime-samples}}}\\
    \bottomrule
\end{tabular}
\begin{tablenotes}
    \footnotesize
    \item[*] For any $\alpha(n) \in \big[1,1+\negl(n)]$, the promise problem $\QSD_{\alpha}[a,b]$ is contained in \QSZK{} only under the polarizing regime $a(n)^2-b(n)\geq 1/O(\log{n})$, which can be slightly improved when $\alpha=1$ (see \Cref{fnote:polarizing regime}). However, establishing containment in a complexity class typically requires the natural regime $a(n)-b(n) \geq 1/\poly(n)$, as in \Cref{thm:QSDalpha-containment-informal}.
\end{tablenotes}
\end{threeparttable}
}
\end{table}

Finally, we present our second main theorem, which addresses the computational hardness of \QSDalpha{} ($\alpha\geq 1$) and \PoweredQSDalpha{} ($0<\alpha<1$), as outlined in \Cref{thm:QSDalpha-hardness-informal}. In this context, \PureQSDalpha{} refers to a restricted variant of \QSDalpha{} (see also \cref{def:QSDalpha}), where the states of interest are pure. For convenience, we say that $\alpha$ is \emph{fixed} if either $\alpha$ is a real constant or $\alpha=\infty$. 

\begin{theorem}[Computational hardness of \QSDalpha{} and \PoweredQSDalpha{}]
    \label{thm:QSDalpha-hardness-informal}
    The promise problems \QSDalpha{} and \PoweredQSDalpha{} capture the computational power of the respective complexity classes in the following regimes:
    \begin{enumerate}[label={\upshape(\arabic*)}, topsep=0.33em, itemsep=0.33em, parsep=0.33em]
        \item \label{thmitem:QSDalpha-easy-regime-informal}\textbf{\emph{Easy regimes:}} For every \emph{fixed} $1 \leq \alpha \leq \infty$, \PureQSDalpha{} is \BQP{}-hard. Consequently, both \QSDalpha{} and \PoweredQSDalpha{} are \BQP{}-complete for every \emph{constant} $\alpha > 1$. 
        \item \label{thmitem:QSDalpha-hard-regime-informal}\textbf{\emph{Hard regimes for \QSDalpha{}:}} For every negligible function $\mu(n)$ and every $1 \leq \alpha \leq 1+\mu(n)$, \QSDalpha{} is \QSZK{}-complete, where the \QSZK{} containment of $\QSDalpha[a,b]$ holds only in the polarizing regime $a(n)^2-b(n)\geq 1/O(\log{n})$. 
        \item \label{thmitem:PoweredQSDalpha-hard-regime-informal}\textbf{\emph{Hard regimes for \PoweredQSDalpha{}:}} For every \emph{constant} $0<\alpha<1$, \PoweredQSDalpha{} is \QSZK{}-hard.
    \end{enumerate}
\end{theorem}
It is noteworthy that the \QSZK{}-hardness regime in \Cref{thm:QSDalpha-hardness-informal}\ref{thmitem:QSDalpha-hard-regime-informal} essentially extends to the less restrictive regime $1 \leq \alpha(n) \leq 1+\frac{1}{n}$.

\subsection{Proof techniques: Quantum algorithms for \texorpdfstring{$\Talpha$}{} and \texorpdfstring{$\Lalpha$}{}}
\label{subsec:techniques-upper-bounds}

At a high level, Quantum Singular Value Transformation~\cite{GSLW19} implies that the main challenge in designing a quantum algorithm based on a smooth function -- such as Grover search~\cite{Grover96} and the OR function, or the HHL algorithm~\cite{HHL09} and the multiplicative inverse function (see~\cite{MRTC21} for additional examples) -- reduces to finding an efficiently computable polynomial approximation. 
Once such an approximation is obtained, the algorithm follows straightforwardly using techniques from~\cite{GSLW19}, with its efficiency determined entirely by the properties of the polynomial. 

Now we focus on quantum algorithms for estimating the Schatten power distance $\Lalpha\rbra{\rho_0, \rho_1}$. We begin by reviewing~\cite{WGL+22} and then provide a high-level overview of our approach. 
The quantum query complexity of estimating the quantum Schatten $\alpha$-norm distance for non-integer $\alpha$ was first considered in \cite[Theorem IV.1]{WGL+22}. 
Their approach begins with the identity
\[
2\Lalpha(\rho_0,\rho_1)=\Abs{\rho_0 - \rho_1}_\alpha^{\alpha} = \tr\rbra*{ \abs*{\nu_-}^{\alpha/2} \Pi_{\nu_+} \abs*{\nu_-}^{\alpha/2} },
\]
where $\nu_\pm = \rho_0 \pm \rho_1$ and $\Pi_{\nu_+}$ denotes the projector onto the support subspace of $\nu_+$. 
According to this identity, they aim to prepare a quantum state that is a block-encoding of (normalized) $\abs*{\nu_-}^{\alpha/2} \Pi_{\nu_+} \abs*{\nu_-}^{\alpha/2}$.\footnote{See \cref{def:block-encoding} for the formal definition of block-encoding.} 
To this end, they first prepare a quantum state that is a block-encoding of $\Pi_{\nu_+}$, and then perform a unitary operator that is a block-encoding of $\abs*{\nu_-}^{\alpha/2}$ on it.\footnote{This is because of the evolution of subnormalized density operators \cite[Lemma II.2]{WGL+22}.} 
Finally, the (unnormalized) Schatten power distance, $\Lalpha(\rho_0,\rho_1)$, can be obtained by estimating the trace of $\abs*{\nu_-}^{\alpha/2} \Pi_{\nu_+} \abs*{\nu_-}^{\alpha/2}$ using quantum amplitude estimation \cite{BHMT02}. 
After the error analysis, their approach was shown to have query complexity $\widetilde O\rbra{r^{3+1/\cbra{\alpha/2}}/\epsilon^{4+1/\cbra{\alpha/2}}} = \poly\rbra{r, 1/\epsilon}$.\footnote{Here, $\cbra{x} \coloneqq x - \floor{x}$ denotes the fractional part of $x$.}
The dependence on the rank is \textit{inherent} in the approach of \cite{WGL+22}, as they have to prepare a rank-dependent quantum state that is a block-encoding of $\Pi_{\nu_+}$, making the rank parameters unavoidable in the error analysis. 

To overcome this technical issue, we utilize an identity different from theirs:
\[
2\Lalpha(\rho_0,\rho_1) = \Abs{\rho_0-\rho_1}_\alpha^\alpha = \tr\rbra*{\rho_0 \cdot \sgn\rbra{\nu_-} \cdot \abs*{\nu_-}^{\alpha-1}} - \tr\rbra*{\rho_1 \cdot \sgn\rbra{\nu_-} \cdot \abs*{\nu_-}^{\alpha-1}}.
\]
The idea is to estimate the terms $\tr\rbra{\rho_j \cdot \sgn\rbra{\nu_-} \cdot \abs*{\nu_-}^{\alpha-1}}$ for $j \in \cbra{0, 1}$ individually, and then combine them to obtain an estimate of $\Lalpha(\rho_0,\rho_1)$.
Our algorithm is sketched as follows: 
\begin{enumerate}[topsep=0.33em, itemsep=0.33em, parsep=0.33em]
    \item Find a good approximation polynomial for $\sgn\rbra{x} \cdot \abs{x}^{\alpha - 1}$. 
    \item Implement a unitary block-encoding $U$ of $\sgn\rbra{\nu_-} \cdot \abs{\nu_-}^{\alpha-1}$ using Quantum Singular Value Transformation (QSVT) \cite{GSLW19} and Linear Combinations of Unitaries (LCU) \cite{CW12,BCC+15}, given the state-preparation circuits of $\rho_0$ and $\rho_1$.
    \item Perform the Hadamard test \cite{AJL09} on $U$ and $\rho_j$ with outcome $b_j \!\in\! \cbra{0,\!1}$ for each $j \!\in\! \cbra{0,\!1}$.
    \item Estimate $\Lalpha(\rho_0,\rho_1)$ by computing the expected value of $b_0-b_1$. 
\end{enumerate}

Our algorithm is actually inspired by the trace distance estimation in \cite{WZ24}, which is essentially the case of $\alpha = 1$. 
Even though the approach in \cite{WZ24} still has a rank-dependent query complexity of $\widetilde O\rbra{r/\epsilon^2}$, compared to the $\widetilde O\rbra{r^5/\epsilon^6}$ in \cite{WGL+22}.\footnote{Some readers may wonder if our approach applies to trace distance estimation ($\alpha = 1$) so that the rank-dependent query complexity $\widetilde O\rbra{r/\epsilon^2}$ and sample complexity $O\rbra{r^2/\epsilon^5}$ in \cite{WZ24} can be made rank-independent. 
However, it turns out that the answer is generally no, as the dependence on the rank $r$ is actually intrinsic for trace distance estimation due to the quantum query complexity lower bound $\Omega\rbra{\sqrt{r}/\epsilon}$ in \cite{CWZ25} (see also \Cref{lemma:lower-bounds-QSD}\ref{thmitem:query-lower-bound-QSD}) and the quantum sample complexity lower bound $\Omega\rbra{r/\epsilon^2}$ in \cite{OW21} (see \cref{lemma:lower-bounds-QSD}\ref{thmitem:sample-lower-bound-QSD}).}

\paragraph{Constant $\alpha > 1$.}
Nevertheless, we discover an approach to estimating the quantum Schatten $\alpha$-norm distance with a \textit{rank-independent} complexity as long as $\alpha$ is constantly greater than $1$. 
Specifically, we use the best uniform approximation polynomial $P_d\rbra{x}$ (of degree $d$) for the function $\sgn\rbra{x} \cdot \abs{x}^{q}$ given in \cite[Theorem 8.1.1]{Ganzburg08} such that
\[
\max_{x \in \sbra{-1, 1}} \abs[\big]{ P_d\rbra{x} - \sgn\rbra{x} \cdot \abs{x}^{q} } \to \frac{1}{d^q}, \text{ as } d \to \infty. 
\]
Our use of the best uniform approximation by polynomials is inspired by the recent work \cite{LW25} on estimating the $q$-Tsallis entropy of quantum states for non-integer $q$, where they used the best uniform approximation polynomial for $x^q$ in the non-negative range $\sbra{0, 1}$ (given in \cite{Timan63}). 
The difference is that in our case, we have to further consider the sign of $x$, thereby requiring the polynomial approximation to behave well in the negative part. 
It turns out that the polynomial approximation given in \cite{Ganzburg08} is suitable for our purpose. 
Having noticed this, we then use the now standard techniques (used in \cite{LGLW23,LW25}) such as Chebyshev truncations and the de La Vall\'{e}e Poussin partial sum (cf.\ \cite{Rivlin90}) to construct efficiently computable asymptotically best approximation polynomials such that
\[
\max_{x \in \sbra{-1, 1}} \abs*{P\rbra{x} - \frac{1}{2} \sgn\rbra{x} \cdot \abs{x}^{q}} \leq \epsilon, ~ \max_{x \in \sbra{-1, 1}} \abs{P\rbra{x}} \leq 1, \text{ and } \deg\rbra{P} = O\rbra*{\frac{1}{\epsilon^{1/q}}}.
\]

Using this efficiently computable polynomial (with $q = \alpha - 1$) and with further analysis, we can then estimate the quantum Schatten $\alpha$-norm distance to within additive error $\epsilon$ with the desired query upper bound in \Cref{thm:QSDalpha-containment-informal}. Moreover, using the samplizer \cite{WZ23,WZ24b,TWZ25,CWZ25}, a quantum query-to-sample simulation, we can also achieve the desired sample upper bound. 

\paragraph{Constant $0 < \alpha < 1$.}
We can also obtain a quantum algorithm that estimates the Schatten power distance $\Lalpha\rbra{\rho_0, \rho_1}$ for $0 < \alpha < 1$. 
The idea is similar to that for $\alpha > 1$ but even simpler. 
Specifically, we use the polynomial approximation for scaled $\sgn\rbra{x} \cdot \abs{x}^{\alpha-1}$, which is already known in \cite{GSLW19} (see \cref{lemma:uniformPolyApprox-signedNegativePower}). 
Due to the significant truncation error when $x$ is close to $0$, the approach obtained in this way thus has query complexity $\widetilde{O}\rbra{r^{\frac{2}{\alpha}-1}/\epsilon^{\frac{2}{\alpha}}}$, which improves the prior best result of $\widetilde{O}\rbra{r^{\frac{5}{\alpha}}/\epsilon^{\frac{5}{\alpha}+1}}$ due to \cite{WGL+22} but is still polynomial in the rank $r$ of $\rho_0$ and $\rho_1$, and the dependence on rank is unavoidable (see \Cref{thm:lower-bounds-PoweredQSDalpha}).

\subsection{Proof techniques: Computational hardness for \texorpdfstring{$\Talpha$}{} and \texorpdfstring{$\Lalpha$}{}}

To establish the \BQP{}- and \QSZK{}-hardness results stated in \Cref{thm:QSDalpha-hardness-informal}, we reduce the promise problems \PureQSD{} and \QSD{}, where $\alpha=1$, to the corresponding promise problems \PureQSDalpha{} ($\alpha>1$),  \QSDalpha{} ($\alpha>1$ negligibly close to $1$), and \PoweredQSDalpha{} ($0<\alpha<1$). The key technique underlying these reductions is given by the following rank-dependent inequalities that generalize the case of $\alpha=2$ from~\cite{Coles12,CCC19}: 
\begin{theorem}[$\Talpha$ and $\Lalpha$ vs.~$\td$, informal]
    \label{thm:Talpha-vs-T-informal}
    For all quantum states $\rho_0$ and $\rho_1$, 
    \begin{enumerate}[label={\upshape(\arabic*)}, topsep=0.33em, itemsep=0.33em, parsep=0.33em]
        \item \label{thmitem:Talpha-vs-T-informal}For \emph{fixed} $\alpha \in [1,\infty]$, the following inequality holds:
        \[ 2^{1-\frac{1}{\alpha}} \cdot \Talpha(\rho_0,\rho_1) \leq \td(\rho_0,\rho_1) \leq 2 \rbra*{ \rank(\rho_0)^{1-\alpha} + \rank(\rho_1)^{1-\alpha} }^{-\frac{1}{\alpha}} \cdot \Talpha(\rho_0,\rho_1).\]
        For $\alpha = \infty$, the inequalities hold in the limit as $\alpha \to \infty$.
        \item \label{thmitem:Lalpha-vs-T-informal}For \emph{constant} $\alpha \in (0,1)$, the following inequality holds:
        \[ \td(\rho_0,\rho_1)^{\alpha} \leq \Lalpha(\rho_0,\rho_1) \leq \frac{\rank(\rho_0)^{1-\alpha}+\rank(\rho_1)^{1-\alpha}}{2}\,\td(\rho_0,\rho_1)^{\alpha}. \]
\end{enumerate}
\end{theorem}

The proof of \Cref{thm:Talpha-vs-T-informal} follows from considering orthogonal positive semi-definite matrices $\varsigma_0$ and $\varsigma_1$ satisfying $\rho_0-\rho_1 = \varsigma_0-\varsigma_1$, and analyzing their properties carefully. 

\vspace{1em}
We then illuminate the hardness results in \Cref{thm:QSDalpha-hardness-informal}:
\begin{itemize}[topsep=0.33em, itemsep=0.33em, parsep=0.33em]
    \item In the easy regime for constant $\alpha>1$, \Cref{thm:Talpha-vs-T-informal}\ref{thmitem:Talpha-vs-T-informal} becomes an equality when both $\rho_0$ and $\rho_1$ are pure states. This equality implies that the \BQP{}-hardness of \PureQSDalpha{}, as well as the query and sample complexity lower bounds, holds for all $1\leq \alpha \leq\infty$, thereby establishing \Cref{thm:QSDalpha-hardness-informal}\ref{thmitem:QSDalpha-easy-regime-informal}. 
    \item In the hard regime for \QSDalpha{}, \Cref{thm:Talpha-vs-T-informal}\ref{thmitem:Talpha-vs-T-informal} is \textit{sensitive} to $\alpha$. In particular, for $\alpha(n) =1+\frac{1}{n}$, if quantum states $\rho_0$ and $\rho_1$ are $\tau$-far, meaning $\td(\rho_0,\rho_1) \geq \tau$, it follows only that $\td_{\alpha(n)=1+\frac{1}{n}}(\rho_0,\rho_1) \geq \tau/2$. However, when $\alpha(n) \leq 1+\frac{1}{n^{1+\delta}}$ for any constant $\delta>0$ that can be made arbitrarily small, the same trace distance condition ensures that $\td_{\alpha}(\rho_0,\rho_1) \geq \tau$ as $n \to \infty$, leading to the \QSZK{}-hardness result in \Cref{thm:QSDalpha-hardness-informal}\ref{thmitem:QSDalpha-hard-regime-informal}. 
    \item In the hard regime for \PoweredQSDalpha{}, where $0<\alpha<1$ is constant, a direct application of \Cref{thm:Talpha-vs-T-informal}\ref{thmitem:Lalpha-vs-T-informal} is insufficient because of the multiplicative factor in the upper bound on $\Lalpha(\sigma_0,\sigma_1)$, while the other direction in the reduction is straightforward. To absorb this factor, we apply the polarization lemma for the trace distance to obtain two new $\widetilde{O}(n)$-qubit quantum states $\sigma_0$ and $\sigma_1$ such that $\td(\sigma_0,\sigma_1)$ is either at least $1-2^{-n^2}$ or at most $2^{-n^2}$. Consequently, using \Cref{thm:Talpha-vs-T-informal}\ref{thmitem:Lalpha-vs-T-informal}, we obtain the desired bound $\Lambda_\alpha(\sigma_0,\sigma_1) \leq 2^{-\widetilde{O}(n)}$ when $\td(\sigma_0,\sigma_1)\leq 2^{-n^2}$, thereby establishing the \QSZK{}-hardness result in \Cref{thm:QSDalpha-hardness-informal}\ref{thmitem:PoweredQSDalpha-hard-regime-informal}. 
\end{itemize}

Lastly, we explain the \QSZK{} containment for $\alpha>1$ negligibly close to $1$.
Simply combining \Cref{thm:Talpha-vs-T-informal} and the \QSZK{} containment of \QSD{} from~\cite{Wat02,Wat09} does not work, as the resulting \QSZK{} containment of $\QSDalpha[a,b]$ holds only for $a(n)^2/2-b(n) \geq 1/O(\log{n})$, which is even weaker than the polarizing regime defined in \Cref{fnote:polarizing regime}. To address this, we establish a \textit{partial} polarization lemma for $\Talpha$ (\Cref{lemma:partial-polarization-Talpha}), which ensures that, given quantum states $\rho_0$ and $\rho_1$ such that $\td(\rho_0,\rho_1)$ is either at least $a$ or at most $b$, we can construct new quantum states $\widetilde{\rho}_0$ and $\widetilde{\rho}_1$ such that $\Talpha(\widetilde{\rho}_0,\widetilde{\rho}_1)$ is either at least $\frac{1}{2}-\frac{1}{2} e^{-k}$ or at most $1/16$, as long as the parameters $a$ and $b$ are in the polarizing regime. 
\Cref{thm:QSDalpha-hardness-informal}\ref{thmitem:QSDalpha-hard-regime-informal} follows by combining this partial polarization lemma for $\Talpha$ with the polarization lemma for $\td$ in~\cite{Wat02}. 

\subsection{Discussion and open problems}
\label{subsec:discussion}

Although the quantum Schatten $\alpha$-norm distance $\Talpha(\cdot,\cdot)$ and its powered version $\Lalpha(\cdot,\cdot)$ are almost computationally interchangeable for constant $\alpha \geq 1$, their behavior differs significantly when $\alpha = \infty$:
\begin{itemize}[topsep=0.33em, itemsep=0.33em, parsep=0.33em]
    \item The quantity $\td_{\infty}(\rho_0,\rho_1)$ is equal to the largest singular value of $(\rho_0-\rho_1)/2$, namely $\max\cbra*{ \lambda_{\max}\rbra[\big]{\frac{\rho_0-\rho_1}{2}}, \lambda_{\max}\rbra[\big]{\frac{\rho_1-\rho_0}{2}} }$. The associated promise problem $\QSD_{\infty}$ is \BQP{}-hard and is contained in \QMA{}.\footnote{The verification circuit in the \QMA{} containment simply follows from phase estimation~\cite{Kitaev95}, where a (normalized) eigenvector corresponding to $\lambda_{\max}$ serves as a witness state.} However, establishing a \BQP{} containment appears challenging, as both $(\rho_0-\rho_1)/2$ and $(\rho_1-\rho_0)/2$ do not directly admit an efficiently computable basis -- unlike its classical counterpart in~\cite{Waggoner15}, which does. 
    \item The quantity $\Lambda_{\infty}(\rho_0,\rho_1)$ takes values in $\cbra*{0,1/2,1}$ for any states $\rho_0$ and $\rho_1$ and is nonzero if and only if the states are orthogonal, with at least one of them being pure. Thus, even the pure-state-restricted variant of the associated promise problem, $\PurePoweredQSDinfty[1,0]$, is $\CeP$-hard (see \Cref{sec:PurePoweredQSD-hardness}). Here, $\CeP=\coNQP$~\cite{ADH97,YY99}, a subclass of \PP{} that provides a precise variant of \BQP{}, ensuring acceptance for all \textit{yes} instances. 
\end{itemize}
This fundamental difference between these quantities raises an intriguing question on $\QSD_\infty$:
\begin{enumerate}[label={\upshape(\roman*)},itemsep=0.3em,topsep=0.3em,parsep=0.3em]
    \item What is the computational complexity of the promise problem $\QSD_{\infty}$, defined by $\td_{\infty}(\cdot,\cdot)$? Can we show that $\QSD_{\infty}$ is also in \BQP{}, or is it inherently more difficult? 
\end{enumerate}

Two further open problems concern quantitative bounds for \QSDalpha{} and \PoweredQSDalpha{}:
\begin{enumerate}[label={\upshape(\roman*)},itemsep=0.3em,topsep=0.3em,parsep=0.3em]
    \setcounter{enumi}{1}
    \item Can the query and sample bounds in \Cref{table:QSDalpha-complexities} be improved for constant $\alpha \geq 1$? Moreover, can tight bounds be established when the states have \textit{small support}, analogous to the classical case in~\cite[Table 1]{Waggoner15}? 
    \item Can the query and sample bounds in \Cref{table:QSDalpha-complexities} be improved for constant $0<\alpha<1$, and can the quadratic gap between the current lower and upper bounds be closed? 
\end{enumerate}

\subsection{Related works}
\label{sec:related-works}

Schatten $p$-norm estimation $\tr(|A|^p)$ of $O(\log n)$-local Hermitian $A$ on $n$ qubits to within additive error $2^{n-p} \epsilon \|A\|^p$ for $\epsilon(n) \leq 1/\poly(n)$ and real $p(n) \leq \poly(n)$ was shown to be $\mathsf{DQC1}$-complete in \cite{CM18}. 
Given a unitary block-encoding of a matrix $A$, in \cite{LS20}, they presented a quantum algorithm that estimates the Schatten $p$-norm $\rbra{\tr\rbra{\abs{A}^p}}^{1/p}$ to relative error $\epsilon$ for integer $p$, where a condition number $\kappa$ satisfying $A \geq I/\kappa$ is required for the case of odd $p$. 

The query complexity of $N$-dimensional quantum state certification (i.e., determine whether two quantum states are identical or $\epsilon$-far) with respect to trace distance was shown to be $O\rbra{N/\epsilon}$ in \cite{GL20}. 
The query complexity of trace distance estimation was shown to be $\widetilde O\rbra{r^5/\epsilon^6}$ in \cite{WGL+22} and later improved to $\widetilde O\rbra{r/\epsilon^2}$ in \cite{WZ24}, where $r$ is the rank of the quantum states, confirming a conjecture in \cite{CCC19} that \textit{low-rank} trace distance estimation is in $\BQP$.
Both Low-rank trace distance and fidelity estimations are known to be $\BQP$-complete \cite{AISW20,WZ24}.  
Based on the approach of \cite{WZ24}, space-bounded quantum state discrimination with respect to trace distance was shown to be $\mathsf{BQL}$-complete in \cite{LGLW23}. 
In addition to trace distance, fidelity is another important measure of the closeness between quantum states.
The query complexity of fidelity estimation was shown to be $\widetilde O\rbra{r^{12.5}/\epsilon^{13.5}}$ in \cite{WZC+23} and later improved to $\widetilde O\rbra{r^{5.5}/\epsilon^{6.5}}$ in \cite{WGL+22} and to $\widetilde O\rbra{r^{2.5}/\epsilon^5}$ in \cite{GP22}. 
Recently, the query complexity of pure-state trace distance and fidelity estimations was shown to be $\Theta\rbra{1/\epsilon}$ in \cite{Wang24pureQSD} and was recently extended in \cite{FW25} to estimating fidelity of a mixed state to a pure state.

In addition to the query complexity, the sample complexity has also been studied in the literature. 
In \cite{BOW19}, the sample complexity of $N$-dimensional quantum state certification was shown to be $\Theta\rbra{N/\epsilon^2}$ with respect to trace distance and $\Theta\rbra{N/\epsilon}$ with respect to fidelity.
The sample complexity of trace distance estimation is known to be $\widetilde O\rbra{r^2/\epsilon^5}$ in \cite{WZ24} and that of fidelity estimation is known to be $\widetilde O\rbra{r^{5.5}/\epsilon^{12}}$, where $r$ is the rank of quantum states. 
The sample complexity of pure-state squared fidelity estimation is known to be $\Theta\rbra{1/\epsilon^2}$ via the SWAP test \cite{BCWdW01}, where the matching lower bound was given in \cite{ALL22}. 
Recently, the sample complexity of pure-state trace distance and fidelity estimations was shown to be $\Theta\rbra{1/\epsilon^2}$ in \cite{WZ24c}, which was achieved by using the samplizer in \cite{WZ24b}.

\paragraph{Recent developments.}
In \cite{CWZ25}, a sample complexity of $O\rbra{1/{\epsilon^{2\alpha + 2 + \frac{2}{\alpha-1}}}}$ was presented after the conference version of this work \cite{LW25Lalpha}, improving the sample complexity $O\rbra{1/\epsilon^{3\alpha+2+\frac{2}{\alpha-1}}}$ given in \cref{thm:Lalpha-sample-algo}. 

\section{Preliminaries}

We assume a fundamental knowledge of quantum computation and quantum information theory. For an introduction, we refer the reader to the textbook~\cite{NC10}. 

We adopt the following notations throughout the paper: (1) $\sbra{n} \coloneqq \cbra{1, 2, \dots, n}$; (2) $\widetilde{O}(f)$ denotes $O(f \operatorname{polylog}(f))$, while $\widetilde{\Omega}(f)$ denotes $\Omega(f/\operatorname{polylog}(f))$; (3) $\ket{\bar{0}}$ represents $\ket{0}^{\otimes a}$ for integer $a>1$; and (4) a function $\mu(n)$ is said to be \textit{negligible} if for every integer $c \geq 1$, there is an integer $n_c > 0$ such that, for all $n \geq n_c$, $\mu\rbra{n} < n^{-c}$. We also use $\negl(n)$ to denote a negligible function.
In addition, the Schatten $\alpha$-norm of a matrix $A$ is defined as
\[
\Abs{A}_{\alpha} \coloneqq \rbra[\big]{\tr\rbra{\abs{A}^{\alpha}}}^{1/\alpha} = \rbra*{\tr\rbra[\Big]{\rbra[\big]{A^{\dagger} A}^{\alpha/2}}}^{1/\alpha}. 
\]

For simplicity, we use the notation $\Abs{A}$ to denote the operator norm (equivalently, the Schatten $\infty$-norm) of a matrix $A$. 
We also require the notion of the diamond norm distance between quantum channels. Let $\D(\calH)$ denote the set of all density matrices, also known as quantum states, which are positive semi-definite and have trace one, defined on a finite-dimensional Hilbert space $\calH$. 
For any two quantum channels $\calE$ and $\calF$ acting on $\D(\calH)$, the \textit{diamond norm distance} between them is defined as
\[ \Abs{\calE-\calF}_{\diamond} \coloneqq \sup_{\rho \in \D(\calH\otimes \calH')} \Abs*{\rbra*{\calE\otimes \calI_{\calH'}}(\rho) - \rbra*{\calF \otimes \calI_{\calH'}}(\rho)}_1. \]

\subsection{Closeness measures for quantum states}

We start by defining the trace distance and providing useful properties of this distance: 
\begin{definition}[Trace distance]
    Let $\rho_0$ and $\rho_1$ be two quantum states that are mixed in general. 
    The trace distance between $\rho_0$ and $\rho_1$ is defined by
    \[\td(\rho_0,\rho_1) \coloneqq \frac{1}{2}\tr\rbra{|\rho_0-\rho_1|}= \frac{1}{2}\Abs{\rho_0-\rho_1}_1.\]
\end{definition}
It is worth noting that the trace distance is a distance metric (e.g.,~\cite[Lemma 9.1.8]{Wil13}), with values ranging between $0$ and $1$. 
Additionally, we need the following inequalities that relate the trace distance to the (Uhlmann) fidelity, $\F(\rho_0,\rho_1) \coloneqq \tr|\sqrt{\rho_0}\sqrt{\rho_1}|$:
\begin{lemma}[Trace distance vs.~Uhlmann fidelity, adapted from~\cite{FvdG99}]
\label{lemma:traceDist-vs-fidelity}
Let $\rho_0$ and $\rho_1$ be two quantum states. Then, it holds that 
\[1-\F(\rho_0,\rho_1) \leq \td(\rho_0,\rho_1) \leq \sqrt{1-\F^2(\rho_0,\rho_1)}.\]
\end{lemma}
Moreover, it is evident that $\F\big(\rho_0^{\otimes l},\rho_1^{\otimes l}\big) = \F(\rho_0,\rho_1)^l$ for any integer $l \geq 1$. 

\vspace{1em}
Next, we define the quantum Schatten $\alpha$-norm distance $\Talpha(\cdot,\cdot)$ and the quantum Schatten $\alpha$-\emph{power} distance $\Lalpha(\cdot,\cdot)$, which generalize the trace distance ($\alpha=1$) via the Schatten norm and its \emph{powered} version, respectively. Notably, the quantum Schatten $\alpha$-norm distance coincides with the Hilbert-Schmidt distance when $\alpha=2$:

\begin{definition}[Quantum Schatten norm and power distances]
    \label{def:Lalpha-distance}
    Let $\rho_0$ and $\rho_1$ be two quantum states that are mixed in general. 
    The quantum Schatten $\alpha$-norm distance $\Talpha(\cdot,\cdot)$ and the quantum Schatten $\alpha$-power distance $\Lalpha(\cdot,\cdot)$ between $\rho_0$ and $\rho_1$ are defined as follows:
     \[\Talpha(\rho_0,\rho_1) \coloneqq \frac{1}{2} \Abs*{\rho_0 - \rho_1}_{\alpha}
     \quad\text{and}\quad 
     \Lalpha(\rho_0,\rho_1) \coloneqq \frac{1}{2} \Abs*{\rho_0 - \rho_1}_{\alpha}^{\alpha}.\]
\end{definition}

By the monotonicity of the Schatten norm, e.g.,~\cite[Equation (1.31)]{AS17}, it holds that:
\[ \forall \alpha \geq 1, \quad 0 \leq \Lalpha(\rho_0,\rho_1) \leq \Talpha(\rho_0,\rho_1) \leq \td(\rho_0,\rho_1) \leq 1.\]

As a corollary of the Davis convexity theorem~\cite{Davis57}, the quantum Schatten $\alpha$-norm distance $\Talpha(\cdot,\cdot)$ defines a distance metric for $\alpha \in [1,+\infty]$:
\begin{lemma}[Triangle inequality for $\Talpha$ with $\alpha\geq 1$, adapted from~{\cite[Proposition 1.16]{AS17}}]
    \label{lemma:Talpha-triangle-inequality}
    For any quantum states $\rho_0$, $\rho_1$, and $\rho_2$, the following holds:
    \[ \forall\alpha\in[1,+\infty], \quad \Talpha(\rho_0,\rho_1) + \Talpha(\rho_1,\rho_2) \geq \Talpha(\rho_0,\rho_2). \]
\end{lemma}

By the subadditivity of $\Abs*{f(|A|+|B|)}$ for every nonnegative concave function $f$, every symmetric norm, and all positive semi-definite matrices $A$ and $B$~\cite[Theorem 1.2]{BL10}, taking $f(t)=t^\alpha$ and the Schatten $1$-norm implies that $\Lalpha(\cdot,\cdot)$ defines a metric for $\alpha \in (0,1)$:

\begin{lemma}[Triangle inequality for $\Lalpha$ with $\alpha\in(0,1)$, adapted from~{\cite[Section 2.3]{FP19}}]
    \label{lemma:Lalpha-triangle-inequality}
    For any quantum states $\rho_0$, $\rho_1$, and $\rho_2$, the following holds:
    \[ \forall\alpha\in (0,1), \quad \Lalpha(\rho_0,\rho_1) + \Lalpha(\rho_1,\rho_2) \geq \Lalpha(\rho_0,\rho_2). \]
\end{lemma}

Lastly, we require the following relationship for additive error estimation between the quantum Schatten $\alpha$-norm distance and its powered version when $\alpha \geq 1$:
\begin{proposition}[$\Talpha$ vs.~$\Lalpha$] 
    \label{prop:Talpha-vs-powered-Talpha}
    The quantum Schatten $\alpha$-norm distance $\Talpha(\cdot,\cdot)$ and the quantum Schatten $\alpha$-power distance $\Lalpha(\cdot,\cdot)$ are related by the equality 
    \[\Talpha(\rho_0,\rho_1) = 2^{\frac{1}{\alpha}-1} \cdot \Lalpha(\rho_0,\rho_1)^{\frac{1}{\alpha}}.\] 
    Accordingly, if $x \geq 0$ is an estimate of $\Lalpha(\rho_0,\rho_1)$ to within additive error $\epsilon$, then $2^{\frac{1}{\alpha}-1} \cdot x^{\frac{1}{\alpha}}$ serves as an estimate of $\Talpha(\rho_0,\rho_1)$ to within additive error $2^{\frac{1}{\alpha}-1} \cdot \epsilon^{\frac{1}{\alpha}}$. 
\end{proposition}

\begin{proof}
    The equality between $\Talpha(\rho_0,\rho_1)$ and $\Lalpha(\rho_0,\rho_1)$ follows from a direct calculation. 
    Let $x$ be an estimate of $\Lalpha(\rho_0,\rho_1)$ to within additive error $\epsilon$, equivalently $\abs*{\Lalpha(\rho_0,\rho_1) - x} \leq \epsilon$.
    Then, assuming that $\abs*{(x +\delta)^{1/\alpha} - x^{1/\alpha} } \leq \abs{\delta}^{1/\alpha}$ for any $\alpha \geq 1$ and $x \geq 0$, it follows that
    \[ \abs*{\Talpha(\rho_0,\rho_1) - 2^{\frac{1}{\alpha}-1} x^{\frac{1}{\alpha}}} 
    = 2^{\frac{1}{\alpha}-1} \cdot \abs*{\Lalpha(\rho_0,\rho_1)^{\frac{1}{\alpha}} - x^{\frac{1}{\alpha}}}
    \leq 2^{\frac{1}{\alpha}-1} \cdot \abs*{\Lalpha(\rho_0,\rho_1) - x}^{\frac{1}{\alpha}}
    \leq 2^{\frac{1}{\alpha}-1} \cdot \epsilon^{\frac{1}{\alpha}}.\]

    By symmetry, it suffices to consider $x+\delta$ and $x$ with $\delta\geq 0$; the case $\delta=0$ is immediate. 
    It suffices to show that $f(x) \coloneqq (x +\delta)^{1/\alpha} - x^{1/\alpha} \leq \delta^{1/\alpha} = f(0)$ for any $\delta > 0$, $\alpha \geq 1$, and $x \geq 0$. We complete the proof by noting that $f'(x) = \frac{1}{\alpha}(x+\delta)^{\frac{1}{\alpha}-1} - \frac{1}{\alpha} x^{\frac{1}{\alpha}-1} \leq 0$, and thus $f(x)$ is monotonically non-increasing when $x\geq 0$. 
\end{proof}

\subsection{Closeness testing of quantum states via state-preparation circuits}
\label{subsec:state-closeness-testing}

We begin by defining the closeness testing of quantum states with respect to the trace distance, denoted as \QSD{}$[a,b]$,\footnote{While \Cref{def:QSD} aligns with the classical counterpart of \QSD{} defined in~\cite[Section 2.2]{SV97}, it is slightly less general than the definition in~\cite[Section 3.3]{Wat02}. Specifically, \Cref{def:QSD} assumes that the input length $m$ and the output length $n$ are \textit{polynomially equivalent}, whereas~\cite[Section 3.3]{Wat02} allows for cases where the output length (e.g., a single qubit) is \textit{much smaller} than the input length.} and two variants of this problem:   

\begin{definition}[Quantum State Distinguishability Problem, \QSD{}, adapted from~{\cite[Section 3.3]{Wat02}}]
	\label{def:QSD}
    Let $Q_0$ and $Q_1$ be quantum circuits acting on $m$ qubits \emph{(}``input length''\emph{)} and having $n$ specified output qubits \emph{(}``output length''\emph{)}, where $m(n)$ is a polynomial function of $n$. Let $\rho_i$ denote the quantum state obtained by running $Q_i$ on state $\ket{0}^{\otimes m}$ and tracing out the non-output qubits. For efficiently computable functions $a(n)$ and $b(n)$ satisfying $0\leq b(n) < a(n) \leq 1$, the promise problem $\QSD[a,b]$ asks one to decide whether: 
	\begin{itemize}[topsep=0.33em, itemsep=0.33em, parsep=0.33em]
		\item \emph{Yes:} A pair of quantum circuits $(Q_0,Q_1)$ such that $\td(\rho_0,\rho_1) \geq a(n)$; 
		\item \emph{No:} A pair of quantum circuits $(Q_0,Q_1)$ such that $\td(\rho_0,\rho_1) \leq b(n)$.
	\end{itemize}

    \noindent Furthermore, we denote the restricted version, where $\rho_0$ and $\rho_1$ are pure states, as \PureQSD{}. 
\end{definition}

In this work, we consider the \textit{purified quantum access input model}, as defined in~\cite{Wat02}, in both white-box and black-box scenarios: 
\begin{itemize}[leftmargin=2em]
    \item \textbf{White-box input model}: The input of the problem \QSD{} consists of descriptions of polynomial-size quantum circuits $Q_0$ and $Q_1$. Specifically, for $b\in\binset$, the description of $Q_b$ includes a sequence of polynomially many $1$- and $2$-qubit gates.
    \item \textbf{Black-box input model}: In this model, instead of providing the descriptions of the quantum circuits $Q_0$ and $Q_1$, only query access to $Q_b$ is allowed, denoted as $O_b$ for $b\in\binset$. For convenience, we also allow query access to $Q_b^{\dagger}$ and controlled-$Q_b$, denoted by $O_b^{\dagger}$ and controlled-$O_b$, respectively. 
\end{itemize}

In addition to query complexity, defined within the black-box input model, \textit{sample complexity} refers to the number of copies of quantum states $\rho_0$ and $\rho_1$ needed to accomplish a specific closeness testing task.

\subsubsection{Computational hardness of \QSD{} and \PureQSD{}}
Polarization lemmas for the total variation distance~\cite{SV97} and the trace distance~\cite{Wat02} share the same inequalities, enabling the \QSZK{}-hardness of $\QSD$ using the parameters specified in~\cite[Theorem 3.14]{BDRV19}:
\begin{lemma}[\QSD{} is \QSZK{}-hard]
    \label{lemma:QSD-is-QSZKhard}
    Let $a(n)$ and $b(n)$ be efficiently computable functions satisfying $a^2(n)-b(n) \geq 1/O(\log{n})$. For any constant $\tau \in (0,1/2)$, $\QSD[a,b]$ is \QSZK{}-hard when $a(n) \leq 1-2^{-n^{\tau}}$ and $b(n) \geq 2^{-n^{\tau}}$ for every $n\in\bbN$.
\end{lemma}

We also require a polarization lemma for the trace distance~\cite{Wat02}, with parameters derived from~\cite[Theorem 3.14]{BDRV19} and its time complexity follows from~\cite[Lemma 38]{CCKV08}:
\begin{lemma}[A polarization lemma for the trace distance, adapted from~{\cite[Section 4.1]{Wat02}}]
    \label{lemma:polarization-trace-distance}
    Let $Q_0$ and $Q_1$ be quantum circuits that prepare quantum states $\rho_0$ and $\rho_1$, respectively. There exists a deterministic procedure that, given an input $(Q_0,Q_1,a, b, k)$ where $a^2 > b$, outputs new quantum circuits $\widetilde{Q}_0$ and $\widetilde{Q}_1$ that prepare corresponding states $\widetilde{\rho}_0$ and $\widetilde{\rho}_1$, respectively. The resulting states satisfy the following:
    \begin{align*}
        \td(\rho_0,\rho_1) \geq a &\quad\Longrightarrow\quad \td(\widetilde{\rho}_0,\widetilde{\rho}_1) \geq 1-2^{-k},\\
        \td(\rho_0,\rho_1) \leq b &\quad\Longrightarrow\quad \td(\widetilde{\rho}_0,\widetilde{\rho}_1) \leq 2^{-k}. 
    \end{align*}
    Here, the states $\widetilde{\rho}_0$ and $\widetilde{\rho}_1$ are defined over $\widetilde{O}\rbra[\big]{n k^{O\rbra[\big]{\frac{b\ln(2/a^2)}{a^2-b}}}}$ qubits. Furthermore, when $k \leq O(1)$ or $a^2-b \geq \Omega(1)$, the time complexity of the procedure is polynomial in the size of $Q_0$ and $Q_1$, $k$, and $\exp\big( \frac{b \log(1/a^2)}{a^2-b} \big)$.
\end{lemma}

Using the construction in~\cite[Theorem 12]{RASW23} (see also~\cite[Lemma 4.23]{LGLW23} and~\cite[Theorem 4.1]{WZ24}), the following \BQP{}-hardness result holds:\footnote{For the choice $l(n'-1)=2n'+1$, the resulting promise thresholds follow by direct calculation. In particular, for \emph{yes} instances, the threshold becomes $\sqrt{1-2^{-2l(n'-1)}} = \sqrt{1-2^{-4n'-2}} \geq 1-2^{-n'}$, whereas for \emph{no} instances, it becomes $2^{\frac{-l(n'-1)+1}{2}} = 2^{\frac{-(2n'+1)+1}{2}} = 2^{-n'}$.}
\begin{lemma}[\PureQSD{} is \BQP{}-hard, adapted from~{\cite[Lemma 2.17]{LW25}}]
    \label{lemma:PureQSD-is-BQPhard}
    Let $\alpha(n)$ and $\beta(n)$ be efficiently computable functions such that $\alpha(n)-\beta(n) \geq 1/\poly(n)$. For any polynomial $l(n)$, let  $n' \coloneqq n+1$, $\PureQSD[\alpha(n'),\beta(n')]$ is \BQP{}-hard when $\alpha(n') \leq \sqrt{1-2^{-2l(n'-1)}}$ and $\beta(n') \geq 2^{\frac{-l(n'-1)+1}{2}}$ for every integer $n' \geq 2$.

    \noindent Specifically, by choosing $l(n'-1)=2n'+1$, it holds that\emph{:}  For every integer $n' \geq 2$, 
    \[\PureQSD\sbra*{1-2^{-n'}, 2^{-n'}} \text{ is } \BQP{}\text{-hard under Karp reduction}. \] 
\end{lemma}

\subsubsection{Quantitative lower bounds for \QSD{} and \PureQSD{}}
We begin by stating a query complexity lower bound for \QSD{}, applicable to \textit{any} promise error $\epsilon \in (0,1/2)$. For any $n$-qubit quantum state $\rho$ of rank $r$, we can define an $n$-qubit state $\rho_{\ttU}$ such that the eigenvalues of $\rho_{\ttU}$ form a uniform distribution on the support of $\rho$. 
Consider the spectral decomposition $\rho = \sum_{i \in [r]} \mu_i \ket{v_i}\bra{v_i}$, where $\{\ket{v_i}\}_{i\in [r]}$ is an orthonormal basis, we have $\td\rbra*{\rho,\rho_{\ttU}} = \TV(\mathbf{\mu},U_{r})$, where $U_{r}$ is a uniform distribution over $[r]$. Then, the following lemma applies to a broad range of $\epsilon$:

\begin{lemma}[Quantitative lower bounds for \QSD{}]
    \label{lemma:lower-bounds-QSD}
    Let $\epsilon \in (0, 1/2]$ and $\Pi_r$ be a projector onto a subspace of even dimension $r$. 
    Consider the following promise problem: given an unknown quantum state $\rho$ such that $\supp\rbra{\rho} \subseteq \supp\rbra{\Pi_r}$, determine whether $\rho = \rho_{\ttU}$ or that $\rho$ has eigenvalues $\frac{1\pm 2\epsilon}{r}$ each with multiplicity $\frac{r}{2}$ (in which case $\td(\rho,\rho_{\ttU}) \geq \epsilon$), where $\rho_{\ttU} = \Pi_r/r$. 
    Then, 
    \begin{enumerate}[label={\upshape(\arabic*)}, leftmargin=2em, topsep=0.33em, itemsep=0.33em, parsep=0.33em]
        \item \label{thmitem:query-lower-bound-QSD}\emph{\textbf{Queries}~(\cite[Theorem 2.13]{CWZ25})}: In the purified quantum query access model, the quantum query complexity is $\Omega\rbra[\big]{\sqrt{r}/\epsilon}$. 
        \item \label{thmitem:sample-lower-bound-QSD}\emph{\textbf{Samples}~(\cite[Corollary 4.3]{OW21})}: The quantum sample complexity is $\Omega\rbra[\big]{r/\epsilon^2}$. 
    \end{enumerate}
\end{lemma}

Notably, the lower bound $\Omega\rbra{\sqrt{r}/\epsilon
}$ in \Cref{lemma:lower-bounds-QSD}\ref{thmitem:query-lower-bound-QSD} improves the previous results: (i) $\Omega\rbra{r^{1/3}}$ due to \cite[Theorem 2]{CFMdW10} for the uniformity testing of discrete distributions and (ii) $\widetilde{\Omega}\rbra{\sqrt{r}}$ due to \cite[Theorem 5]{BKT20} for estimating the total variation distance. 

It is noteworthy that the quantum query model used in~\cite{CFMdW10,BKT20} differs from the purified quantum access model. However, this lower bound also applies to our query model, as discussed after Definition 3 in~\cite{GL20}. 

\vspace{1em}
Next, we present lower bounds on the query and sample complexities for \PureQSD{} by inspecting the proof of the corresponding theorems in~\cite{Wang24pureQSD}. It is noteworthy that the query complexity bound (\Cref{lemma:lower-bounds-PureQSD}\ref{thmitem:query-lower-bound-PureQSD}) follows as a corollary of~\cite[Theorem 4]{Belovs19}:

\begin{lemma}[Quantitative lower bounds for \PureQSD{}]
    \label{lemma:lower-bounds-PureQSD}    
    For any $\epsilon \in (0,1/2)$, there exist $n$-qubit pure states $\ket{\psi_0}$ and $\ket{\psi_1}$ such that deciding whether $\td\rbra*{\ketbra{\psi_0}{\psi_0},\ketbra{\psi_1}{\psi_1}}$ is at least $\epsilon$ or exactly $0$ requires:
    \begin{enumerate}[label={\upshape(\arabic*)}, leftmargin=2em, topsep=0.33em, itemsep=0.33em, parsep=0.33em]
        \item \emph{\textbf{Queries}~(\cite[Theorem V.2]{Wang24pureQSD}}: In the purified quantum access model, the quantum query complexity is $\Omega(1/\epsilon)$. 
        \label{thmitem:query-lower-bound-PureQSD}
        \item \emph{\textbf{Samples}~(\cite[Theorem B.2]{Wang24pureQSD})}: The quantum sample complexity is $\Omega\rbra[\big]{1/\epsilon^2}$. 
        \label{thmitem:sample-lower-bound-PureQSD}
    \end{enumerate}
\end{lemma}

\subsection{Polynomial approximations}
\label{subsec:polynomial-approx}
We now present a few useful results and tools for polynomial approximations. 

\subsubsection{(Best uniform) polynomial approximations}
\label{subsubsec:best-uniform-poly-approx}

Let $f(x)$ be a continuous function defined on the interval $[-1,1]$ that we aim to approximate by a polynomial of degree at most $d$. We define $P^*_d$ as a \textit{best uniform approximation} to $f$ on $[-1,1]$ of degree at most $d$ if, for any degree-$d$ polynomial approximation $P_d$ of $f$, the following holds:  
\[ \max_{x\in[-1,1]} \abs*{f(x) - P^*_d(x)} \leq \max_{x\in[-1,1]} \abs*{f(x) - P_d(x)}.\]

The best uniform (polynomial) approximation of positive (constant) powers $|x|^{\alpha}$ was first established by Serge Bernstein~\cite{Bernstein38,Bernstein38shift}. However, the focus here is on the best uniform approximation of \textit{signed} positive powers $\sgn(x) |x|^{\alpha}$, as stated in \Cref{lemma:bestUniformApprox-signedPositivePower}. This result is often attributed to Bernstein's work (see, e.g.,~\cite[Equation (10.2)]{Totik06}), and a proof of a more general version can be found in~\cite[Theorem 8.1.1]{Ganzburg08}. 

\begin{lemma}[Best uniform approximation of signed positive powers, adapted from~{\cite[Theorem 8.1.1]{Ganzburg08}}]
    \label{lemma:bestUniformApprox-signedPositivePower}
    For any fixed positive real constant order $\alpha$, let $P^*_d \in \bbR[x]$ be the best uniform polynomial approximation to $f(x) = \sgn(x)\abs{x}^{\alpha}$ of degree $d = \ceil*{\rbra*{\beta_{\alpha}/\epsilon}^{1/\alpha}}$, where $\beta_{\alpha}$ is a constant depending on $\alpha$. Then, for sufficiently small $\epsilon$, it holds that 
    \[\max_{x\in[-1,1]} | P^*_d(x) - f(x) | \leq \epsilon.\]
\end{lemma}

We also need the following polynomial approximations of negative power functions, which are not uniform approximations: 

\begin{lemma}[Polynomial approximation of negative power functions {\cite[Corollary 67 in the full version]{GSLW19}}] \label{lemma:uniformPolyApprox-signedNegativePower}
    Let $c > 0$ and $\delta, \epsilon \in (0, 1/2]$. 
    There is a degree-$d$ odd polynomial $P_d \in \mathbb{R}\sbra{x}$, where $d = O\rbra{\frac{c+1}{\delta}\log\rbra{\frac{1}{\epsilon}}}$, such that 
    \[
    \max_{x \in \sbra{-1, 1} \setminus \rbra{-\delta, \delta}} \abs*{P_d\rbra{x} - \frac{1}{2} \delta^c  \sgn\rbra{x} x^{-c}} \leq \epsilon \quad \textup{and} \quad \max_{x \in \sbra{-1, 1}} \abs*{P_d\rbra{x}} \leq 1.
    \]
\end{lemma}

\subsubsection{Chebyshev expansion and truncations}
\label{subsubsec:Chebyshev-truncation}

We introduce Chebyshev polynomials and an averaged variant of the Chebyshev truncation. 
We recommend \cite[Chapter 3]{Rivlin90} for a comprehensive review of Chebyshev expansion.

\begin{definition}[Chebyshev polynomials]
The Chebyshev polynomials \emph{(}of the first kind\emph{)} $T_k(x)$ are defined via the following recurrence relation: $T_0(x)\coloneqq1$, $T_1(x)\coloneqq x$, and $T_{k+1}(x)\coloneqq2x T_k(x)-T_{k-1}(x)$. For $x \in [-1,1]$, an equivalent definition is $T_k(\cos \theta) = \cos(k \theta)$.
\end{definition}

To use Chebyshev polynomials (of the first kind) for Chebyshev expansion, we need to define an inner product between two functions, $f$ and $g$, as long as the following integral exists:

\begin{equation}
    \label{eq:polynomial-inner-product}
    \innerprodF{f}{g} \coloneqq \frac{2}{\pi} \int_{-1}^1 \frac{f(x)g(x)}{\sqrt{1-x^2}} \dx.
\end{equation}

The Chebyshev polynomials form an orthonormal basis in the inner product space induced by $\innerprodF{\cdot}{\cdot}$ defined in \Cref{eq:polynomial-inner-product}.
Consequently, any continuous and integrable function $f: [-1,1] \to \bbR$ whose Chebyshev coefficients satisfy $\lim_{k \to \infty} c_k=0$, where $c_k$ is defined in \Cref{eq:Chebyshev-expansion}, has a Chebyshev expansion expressed as:
\begin{equation}
    \label{eq:Chebyshev-expansion}
    f(x)=\frac{1}{2} c_0 T_0(x) + \sum_{k=1}^{\infty} c_k T_k(x), \text{ where }  c_k\coloneqq\innerprodF{T_k}{f}.
\end{equation} 

Instead of approximating functions directly via the Chebyshev truncation $\tilde{P}_d = c_0/2 + \sum_{k=1}^d c_k T_k$, we utilize the \textit{de La Vall\'ee Poussin partial sum}, and then obtain the degree-$d$ \textit{averaged Chebyshev truncation} $\hat{P}_{d'}$, which is a polynomial of degree $d'=2d-1$:
\begin{equation}
    \label{eq:averaged-Chebyshev-truncation}
    \hat{P}_{d'}(x) \coloneqq \frac{1}{d} \sum_{l=d}^{d'} \tilde{P}_l(x) 
    = \frac{\hat{c}_0}{2} + \sum_{k=1}^{d'} \hat{c}_k T_k(x) \text{ where } \hat{c}_k = \begin{cases}
        c_k ,& 0 \leq k \leq d\\
        \frac{2d-k}{d} c_k,& k > d 
    \end{cases},
\end{equation}
we can achieve the truncation error $4 \epsilon$ for any function that admits Chebyshev expansion. 
\begin{lemma}[Asymptotically best approximation by averaged Chebyshev truncation, adapted from Exercise 3.4.6 and 3.4.7 in~\cite{Rivlin90}]
    \label{lemma:averaged-Chebyshev-truncation}
    For any function $f$ that has a Chebyshev expansion, consider the degree-$d$ averaged Chebyshev truncation $\hat{P}_{d'}$ defined in \Cref{eq:averaged-Chebyshev-truncation}. 
    Let $\varepsilon_d(f)$ be the truncation error corresponds to the degree-$d$ best uniform approximation on $[-1,1]$ to $f$. If there is a degree-$d$ polynomial $P^*_d\in\bbR[x]$ such that $\max_{x\in[-1,1]} |f(x)-P^*_d(x)| \leq \epsilon$, then
    \[ \max_{x\in[-1,1]} \big| f(x) - \hat{P}_{d'}(x) \big| \leq 4 \varepsilon_d(f) \leq 4 \max_{x\in[-1,1]} |f(x)-P^*_d(x)| \leq 4\epsilon. \]
\end{lemma}

\subsection{Quantum algorithmic toolkit}

In this subsection, we provide several quantum algorithmic tools: the quantum singular value transformation, four useful quantum algorithmic subroutines, and the quantum samplizer, which enables a quantum query-to-sample simulation. 

\subsubsection{Quantum singular value transformation}

We begin by introducing the notion of block-encoding.

\begin{definition} [Block-encoding] \label{def:block-encoding}
    A linear operator $A$ on an $\rbra{n+a}$-qubit Hilbert space is said to be an $\rbra{\alpha, a, \epsilon}$-block-encoding of an $n$-qubit linear operator $B$, if 
    \[
    \Abs{\alpha \rbra{\bra{0}^{\otimes a} \otimes I_n} A \rbra{\ket{0}^{\otimes a} \otimes I_n} - B} \leq \epsilon,
    \]
    where $I_n$ is the $n$-qubit identity operator and $\Abs{\cdot}$ is the operator norm.
\end{definition}

Then, we state the quantum singular value transformation: 
\begin{lemma} [Quantum singular value transformation, {\cite[Theorem 31]{GSLW19}}] \label{lemma:qsvt}
    Let $U$ be a unitary operator that is an $\rbra{\alpha, a, \epsilon}$-block-encoding of an Hermitian operator $A$. Let $P_\bullet \in \mathbb{R}\sbra{x}$ be a polynomial of degree $d$ satisfying $\abs{P_\bullet\rbra{x}} \leq \frac 1 2$ for $x \in \sbra{-1, 1}$.
    Then there exists a quantum circuit $\tilde U$ that is a $\rbra{1, a+2, 4d\sqrt{\epsilon/\alpha} + \delta}$-block-encoding of $P_\bullet\rbra{A/\alpha}$. This circuit $\tilde U$ can be implemented using $O\rbra{d}$ queries to $U$ and $O\rbra{\rbra{a+1}d}$ one- and two-qubit gates. 
    Moreover, the classical description of $\tilde U$ can be computed in deterministic time $\poly\rbra{d, \log\rbra{1/\delta}}$. 
\end{lemma}

\subsubsection{Quantum subroutines}

The first subroutine is the quantum amplitude estimation: 
\begin{lemma} [Quantum amplitude estimation, {\cite[Theorem 12]{BHMT02}}] \label{lemma:qaa}
    Suppose that $U$ is a unitary operator such that 
    \[
    U \ket{0} \ket{0} = \sqrt{p} \ket{0} \ket{\phi_0} + \sqrt{1-p} \ket{1} \ket{\phi_1},
    \]
    where $\ket{\phi_0}$ and $\ket{\phi_1}$ are normalized pure quantum states and $p \in \sbra{0, 1}$. 
    Then, there is a quantum query algorithm using $O\rbra{M}$ queries to $U$ that outputs $\tilde p$ such that 
    \[
    \Pr\sbra*{ \abs*{\tilde p - p} \leq \frac{2\pi\sqrt{p\rbra{1-p}}}{M} + \frac{\pi^2}{M^2} } \geq \frac{8}{\pi^2}.
    \]
    Moreover, if $U$ acts on $n$ qubits, then the quantum query algorithm can be implemented by using $O\rbra{Mn}$ one- and two-qubit gates.
\end{lemma}

The second subroutine prepares a purified density matrix, originally stated in \cite{LC19}: 
\begin{lemma} [Block-encoding of density operators, {\cite[Lemma 25]{GSLW19}}] \label{lemma:block-encoding-of-density-operators}
    Suppose that $U$ is an $\rbra{n+a}$-qubit unitary operator that prepares a purification of an $n$-qubit mixed quantum state $\rho$. 
    Then, we can implement a unitary operator $W$ by using $1$ query to each of $U$ and $U^\dag$ such that $W$ is a $\rbra{1, n+a, 0}$-block-encoding of $\rho$. 
\end{lemma}

The third subroutine is linear-combination-of-unitaries (LCU), originally proposed in~\cite{BCC+15}:

\begin{definition} [State preparation pair]
    The pair of $b$-qubit unitary operators $\rbra{P_L, P_R}$ is said to be a $\rbra{\beta, b, \epsilon}$-state-preparation-pair for a vector $y \in \mathbb{C}^m$ with $\Abs{y}_1 \leq \beta$ and $m \leq 2^b$ if $P_L\ket{0} = \sum_{j=0}^{2^b-1} c_j \ket{j}$ and $P_R\ket{0} = \sum_{j=0}^{2^b-1} d_j \ket{j}$ such that $\sum_{j=0}^{m-1} \abs{\beta c_j^* d_j - y_j} \leq \epsilon$ and for all $m \leq j < 2^b$, it holds that $c_j^*d_j = 0$. 
\end{definition}

\begin{lemma} [Linear combination of block-encoded matrices, {\cite[Lemma 29]{GSLW19}}] \label{lemma:lcu}
    Suppose that for each $0 \leq j < m$, $U_j$ be a unitary operator that is a $\rbra{\alpha, a, \epsilon_2}$-block-encoding of an $s$-qubit operator $A_j$.
    Let $\rbra{P_L, P_R}$ be a $\rbra{\beta, b, \epsilon_1}$-state-preparation-pair for $y \in \mathbb{C}^m$.
    Then, we can implement an $\rbra{s+a+b}$-qubit unitary operator $W$ that is an $\rbra{\alpha\beta, a+b, \alpha\epsilon_1 + \alpha\beta\epsilon_2}$-block-encoding of $\sum_{j=0}^{m-1} y_j A_j$, by using $1$ query to each of controlled-$U_j$, $P_L^\dag$ and $P_R$. 
\end{lemma}

The fourth subroutine is a specific version of one-bit precision phase estimation~\cite{Kitaev95}, often referred to as the Hadamard test~\cite{AJL09}, as stated in~\cite{GP22}:
\begin{lemma} [Hadamard test for block-encodings, adapted from {\cite[Lemma 9]{GP22}}] \label{lemma:hadamard}
    Suppose that unitary operator $U$ is a $\rbra{1, a, 0}$-block-encoding of an $n$-qubit operator $A$. 
    Then, we can implement a quantum circuit that, given an input of an $n$-qubit mixed quantum state $\rho$, outputs $0$ with probability $\frac{1}{2}+\frac{1}{2}\Real\sbra{\tr\rbra{A\rho}}$ (resp., $\frac{1}{2}+\frac{1}{2}\Imag\sbra{\tr\rbra{A\rho}}$), using $1$ query to controlled-$U$ and $O\rbra{1}$ one- and two-qubit gates. 

    Moreover, if an $\rbra{n+a}$-qubit unitary operator $\mathcal{O}$ prepares a purification of $\rho$, then, by combining \cref{lemma:qaa}, we can estimate $\tr\rbra{A\rho}$ to within additive error $\epsilon$ by using $O\rbra{1/\epsilon}$ queries to each of $U$ and $\mathcal{O}$ and $O\rbra{\rbra{n+a}/\epsilon}$ one- and two-qubit gates. 
\end{lemma}

\subsubsection{Samplizer and multi-samplizer}

We now introduce the notion of samplizer in \cite{WZ24b}, which helps us establish the sample complexity upper bound from the query complexity upper bound. 

\begin{definition} [Samplizer]
    A samplizer $\mathsf{Samplize}_*\ave{*}$ is a mapping that converts quantum query algorithms (quantum circuit families with query access to quantum unitary oracles) to quantum sample algorithms (quantum channel families with sample access to quantum states) such that: 
    For any $\delta > 0$, quantum query algorithm $\mathcal{A}^{U}$, and quantum state $\rho$, there exists a unitary operator $U_\rho$ that is a $\rbra{2, a, 0}$-block-encoding of $\rho$ for some $a > 0$, satisfying
    \[
    \Abs*{\mathsf{Samplize}_\delta\ave{\mathcal{A}^U}\sbra{\rho} - \mathcal{A}^{U_{\rho}}}_\diamond \leq \delta,
    \]
    where $\Abs{\cdot}_{\diamond}$ denotes the diamond norm and $\mathcal{E}\sbra{\rho}\rbra{\cdot}$ represents a quantum channel $\mathcal{E}$ with sample access to $\rho$. 
\end{definition}

Then, we include an efficient implementation of the samplizer in \cite{WZ24b}, which is based on quantum principal component analysis \cite{LMR14,KLL+17} and generalizes \cite[Corollary 21]{GP22}. 

\begin{lemma} [Optimal samplizer, {\cite[Theorem III.1]{WZ24b}}] \label{lemma:samplizer}
    There is a samplizer $\mathsf{Samplize}_*\ave{*}$ such that for $\delta > 0$ and quantum query algorithm $\mathcal{A}^{U}$ with query complexity $Q$, the implementation of $\mathsf{Samplize}_\delta\ave{\mathcal{A}^U}\sbra{\rho}$ uses $O\rbra{\frac{Q^2}{\delta}\log^2\rbra{\frac{Q}{\delta}}}$ samples of $\rho$.
    Moreover, if $\rho$ is an $n$-qubit quantum state, then the implementation uses $O\rbra{\frac{Q^2n}{\delta}\log^2\rbra{\frac{Q}{\delta}}}$ additional one- and two-qubit gates. 
\end{lemma}

For our purpose, we need the notion of multi-samplizer, which extends the samplizer. 
The notion of multi-samplizer was implicitly used in \cite{GP22,WZ24,LWWZ24}, and was later used in \cite{WZ24c} to optimally estimate the trace distance and fidelity between pure quantum states. 

\begin{definition} [Multi-samplizer]
    A $k$-samplizer $\mathsf{Samplizer}_{*}\ave{*}$ is a mapping that converts quantum query algorithms to quantum sample algorithms such that: 
    For any $\delta > 0$, $k$ quantum states $\rho_1, \rho_2, \dots, \rho_k$, and a quantum query algorithm $\mathcal{A}^{U_1, U_2, \dots, U_k}$ that makes queries to $k$ quantum unitary oracles $U_1, U_2, \dots, U_k$, 
    there exists $k$ unitary operators $U_{\rho_1}, U_{\rho_2}, \dots, U_{\rho_k}$ that are respectively $\rbra{2, a, 0}$-block-encoding of $\rho_1, \rho_2, \dots, \rho_k$ for some $a > 0$, satisfying
    \[
    \Abs*{ \mathsf{Samplizer}_{\delta}\ave{\mathcal{A}^{U_1, U_2, \dots, U_k}}\sbra{\rho_1, \rho_2, \dots, \rho_k} - \mathcal{A}^{U_{\rho_1}, U_{\rho_2}, \dots, U_{\rho_k}} }_\diamond \leq \delta.
    \]
\end{definition}

Similar to the construction of the multi-samplizer for pure states in \cite{WZ24c}, we can obtain a multi-samplizer for the general case as follows: 

\begin{lemma} [Optimal multi-samplizer, adapted from {\cite[Theorem 2.2]{WZ24c}} and see also {\cite[Theorem 2.23]{BGW25}}] \label{lemma:multi-samplizer}
    There is a $k$-samplizer $\mathsf{Samplize}_*\ave{*}$ such that for $\delta > 0$ and quantum query algorithm $\mathcal{A}^{U_1, U_2, \dots, U_k}$ that uses $Q_j$ queries to $U_j$ for each $j$, the implementation of $\mathsf{Samplize}_{\delta}\ave{\mathcal{A}^{U_1, U_2, \dots, U_k}}\sbra{\rho_1, \rho_2, \dots, \rho_k}$ uses $O\rbra{\frac{Q Q_j}{\delta}\log^2\rbra{\frac{Q}{\delta}}}$ samples of $\rho_j$ for each $j$, where $Q = \sum_{j \in \sbra{k}} Q_j$. 
    Moreover, if $\rho_1, \rho_2, \dots, \rho_k$ are $n$-qubit quantum states, then it uses $O\rbra{\frac{Q^2n}{\delta}\log^2\rbra{\frac{Q}{\delta}}}$ additional one- and two-qubit gates. 
\end{lemma}

It is worth noting that the near-optimality of \cref{lemma:multi-samplizer} is implied by \cite[Theorem 2.3]{WZ24c} for constant $k \geq 1$. 

Recently, another type of samplizer was discovered in \cite{TWZ25} (cf.\ \cite{CWZ25}) with better sample complexity but with worse time complexity. 

\begin{lemma} [Adapted from {\cite[Theorem 1.5]{TWZ25}}] \label{lemma:twz-samplizer}
    Given a quantum algorithm $\mathcal{A}$ that uses $Q$ queries to the state-preparation circuits of $n$-qubit quantum states and $T$ one- and two-qubit gates, there is a quantum algorithm that is $0.01$-close in trace norm distance to $\mathcal{A}$, using $O\rbra{Q^2}$ samples of the quantum states and $T+\widetilde{O}\rbra{Q^8\poly\rbra{n}}$ one- and two-qubit gates. 
\end{lemma}

The number of gates in \cref{lemma:twz-samplizer} is due to the use of the current best implementation of the Schur transform given in \cite{BFG+25}. 

\cref{lemma:twz-samplizer} further leads to a tight quantum sample-to-query lifting (see also \cite{CWZ25}), improving the results in \cite{WZ23}. 

\begin{lemma}[Quantum sample-to-query lifting, cf.\ {\cite[Theorem 1.1]{CWZ25}}] \label{lemma:lifting}
    If testing a property of quantum states requires sample complexity at least $S$ to solve with success probability at least $2/3$, then any quantum algorithm that solves this problem with success probability at least $2/3$ using queries to the state-preparation circuits of the quantum states requires query complexity $\Omega\rbra{\sqrt{S}}$. 
\end{lemma}

\section{Efficient quantum algorithms for estimating quantum Schatten power distance}

In this section, we present efficient quantum algorithms for estimating the quantum Schatten $\alpha$-power distance $\Lalpha(\rho_0,\rho_1)$ for \emph{constant} $\alpha >0$. These algorithms use either queries to state-preparation circuits or samples of the states $\rho_0$ and $\rho_1$, where the quantum sample algorithms are obtained by combining the corresponding quantum query algorithms with the (multi-)samplizer techniques in~\cite{WZ24c,TWZ25}:
\begin{itemize}
    \item In \Cref{subsec:algo-Schatten-alpha>1}, we focus on the regime $\alpha>1$, where the quantum Schatten $\alpha$-power distance is a powered version of the Schatten $\alpha$-norm distance up to a constant multiplicative factor depending on $\alpha$. Remarkably, our quantum query and sample algorithms for estimating $\Lalpha(\rho_0,\rho_1)$, and thus $\Talpha(\rho_0,\rho_1)$, are \emph{rank-independent}. Consequently, we establish a \BQP{} containment of the promise problem \QSDalpha{} introduced in \Cref{sec:easy-regime-hardness}. 
    \item In \Cref{subsec:algo-Schatten-0<alpha<1}, we consider the regime $0<\alpha<1$ and provide \emph{rank-efficient} quantum query and sample algorithms for estimating the quantum Schatten $\alpha$-power distance. 
\end{itemize}

The core of our approach for the easy regime $\alpha>1$ is an \textit{efficient} uniform approximation to \textit{signed} positive constant-power functions, developed in \Cref{subsec:efficient-uniform-approx-signed-pos-powers}. This approximation provides a uniform error bound over the entire interval $[-1,1]$ and is of independent interest.

\subsection{Efficient uniform approximations of signed positive powers}
\label{subsec:efficient-uniform-approx-signed-pos-powers}

Leveraging the averaged Chebyshev truncation specified in \Cref{subsubsec:Chebyshev-truncation}, we provide an \textit{efficiently computable} uniform polynomial approximation of \textit{signed} positive constant powers: 

\begin{lemma}[Efficient uniform polynomial approximation of signed positive powers]
    \label{lemma:computable-uniformPolyApprox-signedPositivePower}
    Let $\alpha$ be a positive real constant number. For any $\epsilon \in (0,1/2)$, there is a degree-$d$ polynomial $P_d \in \bbR[x]$, where $d = \ceil*{\rbra*{\beta_{\alpha}'/\epsilon}^{1/\alpha}}$ and $\beta'_{\alpha}$ is a constant depending on $\alpha$, that can be deterministically computed in $\widetilde{O}(d)$ time. For sufficiently small $\epsilon$, it holds that:
    \[ \max_{x \in [-1,1]} \abs*{ \frac{1}{2} \sgn(x) |x|^{\alpha} - P_d(x) }\leq \epsilon \quad\text{and}\quad \max_{x \in [-1,1]} \abs*{P_d(x)} \leq 1.\]
\end{lemma}

\begin{proof}
    Let $f(x) \coloneqq \frac{1}{2} \sgn(x) |x|^{\alpha}$. 
    For any $\tilde{\epsilon} \in (0,1/8)$, using the uniform approximation of signed positive powers (\Cref{lemma:bestUniformApprox-signedPositivePower}), we obtain the degree-$\tilde{d}$ best uniform polynomial approximation $P_{\tilde{d}}^*(x)$, where $\tilde{d} = \ceil*{\rbra*{\beta_{\alpha}/\tilde{\epsilon}}^{1/\alpha}}$ and $\beta_{\alpha}$ is a constant depending on $\alpha$, such that
    \begin{equation}
        \label{eq:signedPositivePower-errorBound}
        \max_{x\in [-1,1]} \abs*{\frac{1}{2} \sgn(x) |x|^{\alpha} - P_{\tilde{d}}^*(x)} \leq \tilde{\epsilon} \quad\text{and}\quad \max_{x \in [-1,1]} \abs*{P^*_{\tilde{d}}(x)} \leq \frac{1}{2}+\tilde{\epsilon}. 
    \end{equation}

    Next, we consider the degree-$\tilde{d}$ averaged Chebyshev truncation (\Cref{eq:averaged-Chebyshev-truncation}) of $f(x)$. In particular, let $d \coloneqq 2\tilde{d}-1 = \ceil*{\rbra*{\beta_{\alpha}'/\epsilon}^{1/\alpha}}$, where $\beta'_{\alpha}$ is another constant depending on $\alpha$ and $\epsilon$ will be specified later. We obtain the following degree-$d$ polynomial: 
    \begin{equation}
        \label{eq:fullrange-positivePower-poly}
        P_d(x) = \frac{\hat{c}_0}{2} + \sum_{k=1}^{d} \hat{c}_k T_k(x), \text{ where } \hat{c}_k \coloneqq \begin{cases}
        c_k,& 0\leq k \leq \tilde{d}\\
        \frac{2\tilde{d}-k}{\tilde{d}} c_k,& k > \tilde{d}
        \end{cases} \quad\text{and}\quad c_k \coloneqq \innerprodF{T_k}{f}.  
    \end{equation}

    Using the asymptotically best uniform approximation by averaged Chebyshev truncation (\Cref{lemma:averaged-Chebyshev-truncation}) and \Cref{eq:signedPositivePower-errorBound}, we can derive that $P_d(x)$ satisfies the following: 
    \[ \max_{x\in [-1,1]} \abs*{\frac{1}{2} \sgn(x) |x|^{\alpha} - P_{d}(x)} \leq 4\tilde{\epsilon} \coloneqq \epsilon \quad\text{and}\quad \max_{x \in [-1,1]} \abs*{P_{d}(x)} \leq \frac{1}{2}+4\tilde{\epsilon} = \frac{1}{2}+\epsilon < 1.  \]

    It is left to prove that $P_d(x)$ can be computed in deterministic time $\widetilde{O}(d)$. As $f(x)$ is an odd function, a direct calculation implies that the Chebyshev coefficient $\{c_k\}_{0 \leq k \leq d}$ in \Cref{eq:fullrange-positivePower-poly} satisfy $c_k=0$ for all even $k$, and the following equalities hold for odd $k$: 
    \begin{align*}
        c_{2l+1} = c_{2l-1} \cdot \frac{\alpha-2l+1}{\alpha+2l+1} \quad\text{and}\quad
        c_{1} = \frac{2}{\pi} \int_{-1}^{1} \frac{\frac{1}{2} \sgn(x) |x|^{\alpha} \cdot T_{1}(x)}{\sqrt{1-x^2}} \dx =  \frac{2}{\sqrt{\pi}} \cdot \frac{\Gamma\rbra*{\frac{1}{2}(\alpha+2)}}{\Gamma\rbra*{\frac{1}{2}(\alpha+3)}}.
    \end{align*}
    Here, the Gamma function $\Gamma(x) \coloneqq \int_{0}^{\infty} t^{x-1} e^{-t} \dd t$ for any $x > 0$. 

    Therefore, the averaged Chebyshev coefficient $\{\hat{c}_k\}_{0 \leq k \leq d}$ can be recursively computed in deterministic time $\widetilde{O}(d)$. 
    We finish the proof by observing that the Chebyshev polynomials $\cbra*{T_k(x)}_{0\leq k \leq d}$ can also be recursively computed in deterministic time $\widetilde{O}(d)$.
\end{proof}

\subsection{Quantum Schatten \texorpdfstring{$\alpha$}{}-power distance estimation for \texorpdfstring{$\alpha>1$}{}}
\label{subsec:algo-Schatten-alpha>1}

\subsubsection{Rank-independent quantum query algorithm}

We now provide efficient quantum query algorithms for estimating $\Lalpha(\rho_0,\rho_1)$ and $\Talpha(\rho_0,\rho_1)$: 

\begin{lemma}[Quantum Schatten $\alpha$-power distance estimation via queries]
    \label{lemma:powered-Lalpha-query-algo}
    Suppose that $Q_0$ and $Q_1$ are unitary operators that prepare purifications of mixed quantum states $\rho_0$ and $\rho_1$, respectively. 
    For every fixed real constant $\alpha>1$, there is a quantum query algorithm that estimates $\Lalpha(\rho_0,\rho_1)$ to within additive error $\epsilon$ using $O\rbra{1/{\epsilon^{1+\frac{1}{\alpha-1}}}}$ queries to $Q_0$ and $Q_1$.
\end{lemma}

By combining \Cref{prop:Talpha-vs-powered-Talpha} with \Cref{lemma:powered-Lalpha-query-algo} for additive error $\epsilon^{\alpha}$, we obtain a quantum query algorithm for estimating $\Talpha(\rho_0,\rho_1)$ when $\alpha>1$ is a fixed real constant: 

\begin{theorem}[Quantum Schatten $\alpha$-norm distance estimation via queries]
    \label{thm:Lalpha-query-algo}
    Suppose that $Q_0$ and $Q_1$ are unitary operators that prepare purifications of mixed quantum states $\rho_0$ and $\rho_1$, respectively. 
    For every fixed real constant $\alpha>1$, there is a quantum query algorithm that estimates $\Talpha(\rho_0,\rho_1)$ to within additive error $\epsilon$ by using $O\rbra{1/{\epsilon^{\alpha+1+\frac{1}{\alpha-1}}}}$ queries to $Q_0$ and $Q_1$.
\end{theorem}

\begin{algorithm}[!ht]
    \caption{Quantum algorithm for estimating $\Lalpha(\rho_0,\rho_1)$ for constant $\alpha > 1$ (query access).}
    \label{algo:q-Tsallis-estimation-query}
    \begin{algorithmic}[1]
        \Require $\rbra{n+a}$-qubit quantum unitary oracles $Q_0$ and $Q_1$ that prepare purifications of $n$-qubit mixed quantum states $\rho_0$ and $\rho_1$, respectively; $\alpha>1$; $\epsilon_p, \delta_p, \epsilon_H \in \rbra{0, 1}$.
        \Ensure 
        An estimate of $\Lalpha(\rho_0,\rho_1)$ with high probability. 

        \State Implement a unitary operator $U_\nu$ that is a $\rbra{2, n+a+1, 0}$-block-encoding of $\nu \coloneqq \rho_0 - \rho_1$, using $O\rbra{1}$ queries to $Q_0$ and $Q_1$ (by \cref{lemma:block-encoding-of-density-operators,lemma:lcu}). 

        \State Let $P\rbra{x}$ be a polynomial of degree $d = O\rbra{{1}/{\epsilon_p^{\frac{1}{\alpha-1}}}}$ such that $\max_{x \in \sbra{0, 1}} \abs{P\rbra{x} - \frac{1}{2} \sgn\rbra{x} \abs{x}^{\alpha-1}} \leq \epsilon_p$ and $\max_{x \in \sbra{-1, 1}} \abs{P\rbra{x}} \leq 1$ (by \cref{lemma:computable-uniformPolyApprox-signedPositivePower}).

        \State Implement a unitary operator $U_{P\rbra{\nu}}$ that is a $\rbra{1, n+a+3, \delta_p}$-block-encoding of $\frac{1}{2}P\rbra{\nu/2}$, using $O\rbra{d}$ queries to $U_{\nu}$ (by \cref{lemma:qsvt}).

        \State Let $\tilde x_0$ be an $\epsilon_H$-estimate of $\tr\rbra{ \bra{\bar 0} U_{P\rbra{\nu}} \ket{\bar 0} \rho_0}$, where $\ket{\bar 0} = \ket{0}^{\otimes\rbra{n+a+3}}$, using $O\rbra{1/\epsilon_H}$ queries to $U_{P\rbra{\nu}}$ and $Q_0$ (by \cref{lemma:hadamard}).

        \State Let $\tilde x_1$ be an $\epsilon_H$-estimate of $\tr\rbra{ \bra{\bar 0} U_{P\rbra{\nu}} \ket{\bar 0} \rho_1}$, where $\ket{\bar 0} = \ket{0}^{\otimes\rbra{n+a+3}}$, using $O\rbra{1/\epsilon_H}$ queries to $U_{P\rbra{\nu}}$ and $Q_1$ (by \cref{lemma:hadamard}).

        \State \Return $2^{\alpha}\rbra{\tilde x_0 - \tilde x_1}$.
    \end{algorithmic}
\end{algorithm}

\begin{proof}[Proof of \Cref{lemma:powered-Lalpha-query-algo}]
    Our approach extends the equality for $\alpha=1$~\cite{WZ24} to broader settings, focusing on a fixed real constant $\alpha>1$.
    Specifically, we consider $ 2 \Lalpha(\rho_0,\rho_1) 
    = \Abs{\rho_0 - \rho_1}_{\alpha}^{\alpha} $: 
    \begin{equation}
        \label{eq:poweredLalpha-decomposition}
        \Abs{\rho_0 - \rho_1}_{\alpha}^{\alpha} 
        = \tr\rbra*{\rho_0 \cdot \sgn\rbra{\rho_0 - \rho_1} \cdot \abs*{\rho_0 - \rho_1}^{\alpha-1}} - \tr\rbra*{\rho_1 \cdot \sgn\rbra{\rho_0 - \rho_1} \cdot \abs*{\rho_0 - \rho_1}^{\alpha-1}}.
    \end{equation}
    Consequently, the task reduces to separately estimating the terms $\tr\rbra{\rho_0 \cdot \sgn\rbra{\rho_0 - \rho_1} \cdot \abs{\rho_0 - \rho_1}^{\alpha-1}}$ and $\tr\rbra{\rho_1 \cdot \sgn\rbra{\rho_0 - \rho_1} \cdot \abs{\rho_0 - \rho_1}^{\alpha-1}}$. 
    
    Suppose that $Q_0$ and $Q_1$ are $\rbra{n+a}$-qubit unitary operators that prepare purifications of $\rho_0$ and $\rho_1$, respectively. 
    We present our quantum algorithm as follows with its formal description given in \cref{algo:q-Tsallis-estimation-query}.

    \paragraph{Step 1: Construct a block-encoding of $\nu \coloneqq \rho_0 - \rho_1$.} This is a standard step achieved by block-encoding density operators and LCU. 
    By \cref{lemma:block-encoding-of-density-operators}, for $b \in \cbra{0, 1}$, we can implement a unitary operator $U_{\rho_b}$ that is $\rbra{1, n+a, 0}$-block-encoding of $\rho_b$, respectively, by using $O\rbra{1}$ queries to $Q_b$. 
    Note that $\rbra{HX, H}$ is a $\rbra{2, 1, 0}$-state-preparation-pair for $y = \rbra{1, -1}$, where $H$ is the Hadamard gate and $X$ is the Pauli-X gate. 
    By \cref{lemma:lcu}, we can implement a unitary operator $U_{\nu}$ that is a $\rbra{2, n+a+1, 0}$-block-encoding of $\nu \coloneqq \rho_0 - \rho_1$. 

    \paragraph{Step 2: Construct a block-encoding of $\sgn\rbra{\nu} \cdot \abs{\nu}^{\alpha-1}$.}
    Let $\epsilon_p, \delta_p \in \rbra{0, 1/2}$ be parameters to be determined. 
    By \cref{lemma:computable-uniformPolyApprox-signedPositivePower}, there is a polynomial $P \in \mathbb{R}\sbra{x}$ of degree $d = O\rbra{1/\epsilon_p^{\frac{1}{\alpha-1}}}$ such that $\max_{x \in \sbra{-1, 1}} \abs{P\rbra{x} - \frac{1}{2} \sgn\rbra{x} \abs{x}^{\alpha-1}} \leq \epsilon_p$ and $\max_{x \in \sbra{-1, 1}} \abs{P\rbra{x}} \leq 1$.
    By \cref{lemma:qsvt} with $P_\bullet \coloneqq \frac{1}{2}P$, $\alpha \coloneqq 1$, $a \coloneqq n+a+1$, $\epsilon \coloneqq 0$ and $d \coloneqq O\rbra{{1}/{\epsilon_p^{\frac{1}{\alpha-1}}}}$, we can implement a quantum circuit $U_{P\rbra{\nu}}$ that is a $\rbra{1, n+a+3, \delta_p}$-block-encoding of $\frac{1}{2}P\rbra{\nu/2}$, by using $O\rbra{{1}/{\epsilon_p^{\frac{1}{\alpha-1}}}}$ queries to $U_{\nu}$. 
    Moreover, the classical description of $U_{P\rbra{\nu}}$ can be computed in deterministic time $\poly\rbra{1/\epsilon_p, \log\rbra{1/\delta_p}}$.

    \paragraph{Step 3: Estimate $\tr\rbra{P\rbra{\nu} \rho_0}$ and $\tr\rbra{P\rbra{\nu} \rho_1}$.}
    Suppose that $U_{P\rbra{\nu}}$ is a $\rbra{1, n+a+3, 0}$-block-encoding of $A$ with $\Abs{A - \frac{1}{2}P\rbra{\nu/2}} \leq \delta_p$. 
    By \cref{lemma:hadamard}, we can obtain an estimate $\tilde x_0$ of $\tr\rbra{A\rho_0}$ to within additive error $\epsilon_H$ by using $O\rbra{1/\epsilon_H}$ queries to $U_{P\rbra{\nu}}$ and $Q_0$ such that
    \begin{equation}\label{eq:lp-had0}
    \Pr\sbra[\big]{ \abs*{\tilde x_0 - \tr\rbra*{A \rho_0}} \leq \epsilon_H } \geq 0.9.
    \end{equation}
    Similarly, we can obtain an estimate $\tilde x_1$ of $\tr\rbra{A\rho_1}$ to within additive error $\epsilon_H$ by using $O\rbra{1/\epsilon_H}$ queries to $U_{P\rbra{\nu}}$ and $Q_1$ such that
    \begin{equation}\label{eq:lp-had1}
    \Pr\sbra[\big]{ \abs*{\tilde x_1 - \tr\rbra*{A \rho_1}} \leq \epsilon_H } \geq 0.9.
    \end{equation}

    \paragraph{Step 4: Return $2^{\alpha}\rbra{\tilde x_0 - \tilde x_1}$ as the estimate.}

    \paragraph{Error analysis and parameter selection.}
    The overall process from Step 1 to Step 3 uses 
    \[
    O\rbra*{\frac{1}{\epsilon_p^{\frac{1}{\alpha-1}}}} \cdot O\rbra*{\frac{1}{\epsilon_H}} = O\rbra*{\frac{1}{\epsilon_p^{\frac{1}{\alpha-1}}\epsilon_H}}
    \]
    queries to $Q_0$ and $Q_1$, and 
    \[
    O\rbra*{\frac{n+a}{\epsilon_p^{\frac{1}{\alpha-1}}\epsilon_H}}
    \]
    one- and two-qubit gates. 
    Moreover, the classical description of this quantum circuit can be computed in deterministic time $\poly\rbra{1/\epsilon_p, 1/\epsilon_H, \log\rbra{1/\delta_p}}$.
    To determine the parameters, we note that 
    \begin{equation} \label{eq:succ-prob-query}
        \Pr\sbra*{ \abs*{ 2^{\alpha}\rbra*{\tilde x_0 - \tilde x_1} - \Lalpha(\rho_0,\rho_1) } \leq 2^{\alpha}\rbra*{2\epsilon_H + 2\delta_p + \epsilon_p} } \geq 0.8.
    \end{equation}
    Then, by setting $\epsilon_H = \delta_p = \epsilon_p = 2^{-\alpha-4}\epsilon$, we can estimate $\Lalpha(\rho_0,\rho_1)$ to within additive error $\epsilon$ with success probability at least $0.8$, using 
    \[
    O\rbra*{\frac{1}{\epsilon^{1+\frac{1}{\alpha-1}}}}
    \]
    queries to $Q_0$ and $Q_1$ and 
    \[
    O\rbra*{\frac{n+a}{\epsilon^{1+\frac{1}{\alpha-1}}}}
    \]
    one- and two-qubit gates. 
    Moreover, the classical description of this quantum circuit can be computed in deterministic time $\poly\rbra{1/\epsilon}$.

    To complete the proof, it remains to show \cref{eq:succ-prob-query}. 
    Using the matrix H\"{o}lder inequality, e.g.,~\cite[Theorem 2]{Bau11}, it follows that:
    \begin{equation}\label{eq:lp-delta}
    \abs*{\tr\rbra*{\frac{1}{2}P\rbra*{\frac{\nu}{2}}\rho_0} - \tr\rbra{A\rho_0}} \leq \Abs*{\frac{1}{2}P\rbra*{\frac{\nu}{2}} - A} \leq \delta_p.
    \end{equation}
    Also, the following inequality holds: 
    \begin{equation}\label{eq:lp-eps}
    \abs*{ \tr\rbra*{\frac{1}{2}P\rbra*{\frac{\nu}{2}} \nu} - \tr\rbra*{\frac{1}{4} \sgn\rbra*{\frac{\nu}{2}} \abs*{\frac{\nu}{2}}^{\alpha-1} \nu} } \leq \epsilon_p.
    \end{equation}
    To see \cref{eq:lp-eps}, suppose that $\nu = \rho_0 - \rho_1 = \sum_{j} \lambda_j \ketbra{\psi_j}{\psi_j}$ is the spectral decomposition with $\sum_j \abs{\lambda_j} \leq 2$ and $\abs{\lambda_j} \leq 1$ for all $j$. 
    Then, 
    \begin{align*}
        \abs*{ \tr\rbra*{\frac{1}{2}P\rbra*{\frac{\nu}{2}} \nu} - \tr\rbra*{\frac{1}{4} \sgn\rbra*{\frac{\nu}{2}} \abs*{\frac{\nu}{2}}^{\alpha-1} \nu} }
        & \leq \sum_j \abs*{ \frac{1}{2}P\rbra*{\frac{\lambda_j}{2}} \lambda_j - \frac{1}{4} \sgn\rbra*{\frac{\lambda_j}{2}} \abs*{\frac{\lambda_j}{2}}^{\alpha-1} \lambda_j } \\
        & = \sum_j \frac{1}{2} \abs*{\lambda_j} \abs*{P\rbra*{\frac{\lambda_j}{2}} - \frac{1}{2} \sgn\rbra*{\frac{\lambda_j}{2}} \abs*{\frac{\lambda_j}{2}}^{\alpha-1}} \\
        & \leq \sum_j \frac{1}{2} \abs*{\lambda_j} \epsilon_p \\
        & \leq \epsilon_p.
    \end{align*}
    Then, \cref{eq:succ-prob-query} can be obtained by combining \cref{eq:lp-had0,eq:lp-had1,eq:lp-delta,eq:lp-eps}.
\end{proof}

\subsubsection{Rank-independent quantum sample algorithm}

We proceed by describing efficient quantum sample algorithms for $\Lalpha(\rho_0,\rho_1)$ and $\Talpha(\rho_0,\rho_1)$:

\begin{algorithm}[!ht]
    \caption{Quantum algorithm for estimating $\Lalpha(\rho_0,\rho_1)$ for constant $\alpha > 1$ (sample access).}
    \label{algo:q-Tsallis-estimation-sample}
    \begin{algorithmic}[1]
        \Require Independent and identical samples of $n$-qubit mixed quantum states $\rho_0$ and $\rho_1$; $\alpha>1$; $\delta, \epsilon_p, \delta_p, \epsilon_H \in \rbra{0, 1}$.
        \Ensure 
        An estimate of $\Lalpha(\rho_0,\rho_1)$ with high probability. 

        \begin{tcolorbox}[colback=gray!12]
        \Function{ApproxDiffPower}{$\alpha, \epsilon_p, \delta_p$}${}^{U_1, U_2}$
        \renewcommand{\algorithmicrequire}{\qquad \textbf{Input:}}
        \renewcommand{\algorithmicensure}{\qquad \textbf{Output:}} 
        \Require Unitary $\rbra{1, a, 0}$-block-encodings $U_1$ and $U_2$ of $A$ and $B$, respectively, and parameters $\alpha >1, \epsilon_p, \delta_p \in \rbra{0, 1/2}$. 
        \Ensure A unitary operator.

        \State Let $U_3$ be a $\rbra{2, a+1, 0}$-block-encoding of $A-B$ by using $O\rbra{1}$ queries to $U_1$ and $U_2$ (by \cref{lemma:lcu}). 
        \State Let $P\rbra{x}$ be a polynomial of degree $d = O\rbra{{1}/{\epsilon_p^{\frac{1}{\alpha-1}}}}$ such that $\max_{x \in \sbra{0, 1}} \abs{P\rbra{x} - \frac{1}{2} \sgn\rbra{x} \abs{x}^{\alpha-1}} \leq \epsilon_p$ and $\max_{x \in \sbra{-1, 1}} \abs{P\rbra{x}} \leq 1$ (by \cref{lemma:computable-uniformPolyApprox-signedPositivePower}).
        \State Construct a unitary $\rbra{1, a+3, \delta_p}$-block-encoding $U_4$ of $\frac{1}{2}P\rbra{\frac{A-B}{2}}$ by using $O\rbra{d}$ queries to $U_3$ (by \cref{lemma:qsvt}). 
        \State \Return $U_4$.
        \EndFunction
        \end{tcolorbox}

        \For {$i \in \binset$}
        \State $k \gets \Theta\rbra{1/\epsilon_H^2}$.
        \For {$j \in \cbra{1, 2, \dots, k}$}
            \State Let $b_{i,j}' \in \binset$ be the outcome of the Hadamard test (by \cref{lemma:hadamard}) performed on $\rho_i$ and $\mathsf{Samplize}_{\delta}\ave{\begin{tcolorbox}[nobeforeafter,
  after={\xspace},
  hbox,
  tcbox raise base,
  fontupper=\ttfamily,
  colback=gray!12,
  size=fbox] $\texttt{ApproxDiffPower}\rbra{\alpha, \epsilon_p, \delta_p}^{U_1, U_2}$\end{tcolorbox}}\sbra{\rho_0, \rho_1}$ (as if it were unitary).
        \EndFor
        \State $p_i \gets \frac{1}{k} \sum_{j=1}^k b_{i,j}'$.
        \EndFor

        \State \Return $2^{2\alpha-1} \rbra{p_0 - p_1}$.
    \end{algorithmic}
\end{algorithm}

\begin{lemma}[Quantum Schatten $\alpha$-power distance estimation via samples]
    \label{lemma:powered-Lalpha-sample-algo}
    For every fixed real constant $\alpha > 1$, $\Lalpha(\rho_0,\rho_1)$ can be estimated to within additive error $\epsilon$ on a quantum computer by using 
    \[
    O\rbra*{\frac{1}{\epsilon^{3+\frac{2}{\alpha-1}}}\log^2\rbra*{\frac{1}{\epsilon}}}
    \] 
    samples of $n$-qubit quantum states $\rho_0$ and $\rho_1$, and
    \[
    O\rbra*{\frac{n}{\epsilon^{3+\frac{2}{\alpha-1}}}\log^2\rbra*{\frac{1}{\epsilon}}}
    \] 
    one- and two-qubit gates.
\end{lemma}

By combining \Cref{prop:Talpha-vs-powered-Talpha} with \Cref{lemma:powered-Lalpha-sample-algo} for additive error $\epsilon^{\alpha}$, we obtain a quantum sample algorithm for estimating $\Talpha(\rho_0,\rho_1)$ when $\alpha>1$ is a fixed real constant:

\begin{theorem}[Quantum Schatten $\alpha$-norm distance estimation via samples]
    \label{thm:Lalpha-sample-algo}
    For every fixed real constant $\alpha>1$, there is a quantum sample algorithm that estimates the quantum Schatten $\alpha$-norm distance $\Talpha(\rho_0,\rho_1)$ to within additive error $\epsilon$ by using 
    \[
    O\rbra*{\frac{1}{\epsilon^{3\alpha + 2 +\frac{2}{\alpha-1}}}\log^2\rbra*{\frac{1}{\epsilon}}}
    \]
    samples of $n$-qubit quantum states $\rho_0$ and $\rho_1$, and
    \[
    O\rbra*{\frac{n}{\epsilon^{3\alpha + 2 +\frac{2}{\alpha-1}}}\log^2\rbra*{\frac{1}{\epsilon}}}
    \]
    one- and two-qubit gates.
\end{theorem}

\begin{proof}[Proof of \Cref{lemma:powered-Lalpha-sample-algo}]
    To estimate $\Abs{\rho_0-\rho_1}_{\alpha}^{\alpha} = 2\Lalpha(\rho_0,\rho_1)$, our approach estimates the terms $\tr\rbra{\rho_0 \cdot \sgn\rbra{\rho_0 - \rho_1} \cdot \abs{\rho_0 - \rho_1}^{\alpha-1}}$ and $\tr\rbra{\rho_1 \cdot \sgn\rbra{\rho_0 - \rho_1} \cdot \abs{\rho_0 - \rho_1}^{\alpha-1}}$ in \Cref{eq:poweredLalpha-decomposition} using samples of $\rho_0$ and $\rho_1$. Specifically, our quantum sample algorithm extends the quantum query algorithm in \Cref{lemma:powered-Lalpha-query-algo} via the (multi-)samplizer~\cite{WZ24b,WZ24c}.
    We present our quantum algorithm as follows with its formal description given in \cref{algo:q-Tsallis-estimation-sample}.
    
    \paragraph{Step 1: Construct a block-encoding of $A - B$ given block-encodings of $A$ and $B$.}
    Let $U_1$ be a $\rbra{1, a, 0}$-block-encoding of $A$ and $U_2$ be a $\rbra{1, a, 0}$-block-encoding of $B$ for some integer $a > 0$. 
    This is a standard step achieved by LCU. 
    Note that $\rbra{HX, H}$ is a $\rbra{2, 1, 0}$-state-preparation-pair for $y = \rbra{1, -1}$, where $H$ is the Hadamard gate and $X$ is the Pauli-X gate. 
    By the LCU lemma (\cref{lemma:lcu}), we can implement a unitary operator $U_{3}$ that is a $\rbra{2, a+1, 0}$-block-encoding of $A - B$. 

    \paragraph{Step 2: Construct a block-encoding of $\sgn\rbra{A-B} \cdot \abs{A - B}^{\alpha-1}$.}
    Let $\epsilon_p, \delta_p \in \rbra{0, 1/2}$ be parameters to be determined. 
    By \cref{lemma:computable-uniformPolyApprox-signedPositivePower}, there is a polynomial $P \in \mathbb{R}\sbra{x}$ of degree $d = O\rbra{1/\epsilon_p^{\frac{1}{\alpha-1}}}$ such that $\max_{x \in \sbra{-1, 1}} \abs{P\rbra{x} - \frac{1}{2} \sgn\rbra{x} \abs{x}^{\alpha-1}} \leq \epsilon_p$ and $\max_{x \in \sbra{-1, 1}} \abs{P\rbra{x}} \leq 1$.
    By \cref{lemma:qsvt} with $P_\bullet \coloneqq \frac{1}{2}P$, $\alpha \coloneqq 1$, $a \coloneqq a+1$, $\epsilon \coloneqq 0$ and $d \coloneqq O\rbra{{1}/{\epsilon_p^{\frac{1}{\alpha-1}}}}$, we can implement a quantum circuit $U_{4}$ that is a $\rbra{1, a+3, \delta_p}$-block-encoding of $\frac{1}{2}P\rbra{\frac{A-B}{2}}$, by using $O\rbra{{1}/{\epsilon_p^{\frac{1}{\alpha-1}}}}$ queries to $U_{3}$. 
    Moreover, the classical description of $U_{4}$ can be computed in deterministic time $\poly\rbra{1/\epsilon_p, \log\rbra{1/\delta_p}}$.

    Combining Steps 1 and 2, let $\texttt{ApproxDiffPower}\rbra{\alpha, \epsilon_p, \delta_p}^{U_1,U_2}$ be the implementation of $U_4$, which uses $O\rbra{1/\epsilon_p^{\frac{1}{\alpha-1}}}$ queries to $U_1$ and $U_2$. 

    \paragraph{Step 3: Estimate $\tr\rbra{\rho_i P\rbra{\rho_0-\rho_1}}$.}
    For our purpose, we first consider the case that $A = \rho_0/2$ and $B = \rho_1/2$. 
    In this case, $U_4$ is a $\rbra{1, a+3, \delta_p}$-block-encoding of $\frac{1}{2}P\rbra{\frac{\rho_0-\rho_1}{4}}$.
    For convenience, suppose that $U_4$ is a $\rbra{1, a+3, 0}$-block-encoding of some operator $D$, where $D$ satisfies 
    \begin{equation} \label{eq:diff-op}
    \Abs*{D - \frac{1}{2}P\rbra*{\frac{\rho_0-\rho_1}{4}}} \leq \delta_p. 
    \end{equation}
    For $i \in \binset$, if we perform the Hadamard test (\cref{lemma:hadamard}) on $\rho_i$ and $U_4$, then an outcome $b_i \in \cbra{0, 1}$ will be obtained with
    \[
    \Pr\sbra{b_i = 0} = \frac{1}{2} + \frac{1}{2} \Real\sbra{\tr\rbra{D \rho_i}}. 
    \]
    Let $\delta > 0$ be a parameter to be determined.
    Then using the multi-samplizer (\cref{lemma:multi-samplizer}), we can approximately implement $U_4$ by $\mathsf{Samplize}_{\delta}\ave{\texttt{ApproxDiffPower}\rbra{\alpha, \epsilon_p, \delta_p}^{U_1, U_2}}\sbra{\rho_0, \rho_1}$.
    Let $b_i' \in \cbra{0, 1}$ be the outcome of \cref{lemma:hadamard} on $\rho_i$ and $\mathsf{Samplize}_{\delta}\ave{\texttt{ApproxDiffPower}\rbra{\alpha, \epsilon_p, \delta_p}^{U_1, U_2}}\sbra{\rho_0, \rho_1}$ (as if it were unitary), then 
    \[
    \abs{\Pr\sbra{b_i = 0} - \Pr\sbra{b'_i = 0}} \leq \delta. 
    \]
    By Hoeffding's inequality, we can obtain an estimate $p_i$ of $\Real\sbra{\tr\rbra{D \rho_i}}$ to additive error $\epsilon_H$ with success probability $\geq 0.9$ by $k = O\rbra{1/\epsilon_H^2}$ repetitions of the Hadamard test, i.e., 
    \begin{equation} \label{eq:sampl-error}
    \Pr\sbra[\big]{\abs*{p_i - \Real\sbra{\tr\rbra{D \rho_i}}} \leq \delta+\epsilon_H} \geq 0.9.
    \end{equation}

    \paragraph{Step 4: Return $2^{2\alpha-1}\rbra{p_0 - p_1}$ as the estimate.}

    \paragraph{Error analysis and parameter selection.}

    By \cref{lemma:multi-samplizer}, the overall process from Step 1 to Step 3 uses 
    \[
    O\rbra*{\rbra*{\frac{1}{\epsilon_p^{\frac{1}{\alpha-1}}}}^2\frac{1}{\delta}\log^2\rbra*{\frac{1}{\epsilon_p^{\frac{1}{\alpha-1}}\delta}}} \cdot O\rbra*{\frac{1}{\epsilon_H^2}} = O\rbra*{\frac{1}{\epsilon_p^{\frac{2}{\alpha-1}}\epsilon_H^2\delta} \log^2\rbra*{\frac{1}{\epsilon_p\delta}}}
    \]
    samples of $\rho_0$ and $\rho_1$, and 
    \[
    O\rbra*{\frac{n}{\epsilon_p^{\frac{2}{\alpha-1}}\epsilon_H^2\delta} \log^2\rbra*{\frac{1}{\epsilon_p\delta}}}
    \]
    one- and two-qubit gates. 
    Moreover, the classical description of this quantum circuit can be computed in deterministic time $\poly\rbra{1/\epsilon_p, 1/\epsilon_H, \log\rbra{1/\delta_p}}$.
    To determine the parameters, we note that 
    \begin{equation} \label{eq:final-error}
    \Pr\sbra*{ \abs*{ 4^{\alpha} \rbra{p_0 - p_1} - \Abs{\rho_0 - \rho_1}_{\alpha}^{\alpha} } \leq 4^{\alpha}\rbra*{\epsilon_p + 2\delta_p + 2\delta + 2\epsilon_H} } \geq 0.8.
    \end{equation}
    Then, by setting $\epsilon_p = \delta_p = \delta = \epsilon_H = 4^{-\alpha-3} \epsilon$, we can estimate $\Lalpha(\rho_0,\rho_1)$ to within additive error $\epsilon$ with success probability at least $0.8$, using 
    \[
    O\rbra*{\frac{1}{\epsilon^{3+\frac{2}{\alpha-1}}}\log^2\rbra*{\frac{1}{\epsilon}}}
    \]
    samples of $\rho_0$ and $\rho_1$ and 
    \[
    O\rbra*{\frac{n}{\epsilon^{3+\frac{2}{\alpha-1}}}\log^2\rbra*{\frac{1}{\epsilon}}}
    \] 
    one- and two-qubit gates. 
    Moreover, the classical description of this quantum circuit can be computed in deterministic time $\poly\rbra{1/\epsilon}$.
    
    To complete the proof, it remains to show \cref{eq:final-error}.
    To see this, by \cref{eq:diff-op}, we have
    \begin{align*}
        \abs*{ \Real\sbra{\tr\rbra{D \rho_i}} - \tr\rbra*{ \frac{1}{2} P\rbra*{\frac{\rho_0-\rho_1}{4}} \rho_i } }
        & \leq \abs*{ \tr\rbra{D \rho_i} - \tr\rbra*{ \frac{1}{2} P\rbra*{\frac{\rho_0-\rho_1}{4}} \rho_i } } \\
        & \leq \Abs*{ D - \frac{1}{2} P\rbra*{\frac{\rho_0-\rho_1}{4}} } \\
        & \leq \delta_p.
    \end{align*}
    By the property of polynomial $P$, we have 
    \begin{align*}
        & \abs*{ \tr\rbra*{ \frac{1}{2} P\rbra*{\frac{\rho_0-\rho_1}{4}} \rho_i } - \tr\rbra*{ \frac{1}{4} \rho_i \cdot \sgn\rbra{\rho_0 - \rho_1} \cdot \abs*{\frac{\rho_0 - \rho_1}{4}}^{\alpha-1} } } \\
        & \leq \Abs*{ \frac{1}{2} P\rbra*{\frac{\rho_0-\rho_1}{4}} - \frac{1}{4} \sgn\rbra{\rho_0 - \rho_1} \cdot \abs*{\frac{\rho_0 - \rho_1}{4}}^{\alpha-1} } \\
        & \leq \frac{\epsilon_p}{2}. 
    \end{align*}
    By combining the above two inequalities together with \cref{eq:sampl-error}, \cref{eq:final-error} holds. 
\end{proof}

\begin{remark}[A sample-time trade-off for estimating $\Talpha(\rho_0,\rho_1)$] \label{remark:sample-time-tradeoff}
    By directly combining \cref{lemma:twz-samplizer} and \cref{thm:Lalpha-query-algo}, an improved quantum sample complexity upper bound of $O\rbra{1/{\epsilon^{2\alpha + 2 + \frac{2}{\alpha-1}}}}$ for estimating $\Talpha(\rho_0,\rho_1)$ to within additive error $\epsilon$ was proved in~\cite[Theorem 3.6]{CWZ25}. However, their construction turns out to use $\widetilde O\rbra{\poly\rbra{n}/{\epsilon^{8\alpha + 8 +\frac{8}{\alpha-1}}}}$ one- and two-qubit gates, which is not stated explicitly in their work and leads to a time complexity worse than \Cref{thm:Lalpha-sample-algo}. 
\end{remark}

\subsection{Quantum Schatten \texorpdfstring{$\alpha$}{}-power distance estimation for \texorpdfstring{$0<\alpha<1$}{}}
\label{subsec:algo-Schatten-0<alpha<1}

\subsubsection{Rank-efficient quantum query algorithm}

We now present a rank-efficient quantum query algorithm for estimating $\Lalpha(\rho_0,\rho_1)$:

\begin{theorem}[Quantum Schatten-$\alpha$ quasi-norm distance estimation via queries] \label{thm:alpha-quasi-query}
    Suppose that $Q_0$ and $Q_1$ are unitary operators that prepare purifications of mixed quantum states $\rho_0$ and $\rho_1$, respectively, of rank at most $r$. 
    For every real $\alpha\in(0,1)$, there is a quantum query algorithm that estimates $\Lalpha(\rho_0,\rho_1)$ to within additive error $\epsilon$ by using 
    \[
    O\rbra*{\frac{r^{\frac{2}{\alpha}-1}}{\epsilon^{\frac{2}{\alpha}}} \log\rbra*{\frac{r}{\epsilon}}}
    \]
    queries to $Q_0$ and $Q_1$.
\end{theorem}

\begin{algorithm}[!ht]
    \caption{Quantum algorithm for estimating $\Lalpha(\rho_0,\rho_1)$ for $0 < \alpha < 1$ (query access).}
    \label{algo:q-Tsallis-estimation-query-less-1}
    \begin{algorithmic}[1]
        \Require $\rbra{n+a}$-qubit quantum unitary oracles $Q_0$ and $Q_1$ that prepare purifications of $n$-qubit mixed quantum states $\rho_0$ and $\rho_1$, respectively; $0 < \alpha < 1$; $\epsilon_p, \delta_p, \delta_q, \epsilon_H \in \rbra{0, 1}$.
        \Ensure 
        An estimate of $\Lalpha(\rho_0,\rho_1)$ with high probability. 

        \State Implement a unitary operator $U_\nu$ that is a $\rbra{2, n+a+1, 0}$-block-encoding of $\nu \coloneqq \rho_0 - \rho_1$, using $O\rbra{1}$ queries to $Q_0$ and $Q_1$ (by \cref{lemma:block-encoding-of-density-operators,lemma:lcu}). 

        \State Let $P\rbra{x}$ be a polynomial of degree $d = O\rbra{\frac{1}{\delta_p}\log\rbra{\frac{1}{\epsilon_p}}}$ such that $\max_{x \in \sbra{-1, 1} \setminus \rbra{-\delta_p, \delta_p}} \abs{P\rbra{x} - \frac{1}{2} \delta_p^{1-\alpha} \sgn\rbra{x} \abs{x}^{\alpha-1}} \leq \epsilon_p$ and $\max_{x \in \sbra{-1, 1}} \abs{P\rbra{x}} \leq 1$ (by \cref{lemma:uniformPolyApprox-signedNegativePower}).

        \State Implement a unitary operator $U_{P\rbra{\nu}}$ that is a $\rbra{1, n+a+3, \delta_q}$-block-encoding of $\frac{1}{2}P\rbra{\nu/2}$, using $O\rbra{d}$ queries to $U_{\nu}$ (by \cref{lemma:qsvt}).

        \State Let $\tilde x_0$ be an $\epsilon_H$-estimate of $\tr\rbra{ \bra{\bar 0} U_{P\rbra{\nu}} \ket{\bar 0} \rho_0}$, where $\ket{\bar 0} = \ket{0}^{\otimes\rbra{n+a+3}}$, using $O\rbra{1/\epsilon_H}$ queries to $U_{P\rbra{\nu}}$ and $Q_0$ (by \cref{lemma:hadamard}).

        \State Let $\tilde x_1$ be an $\epsilon_H$-estimate of $\tr\rbra{ \bra{\bar 0} U_{P\rbra{\nu}} \ket{\bar 0} \rho_1}$, where $\ket{\bar 0} = \ket{0}^{\otimes\rbra{n+a+3}}$, using $O\rbra{1/\epsilon_H}$ queries to $U_{P\rbra{\nu}}$ and $Q_1$ (by \cref{lemma:hadamard}).

        \State \Return $2^{\alpha} \delta_p^{\alpha-1} \rbra*{\tilde x_0 - \tilde x_1}$.
    \end{algorithmic}
\end{algorithm}

\begin{proof}
    The proof is similar to that of \Cref{lemma:powered-Lalpha-query-algo}. 
    For $0 < \alpha < 1$, we can also expand $\tr\rbra{\abs{\rho_0-\rho_1}^\alpha}$ as: 
    \begin{equation} \label{eq:tr-power-less-1}
        \tr\rbra{\abs{\rho_0-\rho_1}^\alpha} 
        = \tr\rbra*{\rho_0 \cdot \sgn\rbra{\rho_0 - \rho_1} \cdot \abs*{\rho_0 - \rho_1}^{\alpha-1}} - \tr\rbra*{\rho_1 \cdot \sgn\rbra{\rho_0 - \rho_1} \cdot \abs*{\rho_0 - \rho_1}^{\alpha-1}}.
    \end{equation}

    Suppose that $Q_0$ and $Q_1$ are $\rbra{n+a}$-qubit unitary operators that prepare purifications of $\rho_0$ and $\rho_1$, respectively. 
    We present our quantum algorithm as follows with its formal description given in \cref{algo:q-Tsallis-estimation-query-less-1}. 

    \paragraph{Step 1: Construct a block-encoding of $\nu \coloneqq \rho_0 - \rho_1$.} 
    As shown in the proof of \Cref{lemma:powered-Lalpha-query-algo}, we can implement a unitary operator $U_{\nu}$ that is a $\rbra{2, n+a+1, 0}$-block-encoding of $\nu \coloneqq \rho_0 - \rho_1$, using $O\rbra{1}$ queries to $Q_0$ and $Q_1$. 

    \paragraph{Step 2: Construct a block-encoding of $\sgn\rbra{\nu} \cdot \abs{\nu}^{\alpha-1}$.}
    Let $\epsilon_p, \delta_p, \delta_q \in \rbra{0, 1/2}$ be parameters to be determined. 
    By \cref{lemma:uniformPolyApprox-signedNegativePower}, there is a polynomial $P \in \mathbb{R}\sbra{x}$ of degree $d = O\rbra{\frac{1}{\delta_p}\log\rbra{\frac{1}{\epsilon_p}}}$ such that $\max_{x \in \sbra{-1, 1} \setminus \rbra{-\delta_p, \delta_p}} \abs{P\rbra{x} - \frac{1}{2} \delta_p^{1-\alpha} \sgn\rbra{x} \abs{x}^{\alpha-1}} \leq \epsilon_p$ and $\max_{x \in \sbra{-1, 1}} \abs{P\rbra{x}} \leq 1$.
    By \cref{lemma:qsvt} with $P_\bullet \coloneqq \frac{1}{2}P$, $\alpha \coloneqq 1$, $a \coloneqq n+a+1$, $\epsilon \coloneqq 0$ and $d \coloneqq O\rbra{\frac{1}{\delta_p}\log\rbra{\frac{1}{\epsilon_p}}}$, we can implement a quantum circuit $U_{P\rbra{\nu}}$ that is a $\rbra{1, n+a+3, \delta_q}$-block-encoding of $\frac{1}{2}P\rbra{\nu/2}$, by using $O\rbra{\frac{1}{\delta_p}\log\rbra{\frac{1}{\epsilon_p}}}$ queries to $U_{\nu}$. 
    Moreover, the classical description of $U_{P\rbra{\nu}}$ can be computed in deterministic time $\poly\rbra{1/\delta_p, \log\rbra{1/\epsilon_p}, \log\rbra{1/\delta_q}}$.

    \paragraph{Step 3: Estimate $\tr\rbra{P\rbra{\nu} \rho_0}$.}
    Suppose that $U_{P\rbra{\nu}}$ is a $\rbra{1, n+a+3, 0}$-block-encoding of $A$ with $\Abs{A - \frac{1}{2}P\rbra{\nu/2}} \leq \delta_q$. 
    By \cref{lemma:hadamard}, we can obtain an estimate $\tilde x_0$ of $\tr\rbra{A\rho_0}$ to within additive error $\epsilon_H$ by using $O\rbra{1/\epsilon_H}$ queries to $U_{P\rbra{\nu}}$ and $Q_0$ such that
    \begin{equation}\label{eq:lp-had0-alpha-leq1}
    \Pr\sbra[\big]{ \abs*{\tilde x_0 - \tr\rbra*{A \rho_0}} \leq \epsilon_H } \geq 0.9.
    \end{equation}
    Similarly, we can obtain an estimate $\tilde x_1$ of $\tr\rbra{A\rho_1}$ to within additive error $\epsilon_H$ by using $O\rbra{1/\epsilon_H}$ queries to $U_{P\rbra{\nu}}$ and $Q_1$ such that
    \begin{equation}\label{eq:lp-had1-alpha-leq1}
    \Pr\sbra[\big]{ \abs*{\tilde x_1 - \tr\rbra*{A \rho_1}} \leq \epsilon_H } \geq 0.9.
    \end{equation}

    \paragraph{Step 4: Return $2^{\alpha} \delta_p^{\alpha-1} \rbra*{\tilde x_0 - \tilde x_1}$ as the estimate.}

    \paragraph{Error analysis and parameter selection.}
    
    The overall process from Step 1 to Step 3 uses 
    \[
    O\rbra*{\frac{1}{\delta_p}\log\rbra*{\frac{1}{\epsilon_p}}} \cdot O\rbra*{\frac{1}{\epsilon_H}} = O\rbra*{\frac{1}{\delta_p\epsilon_H}\log\rbra*{\frac{1}{\epsilon_p}}}
    \]
    queries to $Q_0$ and $Q_1$, and 
    \[
    O\rbra*{\frac{n+a}{\delta_p\epsilon_H}\log\rbra*{\frac{1}{\epsilon_p}}}
    \]
    one- and two-qubit gates. 
    Moreover, the classical description of this quantum circuit can be computed in deterministic time $\poly\rbra{1/\delta_p, 1/\epsilon_H, \log\rbra{1/\epsilon_p}, \log\rbra{1/\delta_q}}$.
    To determine the parameters, we note that 
    \begin{equation} \label{eq:error-query-less-1}
        \Pr\sbra*{ \abs*{2^{\alpha} \delta_p^{\alpha-1} \rbra*{\tilde x_0 - \tilde x_1} - \Lalpha(\rho_0,\rho_1) } \leq 2^{\alpha} \delta_p^{\alpha - 1} \rbra*{2\epsilon_H + 2\delta_q + 3 r \delta_p + \epsilon_p} } \geq 0.8.
    \end{equation}
    Then, by setting $\delta_p = \rbra{2^{-\alpha-4}\epsilon/r}^{1/\alpha}$ and $\epsilon_H = \delta_q = \epsilon_p = 2^{-\alpha-4} \delta_p^{1-\alpha} \epsilon$, we can estimate $\Lalpha(\rho_0,\rho_1)$ to within additive error $\epsilon$ with success probability at least $0.8$, using 
    \[
    O\rbra*{\frac{r^{\frac{2}{\alpha}-1}}{\epsilon^{\frac{2}{\alpha}}}\log\rbra*{\frac{r}{\epsilon}}}
    \]
    queries to $Q_0$ and $Q_1$ and
    \[
    O\rbra*{\frac{\rbra{n+a}r^{\frac{2}{\alpha}-1}}{\epsilon^{\frac{2}{\alpha}}}\log\rbra*{\frac{r}{\epsilon}}}
    \]
    one- and two-qubit gates. 
    Moreover, the classical description of this quantum circuit can be computed in deterministic time $\poly\rbra{r, 1/\epsilon}$.

    To complete the proof, it remains to show \cref{eq:error-query-less-1}. 
    Using the matrix H\"{o}lder inequality, e.g.,~\cite[Theorem 2]{Bau11}, it follows that:
    \begin{equation*}
    \abs*{\tr\rbra*{\frac{1}{2}P\rbra*{\frac{\nu}{2}}\rho_b} - \tr\rbra{A\rho_b}} \leq \Abs*{\frac{1}{2}P\rbra*{\frac{\nu}{2}} - A} \leq \delta_q
    \end{equation*}
    for $b \in \cbra{0, 1}$.
    Also, the following inequality holds: 
    \begin{equation} \label{eq:intrinsic-error}
    \abs*{ \tr\rbra*{\frac{1}{2}P\rbra*{\frac{\nu}{2}} \nu} - \tr\rbra*{\frac{1}{4} \delta_p^{1-\alpha} \sgn\rbra*{\frac{\nu}{2}} \abs*{\frac{\nu}{2}}^{\alpha-1} \nu} } \leq 3 r \delta_p + \epsilon_p.
    \end{equation}
    To see this, suppose that $\nu = \rho_0 - \rho_1 = \sum_j \lambda_j \ketbra{\psi_j}{\psi_j}$ is the spectral decomposition with $\sum_j\abs{\lambda_j} \leq 2$ and $\abs{\lambda_j} \leq 1$ for all $j$. 
    Then, 
    \begin{align*}
        \abs*{ \tr\rbra*{\frac{1}{2}P\rbra*{\frac{\nu}{2}} \nu} - \tr\rbra*{\frac{1}{4} \delta_p^{1-\alpha} \sgn\rbra*{\frac{\nu}{2}} \abs*{\frac{\nu}{2}}^{\alpha-1} \nu} }
        & \leq \sum_{j} \abs*{g\rbra{\lambda_j}},
    \end{align*}
    where
    \[
    g\rbra{x} \coloneqq \frac{1}{2}P\rbra*{\frac{x}{2}} x - \frac{1}{4} \delta_p^{1-\alpha} \sgn\rbra*{\frac{x}{2}} \abs*{\frac{x}{2}}^{\alpha-1} x.
    \]
    Here, we split the summation into three parts:
    \[
    \sum_{j} = \sum_{j \colon \abs{\lambda_j} > 2\delta_p} + \sum_{j \colon 0 < \abs{\lambda_j} < 2\delta_p} + \sum_{j \colon \lambda_j = 0}. 
    \]
    For the first part, 
    \begin{align*}
        \sum_{j \colon \abs{\lambda_j} > 2\delta_p} \abs*{g\rbra{\lambda_j}}
        & = \sum_{j \colon \abs{\lambda_j} > 2\delta_p} \frac{1}{2} \abs*{\lambda_j} \abs*{P\rbra*{\frac{\lambda_j}{2}} - \frac{1}{2} \delta_p^{1-\alpha} \sgn\rbra*{\frac{\lambda_j}{2}} \abs*{\frac{\lambda_j}{2}}^{\alpha-1}} \\
        & \leq \sum_{j \colon \abs{\lambda_j} > 2\delta_p} \frac{1}{2} \abs*{\lambda_j} \epsilon_p \\
        & \leq \epsilon_p.
    \end{align*}
    For the second part, 
    \begin{align*}
        \sum_{j \colon 0 < \abs{\lambda_j} < 2\delta_p} \abs*{g\rbra{\lambda_j}}
        & \leq \sum_{j \colon 0 < \abs{\lambda_j} < 2\delta_p} \rbra*{ \frac{1}{2} \abs*{\lambda_j} \abs*{P\rbra*{\frac{\lambda_j}{2}}} + \frac{1}{4} \delta_p^{1-\alpha} \abs*{\frac{\lambda_j}{2}}^{\alpha-1} \abs{\lambda_j} } \\
        & \leq \sum_{j \colon 0 < \abs{\lambda_j} < 2\delta_p} \rbra*{ \delta_p + \frac{1}{2} \delta_p } \\
        & \leq 3 r \delta_p.
    \end{align*}
    For the third part, 
    \[
    \sum_{j \colon \lambda_j = 0} \abs*{g\rbra{\lambda_j}} = 0.
    \]
    Summing up these three cases yields \cref{eq:intrinsic-error}.

    By \cref{eq:lp-had0-alpha-leq1,eq:lp-had1-alpha-leq1}, we have that with probability at least $0.9^2 > 0.8$, it holds that $\abs{\tilde x_0 - \tr\rbra{A\rho_0}} \leq \epsilon_H$ and $\abs{\tilde x_1 - \tr\rbra{A\rho_1}} \leq \epsilon_H$, in which case, 
    \begin{align*}
        \abs*{\rbra*{\tilde x_0 - \tilde x_1} - \frac{\delta_p^{1-\alpha}}{2^{\alpha+1}} \tr\rbra*{\abs{\nu}^\alpha}}
        & \leq \abs{\tilde x_0 - \tr\rbra{A\rho_0}} + \abs{\tilde x_1 - \tr\rbra{A\rho_1}} \\
        & \qquad + \abs*{\tr\rbra{A\rho_0} - \tr\rbra*{\frac{1}{2}P\rbra*{\frac{\nu}{2}}\rho_0}} +  \abs*{\tr\rbra{A\rho_1} - \tr\rbra*{\frac{1}{2}P\rbra*{\frac{\nu}{2}}\rho_1}} \\
        & \qquad + \abs*{\rbra*{\tr\rbra*{\frac{1}{2}P\rbra*{\frac{\nu}{2}}\rho_0} - \tr\rbra*{\frac{1}{2}P\rbra*{\frac{\nu}{2}}\rho_1}} - \frac{\delta_p^{1-\alpha}}{2^{\alpha+1}} \tr\rbra*{\abs{\nu}^\alpha}} \\
        & \leq 2\epsilon_H + 2\delta_q + 3 r \delta_p + \epsilon_p.
    \end{align*}
    These yield \cref{eq:error-query-less-1}.  
\end{proof}

\subsubsection{Rank-efficient quantum sample algorithm}

We then present rank-efficient quantum sample algorithms for estimating $\Lalpha(\rho_0,\rho_1)$: 

\begin{theorem}[Quantum Schatten $\alpha$-power distance estimation via samples]
    \label{thm:alpha-quasi-sample}
    For every constant $0 < \alpha < 1$, there is an explicit quantum algorithm that estimates $\Lalpha(\rho_0,\rho_1)$ to within additive error $\epsilon$ using 
    \[
    O\rbra*{\frac{r^{\frac{5}{\alpha}-3}}{\epsilon^{\frac{5}{\alpha}}}\log^4\rbra*{\frac{r}{\epsilon}}}
    \]
    samples of $n$-qubit quantum states $\rho_0$ and $\rho_1$, and 
    \[
    O\rbra*{\frac{nr^{\frac{5}{\alpha}-3}}{\epsilon^{\frac{5}{\alpha}}}\log^4\rbra*{\frac{r}{\epsilon}}}
    \]
    one- and two-qubit gates, where $r$ is a known upper bound on the ranks of $\rho_0$ and $\rho_1$. 
\end{theorem}

\begin{algorithm}[!ht]
    \caption{Quantum algorithm for estimating $\Lalpha(\rho_0,\rho_1)$ for $0 < \alpha < 1$ (sample access).}
    \label{algo:alpha-distance-less-1}
    \begin{algorithmic}[1]
        \Require Independent and identical samples of $n$-qubit mixed quantum states $\rho_0$ and $\rho_1$, and parameters $0 < \alpha < 1$ and $\delta, \epsilon_p, \delta_p, \delta_q \in \rbra{0, 1}$.
        \Ensure 
        An estimate of $\Lalpha(\rho_0,\rho_1)$ with high probability. 

        \begin{tcolorbox}[colback=gray!12]
        \Function{ApproxDiffPower}{$\alpha, \epsilon_p, \delta_p, \delta_q$}${}^{U_1, U_2}$
        \renewcommand{\algorithmicrequire}{\qquad \textbf{Input:}}
        \renewcommand{\algorithmicensure}{\qquad \textbf{Output:}} 
        \Require Unitary $\rbra{1, a, 0}$-block-encodings $U_1$ and $U_2$ of $A$ and $B$, respectively, and parameters $\alpha\in(0,1)$, $\epsilon_p, \delta_p, \delta_q \in \rbra{0, 1/2}$. 
        \Ensure A unitary operator.

        \State Let $U_3$ be a $\rbra{2, a+1, 0}$-block-encoding of $A-B$ by using $O\rbra{1}$ queries to $U_1$ and $U_2$ (by \cref{lemma:lcu}). 
        \State Let $P\rbra{x}$ be a polynomial of degree $d = O\rbra{\log\rbra{1/\epsilon_p}/{\delta_p}}$ such that $\max_{x \in \sbra{-1, 1} \setminus \rbra{-\delta_p, \delta_p}} \abs{P\rbra{x} - \frac{1}{2} \delta_p^{1-\alpha} \sgn\rbra{x} \abs{x}^{\alpha-1}} \leq \epsilon_p$ and $\max_{x \in \sbra{-1, 1}} \abs{P\rbra{x}} \leq 1$ (by \cref{lemma:uniformPolyApprox-signedNegativePower}).
        \State Construct a unitary $\rbra{1, a+3, \delta_q}$-block-encoding $U_4$ of $\frac{1}{2}P\rbra{\frac{A-B}{2}}$ by using $O\rbra{d}$ queries to $U_3$ (by \cref{lemma:qsvt}). 
        \State \Return $U_4$.
        \EndFunction
        \end{tcolorbox}

        \For {$i \in \binset$}
        \State $k \gets \Theta\rbra{1/\epsilon_H^2}$.
        \For {$j \in \cbra{1, 2, \dots, k}$}
            \State Let $b_{i,j}' \in \binset$ be the outcome of the Hadamard test (by \cref{lemma:hadamard}) performed on $\rho_i$ and $\mathsf{Samplize}_{\delta}\ave{\begin{tcolorbox}[nobeforeafter,
  after={\xspace},
  hbox,
  tcbox raise base,
  fontupper=\ttfamily,
  colback=gray!12,
  size=fbox] $\texttt{ApproxDiffPower}\rbra{\alpha, \epsilon_p, \delta_p, \delta_q}^{U_1, U_2}$\end{tcolorbox}}\sbra{\rho_0, \rho_1}$ (as if it were unitary).
        \EndFor
        \State $p_i \gets \frac{1}{k} \sum_{j=1}^k b_{i,j}'$.
        \EndFor

        \State \Return $2^{2\alpha-1} \delta_p^{\alpha-1} \rbra{p_0-p_1}$.
    \end{algorithmic}
\end{algorithm}

\begin{proof}
    To estimate $\tr\rbra{\abs{\rho_0-\rho_1}^{\alpha}}$, our approach estimates the terms $\tr\rbra{\rho_0 \cdot \sgn\rbra{\rho_0 - \rho_1} \cdot \abs{\rho_0 - \rho_1}^{\alpha-1}}$ and $\tr\rbra{\rho_1 \cdot \sgn\rbra{\rho_0 - \rho_1} \cdot \abs{\rho_0 - \rho_1}^{\alpha-1}}$ in \Cref{eq:tr-power-less-1} using samples of $\rho_0$ and $\rho_1$. Specifically, our quantum sample algorithm extends the quantum query algorithm in \Cref{thm:alpha-quasi-query} via the (multi-)samplizer~\cite{WZ24b,WZ24c}.
    We present our quantum algorithm as follows with its formal description given in \cref{algo:alpha-distance-less-1}.
    
    \paragraph{Step 1: Construct a block-encoding of $A - B$ given block-encodings of $A$ and $B$.}
    This step is the same as Step 1 in the proof of \cref{thm:Lalpha-sample-algo}.
    Given two unitary operators $U_1$, a $\rbra{1, a, 0}$-block-encoding of $A$, and $U_2$, a $\rbra{1, a, 0}$-block-encoding of $B$, we can implement a unitary operator $U_{3}$ that is a $\rbra{2, a+1, 0}$-block-encoding of $A - B$, using $1$ query to $U_1$ and $U_2$. 

    \paragraph{Step 2: Construct a block-encoding of $\sgn\rbra{A-B} \cdot \abs{A - B}^{\alpha-1}$.}
    Let $\epsilon_p, \delta_p, \delta_q \in \rbra{0, 1/2}$ be parameters to be determined. 
    By \cref{lemma:uniformPolyApprox-signedNegativePower}, there is a polynomial $P \in \mathbb{R}\sbra{x}$ of degree $d = O\rbra{\log\rbra{1/\epsilon_p}/{\delta_p}}$ such that $\max_{x \in \sbra{-1, 1} \setminus \rbra{-\delta_p, \delta_p}} \abs{P\rbra{x} - \frac{1}{2} \delta_p^{1-\alpha} \sgn\rbra{x} \abs{x}^{\alpha-1}} \leq \epsilon_p$ and $\max_{x \in \sbra{-1, 1}} \abs{P\rbra{x}} \leq 1$.
    By \cref{lemma:qsvt} with $P_\bullet \coloneqq \frac{1}{2}P$, $\alpha \coloneqq 1$, $a \coloneqq a+1$, $\epsilon \coloneqq 0$ and $d \coloneqq O\rbra{\log\rbra{1/\epsilon_p}/{\delta_p}}$, we can implement a quantum circuit $U_{4}$ that is a $\rbra{1, a+3, \delta_q}$-block-encoding of $\frac{1}{2}P\rbra{\frac{A-B}{2}}$, by using $O\rbra{\log\rbra{1/\epsilon_p}/{\delta_p}}$ queries to $U_{3}$. 
    Moreover, the classical description of $U_{4}$ can be computed in deterministic time $\poly\rbra{1/\delta_p, \log\rbra{1/\epsilon_p}, \log\rbra{1/\delta_q}}$.

    Combining Steps 1 and 2, let $\texttt{ApproxDiffPower}\rbra{\alpha, \epsilon_p, \delta_p, \delta_q}^{U_1,U_2}$ be the implementation of $U_4$ by using $O\rbra{\log\rbra{1/\epsilon_p}/{\delta_p}}$ queries to $U_1$ and $U_2$. 

    \paragraph{Step 3: Estimate $\tr\rbra{\rho_i P\rbra{\rho_0-\rho_1}}$.}
    For our purpose, we first consider the case that $A = \rho_0/2$ and $B = \rho_1/2$. 
    In this case, $U_4$ is a $\rbra{1, a+3, \delta_q}$-block-encoding of $\frac{1}{2}P\rbra{\frac{\rho_0-\rho_1}{4}}$.
    For convenience, suppose that $U_4$ is a $\rbra{1, a+3, 0}$-block-encoding of some operator $D$, where $D$ satisfies 
    \begin{equation} \label{eq:diff-op-quasi-less-1}
    \Abs*{D - \frac{1}{2}P\rbra*{\frac{\rho_0-\rho_1}{4}}} \leq \delta_q. 
    \end{equation}
    For $i \in \binset$, if we perform the Hadamard test (\cref{lemma:hadamard}) on $\rho_i$ and $U_4$, then an outcome $b_i \in \cbra{0, 1}$ will be obtained with
    \[
    \Pr\sbra{b_i = 0} = \frac{1}{2} + \frac{1}{2} \Real\sbra{\tr\rbra{D \rho_i}}. 
    \]
    Let $\delta > 0$ be a parameter to be determined.
    Then using the multi-samplizer (\cref{lemma:multi-samplizer}), we can approximately implement $U_4$ by $\mathsf{Samplize}_{\delta}\ave{\texttt{ApproxDiffPower}\rbra{\alpha, \epsilon_p, \delta_p, \delta_q}^{U_1, U_2}}\sbra{\rho_0, \rho_1}$, using
    \[
    O\rbra*{\frac{1}{\delta}\rbra*{\frac{1}{\delta_p}\log\rbra*{\frac{1}{\epsilon_p}}}^2 \log^2\rbra*{\frac{1}{\delta}\frac{1}{\delta_p}\log\rbra*{\frac{1}{\epsilon_p}}}} = O\rbra*{\frac{1}{\delta \delta_p^2} \log^2\rbra*{\frac{1}{\epsilon_p}} \log^2\rbra*{\frac{1}{\delta\delta_p}\log\rbra*{\frac{1}{\epsilon_p}}}}
    \]
    samples of $\rho_0$ and $\rho_1$. 
    Let $b_i' \in \cbra{0, 1}$ be the outcome of the Hadamard test (by \cref{lemma:hadamard}) on $\rho_i$ and $\mathsf{Samplize}_{\delta}\ave{\texttt{ApproxDiffPower}\rbra{\alpha, \epsilon_p, \delta_p, \delta_q}^{U_1, U_2}}\sbra{\rho_0, \rho_1}$ (as if it were unitary), then 
    \[
    \abs{\Pr\sbra{b_i = 0} - \Pr\sbra{b'_i = 0}} \leq \delta. 
    \]
    By Hoeffding's inequality, we can obtain an estimate $p_i$ of $\Real\sbra{\tr\rbra{D \rho_i}}$ to additive error $\epsilon_H$ with success probability $\geq 0.9$ by $O\rbra{1/\epsilon_H^2}$ repetitions of the Hadamard test, i.e., 
    \begin{equation} \label{eq:sampl-error-less-1}
    \Pr\sbra[\big]{\abs*{p_i - \Real\sbra{\tr\rbra{D \rho_i}}} \leq \delta+\epsilon_H} \geq 0.9.
    \end{equation}

    \paragraph{Step 4: Return $2^{2\alpha-1} \delta_p^{\alpha-1} \rbra{p_0-p_1}$ as the estimate.}

    \paragraph{Error analysis and parameter selection.}
    
    The overall process from Step 1 to Step 3 uses 
    \[
    O\rbra*{\frac{1}{\epsilon_H^2}} \cdot O\rbra*{\frac{1}{\delta \delta_p^2} \log^2\rbra*{\frac{1}{\epsilon_p}} \log^2\rbra*{\frac{1}{\delta\delta_p}\log\rbra*{\frac{1}{\epsilon_p}}}} = O\rbra*{\frac{1}{\delta \delta_p^2 \epsilon_H^2} \log^2\rbra*{\frac{1}{\epsilon_p}} \log^2\rbra*{\frac{1}{\delta\delta_p}\log\rbra*{\frac{1}{\epsilon_p}}}}
    \]
    samples of $\rho_0$ and $\rho_1$, and
    \[
    O\rbra*{\frac{n}{\delta \delta_p^2 \epsilon_H^2} \log^2\rbra*{\frac{1}{\epsilon_p}} \log^2\rbra*{\frac{1}{\delta\delta_p}\log\rbra*{\frac{1}{\epsilon_p}}}}
    \]
    one- and two-qubit gates. 
    Moreover, the classical description of this quantum circuit can be computed in deterministic time $\poly\rbra{1/\delta, 1/\delta_p, 1/\epsilon_H, \log\rbra{1/\epsilon_p}, \log\rbra{1/\delta_q}}$.
    To determine the parameters, we note that 
    \begin{equation} \label{eq:error-sample-less-1}
        \Pr\sbra*{ \abs*{ 2^{2\alpha-1} \delta_p^{\alpha-1} \rbra{p_0-p_1} - \Lalpha(\rho_0,\rho_1)} \leq 2^{2\alpha-1} \delta_p^{\alpha - 1} \rbra*{2\delta + 2\epsilon_H + 2\delta_q + 6 r \delta_p + \epsilon_p} } \geq 0.8.
    \end{equation}
    Then, by setting $\delta_p = \rbra{4^{-\alpha-4}\epsilon/r}^{1/\alpha}$ and $\delta = \epsilon_H = \delta_q = \epsilon_p = 4^{-\alpha-2} \delta_p^{1-\alpha} \epsilon$, we can estimate $\Lalpha(\rho_0,\rho_1)$ to within additive error $\epsilon$ with success probability at least $0.8$, using 
    \[
    O\rbra*{\frac{r^{\frac{5}{\alpha}-3}}{\epsilon^{\frac{5}{\alpha}}}\log^4\rbra*{\frac{r}{\epsilon}}}
    \]
    samples of $\rho_0$ and $\rho_1$, and 
    \[
    O\rbra*{\frac{nr^{\frac{5}{\alpha}-3}}{\epsilon^{\frac{5}{\alpha}}}\log^4\rbra*{\frac{r}{\epsilon}}}
    \]
    one- and two-qubit gates. 
    Moreover, the classical description of this quantum circuit can be computed in deterministic time $\poly\rbra{r, 1/\epsilon}$.

    To complete the proof, it remains to show \cref{eq:error-sample-less-1}.
    By \cref{eq:diff-op-quasi-less-1}, we have
    \begin{align*}
        \abs*{ \Real\sbra{\tr\rbra{D \rho_i}} - \tr\rbra*{ \frac{1}{2} P\rbra*{\frac{\rho_0-\rho_1}{4}} \rho_i } }
        & \leq \abs*{ \tr\rbra{D \rho_i} - \tr\rbra*{ \frac{1}{2} P\rbra*{\frac{\rho_0-\rho_1}{4}} \rho_i } } \\
        & \leq \Abs*{ D - \frac{1}{2} P\rbra*{\frac{\rho_0-\rho_1}{4}} } \\
        & \leq \delta_q.
    \end{align*}
    On the other hand, 
    \begin{equation} \label{eq:abs-diff-sample-less-1}
        \abs*{ \frac{1}{2}\tr\rbra*{P\rbra*{\frac{\nu}{4}} \nu} -  \frac{1}{2}\tr\rbra*{\frac{1}{2} \delta_p^{1-\alpha} \sgn\rbra*{\frac{\nu}{4}} \abs*{\frac{\nu}{4}}^{\alpha-1} \nu} } \leq \epsilon_p + 6 r \delta_p.
    \end{equation}
    To see this, suppose that $\nu = \rho_0 - \rho_1 = \sum_j \lambda_j \ketbra{\psi_j}{\psi_j}$ is the spectral decomposition with $\sum_j\abs{\lambda_j} \leq 2$ and $\abs{\lambda_j} \leq 1$ for all $j$. 
    Then, 
    \[
    \abs*{ \frac{1}{2}\tr\rbra*{P\rbra*{\frac{\nu}{4}} \nu} -  \frac{1}{2}\tr\rbra*{\frac{1}{2} \delta_p^{1-\alpha} \sgn\rbra*{\frac{\nu}{4}} \abs*{\frac{\nu}{4}}^{\alpha-1} \nu} } \leq \sum_j \abs*{g\rbra{\lambda_j}},
    \]
    where
    \[
    g\rbra{x} \coloneqq \frac{1}{2}P\rbra*{\frac{x}{4}} x - \frac{1}{4} \delta_p^{1-\alpha} \sgn\rbra*{\frac{x}{4}} \abs*{\frac{x}{4}}^{\alpha-1} x.
    \]
    Here, we split the summation into three parts:
    \[
    \sum_{j} = \sum_{j \colon \abs{\lambda_j} > 4\delta_p} + \sum_{j \colon 0 < \abs{\lambda_j} < 4\delta_p} + \sum_{j \colon \lambda_j = 0}. 
    \]
    For the first part, 
    \begin{align*}
        \sum_{j \colon \abs{\lambda_j} > 4\delta_p} \abs*{g\rbra{\lambda_j}}
        & = \sum_{j \colon \abs{\lambda_j} > 4\delta_p} \frac{1}{2} \abs*{\lambda_j} \abs*{P\rbra*{\frac{\lambda_j}{4}} - \frac{1}{2} \delta_p^{1-\alpha} \sgn\rbra*{\frac{\lambda_j}{4}} \abs*{\frac{\lambda_j}{4}}^{\alpha-1}} \\
        & \leq \sum_{j \colon \abs{\lambda_j} > 4\delta_p} \frac{1}{2} \abs*{\lambda_j} \epsilon_p \\
        & \leq \epsilon_p.
    \end{align*}
    For the second part, 
    \begin{align*}
        \sum_{j \colon 0 < \abs{\lambda_j} < 4\delta_p} \abs*{g\rbra{\lambda_j}}
        & \leq \sum_{j \colon 0 < \abs{\lambda_j} < 4\delta_p} \rbra*{ \frac{1}{2} \abs*{\lambda_j} \abs*{P\rbra*{\frac{\lambda_j}{4}}} + \frac{1}{4} \delta_p^{1-\alpha} \abs*{\frac{\lambda_j}{4}}^{\alpha-1} \abs{\lambda_j} } \\
        & \leq \sum_{j \colon 0 < \abs{\lambda_j} < 4\delta_p} \rbra*{ 2 \delta_p + \delta_p } \\
        & \leq 6 r \delta_p.
    \end{align*}
    For the third part, 
    \[
    \sum_{j \colon \lambda_j = 0} \abs*{g\rbra{\lambda_j}} = 0.
    \]
    Summing up these three cases yields \cref{eq:abs-diff-sample-less-1}.

    By \cref{eq:sampl-error-less-1}, we have that with probability at least $0.9^2 > 0.8$, it holds that $\abs{p_0 - \Real\sbra{\tr\rbra{D\rho_0}}} \leq \delta + \epsilon_H$ and $\abs{p_1 - \Real\sbra{\tr\rbra{D\rho_1}}} \leq \delta + \epsilon_H$, in which case, 
    \begin{align*}
        \abs*{\rbra*{p_0 - p_1} - \frac{\delta_p^{1-\alpha}}{4^{\alpha}} \tr\rbra*{\abs{\nu}^\alpha}}
        & \leq \abs{p_0 - \Real\sbra{\tr\rbra{D\rho_0}}} + \abs{p_1 - \Real\sbra{\tr\rbra{D\rho_1}}} \\
        & \qquad + \abs*{\Real\sbra{\tr\rbra{D\rho_0}} - \tr\rbra*{\frac{1}{2}P\rbra*{\frac{\nu}{4}}\rho_0}} +  \abs*{\Real\sbra{\tr\rbra{D\rho_1}} - \tr\rbra*{\frac{1}{2}P\rbra*{\frac{\nu}{4}}\rho_1}} \\
        & \qquad + \abs*{\rbra*{\tr\rbra*{\frac{1}{2}P\rbra*{\frac{\nu}{4}}\rho_0} - \tr\rbra*{\frac{1}{2}P\rbra*{\frac{\nu}{4}}\rho_1}} - \frac{\delta_p^{1-\alpha}}{4^{\alpha}} \tr\rbra*{\abs{\nu}^\alpha}} \\
        & \leq 2\delta + 2\epsilon_H + 2\delta_q + 6 r \delta_p + \epsilon_p.
    \end{align*}
    These yield \cref{eq:error-sample-less-1}. 
\end{proof}

Another quantum sample algorithm with worse time complexity can be obtained by directly combining \cref{lemma:twz-samplizer} and \cref{thm:alpha-quasi-query}, as mentioned in \cref{remark:sample-time-tradeoff}. 

\begin{theorem} \label{thm:alpha-le-1-better-samplizer}
    For every constant $0 < \alpha < 1$, there is a quantum algorithm that estimates $\Lalpha(\rho_0,\rho_1)$ to within additive error $\epsilon$, using 
    \[
    O\rbra*{\frac{r^{\frac{4}{\alpha}-2}}{\epsilon^{\frac{4}{\alpha}}} \log^2\rbra*{\frac{r}{\epsilon}}}
    \]
    samples of $n$-qubit quantum states $\rho_0$ and $\rho_1$ and 
    \[
    \widetilde O\rbra*{\frac{r^{\frac{16}{\alpha}-8}}{\epsilon^{\frac{16}{\alpha}}} \log^8\rbra*{\frac{r}{\epsilon}} \poly\rbra{n}}
    \] 
    one- and two-qubit gates.
\end{theorem}

\section{New inequalities relating trace distance to Schatten norm and power distances}
\label{sec:new-inequalities}

In this section, we provide new inequalities that connect the trace distance to two Schatten-type distances: the quantum Schatten $\alpha$-norm distance for $\alpha \geq 1$ in \Cref{subsec:inequalities-T-vs-Talpha}, and the quantum Schatten $\alpha$-power distance for $0<\alpha<1$ in \Cref{subsec:inequalities-T-vs-Lalpha}. Notably, both distances are metrics in their respective regimes.

\subsection{Rank-dependent inequalities between trace distance and \texorpdfstring{$\Talpha$}{} for \texorpdfstring{$\alpha\geq 1$}{}}
\label{subsec:inequalities-T-vs-Talpha}

We generalize the rank-dependent inequalities between the (squared) Hilbert-Schmidt distance and the trace distance, as demonstrated in~\cite[Appendix G]{Coles12} and~\cite[Theorem 1]{CCC19} for the case of $\alpha=2$, to all $1 \leq \alpha \leq \infty$:

\begin{theorem}[$\Talpha$ vs.~$\td$ for $\alpha \geq 1$]
    \label{thm:Talpha-vs-T}
    Let $\rho_0$ and $\rho_1$ be quantum states. The following holds:
    \begin{enumerate}[label={\upshape(\arabic*)},topsep=0.33em, itemsep=0.33em, parsep=0.33em]
        \item For any $\alpha$ in the range $1\leq\alpha < \infty$, 
        \label{thmitem:Talpha-inequalities-alpha}
        \[2^{1-\frac{1}{\alpha}} \cdot \Talpha(\rho_0,\rho_1) \leq \td(\rho_0,\rho_1) \leq 2 \rbra*{ \rank(\rho_0)^{1-\alpha} + \rank(\rho_1)^{1-\alpha} }^{-\frac{1}{\alpha}} \cdot \Talpha(\rho_0,\rho_1).\]
        \item For $\alpha = \infty$, 
        $2 \cdot \td_{\infty}(\rho_0,\rho_1) \leq \td(\rho_0,\rho_1) \leq 2 \min\cbra*{\rank(\rho_0),\rank(\rho_1)} \cdot \td_{\infty}(\rho_0,\rho_1).$
        \label{thmitem:Talpha-inequalities-infty}
    \end{enumerate}
\end{theorem}

It is worth noting that \Cref{thmitem:Talpha-inequalities-alpha,thmitem:Talpha-inequalities-infty} in \Cref{thm:Talpha-vs-T} are consistent, specifically 
\[\lim_{\alpha \to \infty} \rbra*{ \rank(\rho_0)^{1-\alpha} + \rank(\rho_1)^{1-\alpha} }^{-\frac{1}{\alpha}} = \min\cbra*{\rank(\rho_0),\rank(\rho_1)}.\]  
Additionally, the inequalities in \Cref{thm:Talpha-vs-T} sharpen the inequalities between the trace norm and the Schatten norm (see, e.g.,~\cite[Equation (1.31)]{AS17}): 
\begin{equation}
    \label{eq:Schatten-norms-bound}
    \forall 1 \leq p \leq \infty, \quad \Abs*{A}_p \leq \Abs*{A}_1 \leq r_A^{1-1/p} \cdot \Abs*{A}_p. 
\end{equation}

By considering the maximum rank of $\rho_0$ and $\rho_1$, we can derive a simplified form of \Cref{thm:Talpha-vs-T} for convenience: 

\begin{corollary}[$\Talpha$ vs.~$\td$, simplified]
    \label{corr:Talpha-vs-T-simplified}
    For any quantum states $\rho_0$ and $\rho_1$, the following holds: 
    \[ \forall 1 \leq \alpha < \infty, \quad 2^{1-\frac{1}{\alpha}} \cdot \Talpha(\rho_0,\rho_1) \leq \td(\rho_0,\rho_1) \leq \rbra*{2 \max\cbra*{\rank(\rho_0),\rank(\rho_1)}}^{1-\frac{1}{\alpha}} \cdot \Talpha(\rho_0,\rho_1). \]
\end{corollary}

Moreover, for pure quantum states, \Cref{thm:Talpha-vs-T} yields the following identity: 
\begin{corollary}[$\Talpha=\td$ for pure states]
    \label{corr:Talpha-vs-T-pure}
    For any pure states $\ketbra{\psi_0}{\psi_0}$ and $\ketbra{\psi_1}{\psi_1}$, we have: 
    \[ \forall 1\leq \alpha \leq \infty, \quad 2^{1-\frac{1}{\alpha}} \cdot \Talpha\rbra*{\ketbra{\psi_0}{\psi_0},\ketbra{\psi_1}{\psi_1}} = \td\rbra*{\ketbra{\psi_0}{\psi_0},\ketbra{\psi_1}{\psi_1}}. \]
\end{corollary}

We now proceed with the proof of \Cref{thm:Talpha-vs-T}: 

\begin{proof}[Proof of \Cref{thm:Talpha-vs-T}]
    We begin by defining the following positive semi-definite matrices:
    \begin{equation}
        \label{eq:pos-and-neg-in-stateDiff}
        \varsigma_0 \coloneqq \frac{1}{2}\rbra*{ \rho_0-\rho_1 + \abs*{\rho_0-\rho_1} } \quad\text{and}\quad \varsigma_1 \coloneqq \frac{1}{2}\rbra*{ \rho_1-\rho_0 + \abs*{\rho_0-\rho_1} }.
    \end{equation}

    It is easy to verify that $\rho_0-\rho_1 = \varsigma_0-\varsigma_1$, the supports of $\varsigma_0$ and $\varsigma_1$ are orthogonal, and
    \begin{equation}
        \label{eq:T-positive-supports}
        \forall b \in \binset, \quad \tr(\varsigma_b) = \frac{1}{2} \tr\abs*{\rho_0-\rho_1} = \td(\rho_0,\rho_1). 
    \end{equation}
    If $\td(\rho_0,\rho_1)=0$, then $\rho_0=\rho_1$ and the statement is trivial. Assume from now on that $\td(\rho_0,\rho_1)>0$, and consequently, $\tr(\varsigma_b)>0$ for $b\in\binset$.

    \paragraph{The case $1 \leq \alpha < \infty$.}
    We now establish \Cref{thmitem:Talpha-inequalities-alpha}. The lower bound follows directly from the inequality $\frac{1}{2} \tr\rbra*{\abs{\rho_0-\rho_1}^{\alpha}} \leq \td(\rho_0,\rho_1)^{\alpha}$, which is proven in~\cite[Lemma 4.7]{LW25}. To prove the upper bound, we proceed by noticing the following equalities:
    \begin{equation}
        \label{eq:Talpha-vs-T-expression}
        \tr\rbra*{ \abs*{\rho_0-\rho_1}^{\alpha} } = \tr\rbra*{ \abs*{\varsigma_0 - \varsigma_1}^{\alpha} } = \tr\rbra*{\varsigma_0^{\alpha}} + \tr\rbra*{\varsigma_1^{\alpha}}.
    \end{equation}

    Hence, it suffices to bound $\tr\rbra*{\varsigma_b}^{\alpha}$ for $b\in\binset$. To achieve this, we introduce the normalized states $\widetilde{\varsigma}_b \coloneqq \varsigma_b / \tr(\varsigma_b)$, which enables us to derive the following inequalities: 
    \begin{equation}
        \label{eq:Talpha-vs-T-bounds}
        \forall b \in \binset, \quad \tr\rbra*{ \widetilde{\varsigma}_b^{\alpha} } 
        \geq \sum_{i=1}^{\rank\rbra*{\widetilde{\varsigma}_b}} \rank\rbra*{\widetilde{\varsigma}_b}^{-\alpha} = \rank\rbra*{\widetilde{\varsigma}_b}^{1-\alpha} \geq \rank\rbra*{\rho_b}^{1-\alpha}. 
    \end{equation}
    
    Here, the first inequality follows from the convexity of $x^q$, while the last inequality relies on the facts that
    \begin{equation}
        \label{eq:stateDiff-rank-bound}
        \forall b\in\binset, \quad \rank\rbra*{\rho_b} \geq \rank\rbra*{\widetilde{\varsigma}_b},\footnote{We only consider the case $b=0$, as the case $b=1$ is identical. Let $\Pi_0$ be the projector onto $\supp(\varsigma_0)$. Then $\Pi_0\rho_0 \Pi_0 - \Pi_0\rho_1 \Pi_0 = \Pi_0(\rho_0-\rho_1)\Pi_0 = \varsigma_0$, and hence $\Pi_0\rho_0 \Pi_0 = \varsigma_0 + \Pi_0\rho_1 \Pi_0 \succeq \varsigma_0$. Since $\varsigma_0$ is positive definite on $\supp(\varsigma_0)$, the operator $\Pi_0\rho_0 \Pi_0$ is also positive definite on $\supp(\varsigma_0)$, which implies that $\dim\supp(\varsigma_0) \leq \rank(\rho_0)$.}
    \end{equation}
    and that the function $x^{1-\alpha}$ is monotonically non-increasing for $\alpha\geq 1$. 
    Plugging \Cref{eq:Talpha-vs-T-bounds} into \Cref{eq:Talpha-vs-T-expression}, the following holds: 
    \begin{align*}
        \tr\rbra*{ \abs*{\rho_0-\rho_1}^{\alpha} } = \tr\rbra*{\varsigma_0^{\alpha}} + \tr\rbra*{\varsigma_1^{\alpha}}
        &\geq \tr\rbra*{\varsigma_0}^{\alpha} \rank\rbra*{\rho_0}^{1-\alpha} + \tr\rbra*{\varsigma_1}^{\alpha} \rank\rbra*{\rho_1}^{1-\alpha}\\
        &= \td\rbra*{\rho_0,\rho_1}^{\alpha} \rbra*{ \rank\rbra*{\rho_0}^{1-\alpha} + \rank\rbra*{\rho_1}^{1-\alpha} }.
    \end{align*}
    
    We thus prove \Cref{thmitem:Talpha-inequalities-alpha} by noting that the last line follows from \Cref{eq:T-positive-supports}.

    \paragraph{The case $\alpha = \infty$.} Next, we demonstrate \Cref{thmitem:Talpha-inequalities-infty}. We start by establishing the lower bound. Let $\lambda(A)$ denote the largest eigenvalue of a matrix $A$. Then, the following holds
    \begin{align*}
        \td(\rho_0,\rho_1) = \frac{1}{2} \rbra*{\tr(\varsigma_0)+\tr(\varsigma_1)}
        &\geq \max\cbra*{\lambda_{\max}(\varsigma_0), \lambda_{\max}(\varsigma_1)}\\
        &= 2 \max\cbra*{ \lambda_{\max}\rbra[\Big]{\frac{\rho_0-\rho_1}{2}}, \lambda_{\max}\rbra[\Big]{\frac{\rho_1-\rho_0}{2}} } 
        = 2 \td_{\infty}(\rho_0,\rho_1).
    \end{align*}
    Here, the first line follows from \Cref{eq:T-positive-supports}, and the second line is a consequence of the definitions of $\varsigma_0$ and $\varsigma_1$. 

    To prove the upper bound, by symmetry, assume without loss of generality that $\rank(\rho_0)\leq \rank(\rho_1)$. Consider the spectral decomposition of $\varsigma_0$ and $\varsigma_1$, where $\cbra*{\lambda_{b,i}}_{1\leq i\leq r_b}$ denotes the set of eigenvalues of $\varsigma_b$ for $b\in\binset$ and $\cbra*{\ket{v_{b,i}}}_{1\leq i \leq r_b}$ denotes an orthonormal eigenbasis for the support of $\varsigma_b$:
    \begin{equation}
        \label{eq:T-spectral-decomp}
        \varsigma_0 = \sum_{i=1}^{r_0} \lambda_{0,i} \ketbra{v_{0,i}}{v_{0,i}} \quad \text{and} \quad \varsigma_1 = \sum_{i=1}^{r_1} \lambda_{1,i} \ketbra{v_{1,i}}{v_{1,i}}. 
    \end{equation}
    Here, $r_b \coloneqq \rank(\varsigma_b)$ satisfies $r_b \leq \rank(\rho_b)$ for $b\in\binset$ by \Cref{eq:stateDiff-rank-bound}. The inequality arises because, intuitively, the support of $\varsigma_b$ consists of those supports of $\rho_b$ that is ``larger than'' $\rho_{1-b}$.
    Substituting \Cref{eq:T-spectral-decomp} into \Cref{eq:T-positive-supports}, it follows that: 
    \begin{subequations}
    \label{eq:T-vs-Tinf-upperBound}
    \begin{align}
        \td(\rho_0,\rho_1) = \tr(\varsigma_0) 
        = \sum_{i=1}^{r_0} \lambda_{0,i}
        &\leq \rank(\rho_0) \max\cbra*{\lambda_{\max}(\varsigma_0),\lambda_{\max}(\varsigma_1)}\\
        &= 2\rank(\rho_0) \lambda_{\max}\rbra*{\frac{\abs{\rho_0-\rho_1}}{2}} = 2\rank(\rho_0) \td_{\infty}(\rho_0,\rho_1). 
    \end{align}
    \end{subequations}

    Finally, by combining \Cref{eq:T-vs-Tinf-upperBound} with symmetry, we obtain the desired bound: 
    \[ \td(\rho_0,\rho_1) \leq 2\min\cbra*{\rank(\rho_0),\rank(\rho_1)} \cdot \td_{\infty}(\rho_0,\rho_1). \qedhere\]
\end{proof}

\subsection{Rank-dependent inequalities between trace distance and \texorpdfstring{$\Lambda_\alpha$}{} for \texorpdfstring{$0<\alpha<1$}{}}
\label{subsec:inequalities-T-vs-Lalpha}

We establish a rank-dependent inequality between the quantum Schatten $\alpha$-power distance and the trace distance for $\alpha \in (0,1)$, analogous to \Cref{thm:Talpha-vs-T}: 

\begin{theorem}[$\Lalpha$ vs.~$\td$ for $0<\alpha<1$]
    \label{thm:Lalpha-vs-T-less-1}
    Let $\rho_0$ and $\rho_1$ be quantum states. Then, for any $\alpha$ in the range $0<\alpha<1$, the following holds: 
    \[
        \td(\rho_0,\rho_1)^{\alpha}
        \leq \Lalpha(\rho_0,\rho_1)
        \leq \frac{\rank(\rho_0)^{1-\alpha}+\rank(\rho_1)^{1-\alpha}}{2}\,\td(\rho_0,\rho_1)^{\alpha}.
    \]
\end{theorem}

By taking the maximum rank of $\rho_0$ and $\rho_1$, we obtain a simplified form: 
\begin{corollary}[$\Lalpha$ vs.~$\td$, simplified]
    \label{corr:Lalpha-vs-T-simplified-less-1}
    For any quantum states $\rho_0$ and $\rho_1$,
    \[
        \forall \alpha\in(0,1), \quad
        \td(\rho_0,\rho_1)^{\alpha}
        \leq \Lalpha(\rho_0,\rho_1)
        \leq \max\cbra*{\rank(\rho_0),\rank(\rho_1)}^{1-\alpha}\,\td(\rho_0,\rho_1)^{\alpha}.
    \]
\end{corollary}

Furthermore, for pure quantum states, \Cref{thm:Lalpha-vs-T-less-1} implies the following identity: 
\begin{corollary}[$\td^\alpha=\Lalpha$ for pure states]
    \label{corr:Lalpha-vs-T-pure}
    For any pure states $\ketbra{\psi_0}{\psi_0}$ and $\ketbra{\psi_1}{\psi_1}$, we have
    \[ \forall \alpha\in(0,1), \quad \td(\ketbra{\psi_0}{\psi_0},\ketbra{\psi_1}{\psi_1})^\alpha = \Lalpha(\ketbra{\psi_0}{\psi_0},\ketbra{\psi_1}{\psi_1}). \]
\end{corollary}
We now continue with the proof of \Cref{thm:Lalpha-vs-T-less-1}:

\begin{proof}[Proof of \Cref{thm:Lalpha-vs-T-less-1}]
    Our construction is similar to the proof of \Cref{thm:Talpha-vs-T}. We start by defining the positive semi-definite matrices $\varsigma_0$ and $\varsigma_1$ as in \Cref{eq:pos-and-neg-in-stateDiff}, with orthogonal supports and satisfying $\rho_0-\rho_1=\varsigma_0-\varsigma_1$ and $\tr(\varsigma_b)=\td(\rho_0,\rho_1)$ for $b\in\binset$, as in \Cref{eq:T-positive-supports}. Since the statement trivially holds when $\td(\rho_0,\rho_1)=0$, namely when $\rho_0=\rho_1$, we can without loss of generality assume that $\td(\rho_0,\rho_1)>0$. Then, by defining normalized states $\widetilde{\varsigma}_b \coloneqq \varsigma_b / \tr(\varsigma_b)$ and applying \Cref{eq:Talpha-vs-T-expression}, we obtain
    \begin{equation}
        \label{eq:Lalpha-expression-less-1}
        2\Lalpha(\rho_0,\rho_1)
        = \tr\rbra*{\abs*{\rho_0-\rho_1}^{\alpha}}
        = \tr\rbra*{\varsigma_0^{\alpha}} + \tr\rbra*{\varsigma_1^{\alpha}}
        = \td(\rho_0,\rho_1)^{\alpha} \sum_{b\in\binset} \tr\rbra*{\widetilde{\varsigma}_b^{\alpha}}.
    \end{equation}

    It therefore suffices to bound $\tr(\widetilde{\varsigma}_b^{\alpha})$ for $b\in\binset$. We consider only the case $b=0$ with $\widetilde{\varsigma}_0$, as the case $b=1$ follows similarly.  Let $\lambda_1,\dots,\lambda_r$ be the nonzero eigenvalues of the quantum state $\widetilde{\varsigma}_0$ of rank $r\coloneqq\rank(\widetilde{\varsigma}_0)$. Then, the following bound holds: 
    \begin{equation}
        \label{eq:Lalpha-vs-T-bounds}
        \forall b\in\binset, \quad 1 \leq \tr\rbra*{\widetilde{\varsigma}_b^\alpha} \leq \rank(\widetilde{\varsigma}_b)^{1-\alpha} \leq \rank(\rho_b)^{1-\alpha}. 
    \end{equation}
    
    Noting that $\tr\rbra*{\widetilde{\varsigma}_0^\alpha} = \sum_{i\in[r]}\lambda_i^\alpha$, to show the case $b=0$ in \Cref{eq:Lalpha-vs-T-bounds}, the first inequality follows by combining the facts that 
    \[ 1 = 1^\alpha = \rbra[\bigg]{\sum_{i\in[r]} \lambda_i}^\alpha \quad\text{and}\quad \rbra[\bigg]{\sum_{i\in[r]} \lambda_i}^\alpha \leq \sum_{i\in[r]}\lambda_i^\alpha \text{ for } \alpha\in(0,1);\]
    the second inequality follows from Jensen's inequality for the concave function $x\mapsto x^{\alpha}$, namely 
    \[ \frac{1}{r}\sum_{i\in[r]} \lambda_i^{\alpha} \leq \rbra[\bigg]{\frac{1}{r}\sum_{i\in[r]} \lambda_i}^{\alpha} = \rbra*{\frac{1}{r}}^{\alpha};\]
    and the third inequality follows from \Cref{eq:stateDiff-rank-bound} and the fact that the function $x^{1-\alpha}$ is monotonically non-decreasing for $\alpha \in (0,1)$.

    Finally, we complete the proof by substituting \Cref{eq:Lalpha-vs-T-bounds} into \Cref{eq:Lalpha-expression-less-1}, which yields the desired inequality. 
\end{proof}

\section{Computational hardness and lower bounds for fixed \texorpdfstring{$\alpha$}{}}
\label{sec:easy-regime-hardness}

We begin by introducing generalizations of \QSD{} from~\cite{Wat02}, in which the trace distance is replaced by quantum Schatten norm and power distances as the closeness measures: 

\begin{definition}[Quantum State Distinguishability Problem with Schatten $\alpha$-norm, \QSDalpha{}]
	\label{def:QSDalpha}
    Let $Q_0$ and $Q_1$ be quantum circuits acting on $m$ qubits \emph{(}``input length''\emph{)} and having $n$ specified output qubits \emph{(}``output length''\emph{)}, where $m(n)$ is a polynomial function of $n$. Let $\rho_i$ denote the quantum state obtained by running $Q_i$ on state $\ket{0}^{\otimes m}$ and tracing out the non-output qubits. For efficiently computable functions $a(n)$ and $b(n)$ satisfying $0\leq b(n) < a(n) \leq 1$, the promise problem \QSDalpha$[a,b]$ asks one to decide whether: 
	\begin{itemize}[topsep=0.33em, itemsep=0.33em, parsep=0.33em]
		\item \emph{Yes:} A pair of quantum circuits $(Q_0,Q_1)$ such that $\Talpha(\rho_0,\rho_1) \geq a(n)$; 
		\item \emph{No:} A pair of quantum circuits $(Q_0,Q_1)$ such that $\Talpha(\rho_0,\rho_1) \leq b(n)$.
	\end{itemize}

    \noindent Moreover, we denote the restricted version, where $\rho_0$ and $\rho_1$ are pure states, as \PureQSDalpha{}. 
\end{definition}

\begin{definition}[Quantum State Distinguishability Problem with Schatten $\alpha$-power, \PoweredQSDalpha]
    \label{def:PoweredQSDalpha-less-1}
    Let $Q_0$ and $Q_1$ be quantum circuits acting on $m$ qubits (``input length'') and having $n$ specified output qubits (``output length''), where $m(n)$ is polynomial in $n$. Let $\rho_0$ and $\rho_1$ denote the output states obtained by running $Q_0$ and $Q_1$ on $\ket{0}^{\otimes m}$ and tracing out the non-output qubits. For efficiently computable functions $a(n)$ and $b(n)$ satisfying $0\leq b(n) < a(n) \leq 1$, the promise problem \PoweredQSDalpha$[a,b]$ asks one to decide whether:
    \begin{itemize}[topsep=0.33em, itemsep=0.33em, parsep=0.33em]
        \item \emph{Yes:} The pair of quantum circuits $(Q_0,Q_1)$ such that $\Lalpha(\rho_0,\rho_1) \geq a(n)$;
        \item \emph{No:} The pair of quantum circuits $(Q_0,Q_1)$ such that $\Lalpha(\rho_0,\rho_1) \leq b(n)$.
    \end{itemize}
\end{definition}

In the remainder of this section, we first present computational hardness results for \PureQSDalpha{} with $1\leq\alpha\leq\infty$ and \PoweredQSDalpha{} with $0<\alpha<1$ in \Cref{subsec:comp-hardness-fixed}. We then establish quantum query and sample lower bounds for the corresponding estimation tasks in \Cref{subsec:lower-bounds-fixed}.

\subsection{Computational hardness for fixed \texorpdfstring{$\alpha$}{}}
\label{subsec:comp-hardness-fixed}

We first establish the hardness result of \PureQSDalpha{} with $1 \leq \alpha \leq \infty$:

\begin{theorem}[\PureQSDalpha{} is \BQP{}-hard]
    \label{thm:PureQSDalpha-is-BQPhard}
    For any $1 \leq \alpha \leq \infty$ and $n \geq 2$, it holds that:
    \[ \PureQSDalpha\sbra*{2^{\frac{1}{\alpha}-1} \cdot \rbra*{1-2^{-n}},~ 2^{\frac{1}{\alpha}-1-n}} \text{ is } \BQP\text{-hard}. \]
\end{theorem}

\begin{proof}
    Using \Cref{lemma:PureQSD-is-BQPhard}, it follows that $\PureQSD\big[ 1-2^{-n}, 2^{-n} \big]$ is \BQP{}-hard for $n\geq 2$. Let $Q_0$ and $Q_1$ be the corresponding \BQP{}-hard instance, where these quantum circuits are of polynomial size and prepare the pure states $\ketbra{\psi_0}{\psi_0}$ and $\ketbra{\psi_1}{\psi_1}$, respectively. 
    Using the same hard instance $(Q_0,Q_1)$, we can derive the following from the equality in \Cref{corr:Talpha-vs-T-pure}: 
    \begin{itemize}[topsep=0.33em, itemsep=0.33em, parsep=0.33em]
        \item For \textit{yes} instances, $\td\rbra*{\ketbra{\psi_0}{\psi_0}, \ketbra{\psi_1}{\psi_1}} \geq 1-2^{-n}$ implies that
        \begin{equation}
            \label{eq:PureQSDalpha-BQPhard-yes}
            \Talpha\rbra*{\ketbra{\psi_0}{\psi_0}, \ketbra{\psi_1}{\psi_1}} 
            = 2^{\frac{1}{\alpha}-1} \cdot \td\rbra*{\ketbra{\psi_0}{\psi_0}, \ketbra{\psi_1}{\psi_1}} 
            \geq 2^{\frac{1}{\alpha}-1} \cdot \rbra*{1-2^{-n}}.
        \end{equation}
        \item For \textit{no} instances, $\td\rbra*{\ketbra{\psi_0}{\psi_0}, \ketbra{\psi_1}{\psi_1}} \leq 2^{-n}$ yields that 
        \begin{equation}
            \label{eq:PureQSDalpha-BQPhard-no}
            \Talpha\rbra*{\ketbra{\psi_0}{\psi_0}, \ketbra{\psi_1}{\psi_1}} 
            = 2^{\frac{1}{\alpha}-1} \cdot \td\rbra*{\ketbra{\psi_0}{\psi_0}, \ketbra{\psi_1}{\psi_1}} 
            \leq 2^{\frac{1}{\alpha}-1-n}.
        \end{equation}
    \end{itemize}

    Combining \Cref{eq:PureQSDalpha-BQPhard-yes,eq:PureQSDalpha-BQPhard-no}, it follows that the promise gap
    \[ 2^{\frac{1}{\alpha}-1} \cdot \rbra*{1-2^{-n}} - 2^{\frac{1}{\alpha}-1-n}
    = 2^{\frac{1}{\alpha}-1} \cdot \rbra*{1-2^{-n} - 2^{-n}}
    \coloneqq 2^{\frac{1}{\alpha}-1} \cdot f(n)
    \geq 2^{\frac{1}{\alpha}-1} \cdot f(2).\]
    Here, the last inequality holds because $f(n)$ is a monotonically increasing function. We complete the proof by observing that $f(2) = \frac{1}{2}$ and $2^{\frac{1}{\alpha}-1} \geq \frac{1}{2}$ for all $1 \leq \alpha \leq \infty$. 
\end{proof}

\begin{remark}[\PurePoweredQSDalpha{} is \BQP{}-hard for $0<\alpha<1$]
    Similar to the proof of \Cref{thm:PureQSDalpha-is-BQPhard}, combining \Cref{lemma:PureQSD-is-BQPhard} with \Cref{corr:Lalpha-vs-T-pure} implies that $\PurePoweredQSDalpha\sbra*{1-2^{-n}, 2^{-\alpha n}}$ is \BQP{}-hard for every $0<\alpha<1$ and $n\geq 2$.\footnote{Indeed, for \emph{yes} instances, it holds that $\Lalpha\rbra*{\ketbra{\psi_0}{\psi_0}, \ketbra{\psi_1}{\psi_1}} = \td\rbra*{\ketbra{\psi_0}{\psi_0}, \ketbra{\psi_1}{\psi_1}}^\alpha \geq \rbra*{1-2^{-n}}^\alpha \geq 1-2^{-n}$, since $(1-x)^\alpha \geq 1-x$ for all $x\in[0,1]$ and $0<\alpha <1$. For \emph{no} instances, it follows that $\Lalpha\rbra*{\ketbra{\psi_0}{\psi_0}, \ketbra{\psi_1}{\psi_1}} = \td\rbra*{\ketbra{\psi_0}{\psi_0}, \ketbra{\psi_1}{\psi_1}}^\alpha \leq 2^{-\alpha n}$.} Together with any rank-efficient quantum query algorithm for \PoweredQSDalpha{}, such as~\cite[Theorem IV.1]{WGL+22} or \Cref{thm:alpha-quasi-query}, the above hardness result implies that $\textsc{LowRankPoweredQSD}_{\alpha}$ is \BQP{}-complete for $0<\alpha<1$. 
\end{remark}

Next, we show the hardness result of \PoweredQSDalpha{} with $0 < \alpha < 1$:

\begin{theorem}[\PoweredQSDalpha{} is \QSZK-hard for $0<\alpha<1$]
    \label{thm:PoweredQSDalpha-QSZK-hard-less-1}
    For any constant $\alpha\in(0,1)$, the following holds for sufficiently large $n$:     
    \[ \PoweredQSDalpha\sbra*{1/3,2^{-n}} \text{ is } \QSZK\text{-hard}. \]
\end{theorem}

\begin{proof}
    We begin by defining $\delta_n \coloneqq 2^{-n^{1/4}}$. By \Cref{lemma:QSD-is-QSZKhard}, it follows that $\QSD\sbra*{1-\delta_n,\delta_n}$ is \QSZK{}-hard. Let $Q_0$ and $Q_1$ be a corresponding \QSZK{}-hard instance, where these quantum circuits are of polynomial size and prepare the quantum states $\rho_0$ and $\rho_1$, respectively.

    We apply the polarization lemma for $\td$ (\Cref{lemma:polarization-trace-distance}) to $(Q_0,Q_1,a,b,k)$ with the parameters $a=1-\delta_n$, $b=\delta_n$, and $k=n^2$. The required condition $a^2>b$ holds for sufficiently large $n$. Let $\sigma_0$ and $\sigma_1$ be the resulting states, acting on $N_n$ qubits, whose state-preparation circuits are $\widetilde{Q}_0$ and $\widetilde{Q}_1$, respectively. These states satisfy
    \begin{align}
        \td(\rho_0,\rho_1) \geq 1-\delta_n &\quad\Longrightarrow\quad \td(\sigma_0,\sigma_1) \geq 1-2^{-n^2},
        \label{eq:trace-pol-yes-less-1}\\
        \td(\rho_0,\rho_1) \leq \delta_n &\quad\Longrightarrow\quad \td(\sigma_0,\sigma_1) \leq 2^{-n^2}.
        \label{eq:trace-pol-no-less-1}
    \end{align}

    We first bound $N_n$. As stated in \Cref{lemma:polarization-trace-distance}, $N_n = \widetilde O\rbra*{n(n^2)^{\eta_n}}$, where 
    \[ \eta_n = O\rbra*{\frac{\delta_n\ln\rbra*{2/(1-\delta_n)^2}}{(1-\delta_n)^2-\delta_n}} \leq  O\rbra*{2^{-n^{1/4}}}. \]
    Here, the last inequality follows from $\delta_n=2^{-n^{1/4}}$. 
    
    Consequently, for sufficiently large $n$, we obtain
    \begin{equation}
        \label{eq:Nn-bound-less-1}
        N_n = \widetilde{O}\rbra*{n \cdot (n^2)^{\eta_n}} = \widetilde{O}\rbra*{n \cdot \exp\rbra*{2\eta_n\ln n}} \leq  \widetilde{O}\rbra*{n \cdot 2} = \widetilde{O}(n).
    \end{equation}

    Next, for the resulting hard instance $\rbra[\big]{\widetilde{Q}_0,\widetilde{Q}_1}$, we obtain the following bounds: 
    \begin{itemize}
        \item For \emph{yes} instances, combining \Cref{eq:trace-pol-yes-less-1} with the lower bound in \Cref{thm:Lalpha-vs-T-less-1} yields
        \[ \Lalpha(\sigma_0,\sigma_1)
            \geq \td(\sigma_0,\sigma_1)^{\alpha}
            \geq \rbra*{1-2^{-n^2}}^{\alpha}
            \geq 1-2^{-n^2}. \]
        In particular, $\Lalpha(\sigma_0,\sigma_1)\geq 1/3$ for sufficiently large $n$.
        \item For \emph{no} instances, substituting \Cref{eq:trace-pol-no-less-1} into the upper bound in \Cref{corr:Lalpha-vs-T-simplified-less-1} gives 
        \begin{align*}
            \Lalpha(\sigma_0,\sigma_1)
            \leq 2^{N_n(1-\alpha)}\td(\sigma_0,\sigma_1)^{\alpha}
            &\leq 2^{N_n(1-\alpha)}2^{-\alpha n^2}\\
            &\leq 2^{-\alpha n^2/2}\\
            &\leq 2^{-N_n}. 
        \end{align*}
        Here, the second line follows because the exponent $N_n(1-\alpha)-\alpha n^2$ is at most $-\alpha n^2/2$ for sufficiently large $n$, as guaranteed by \Cref{eq:Nn-bound-less-1}; and the last line follows from $N_n=\widetilde O(n)=o(n^2)$ as implied by  \Cref{eq:Nn-bound-less-1}. 
    \end{itemize}

    Therefore, the mapping $(Q_0,Q_1)\mapsto (\widetilde Q_0,\widetilde Q_1)$ is a Karp reduction from $\QSD\sbra*{1-\delta_n,\delta_n}$, which is \QSZK{}-hard, to  an instance of $\PoweredQSDalpha\sbra*{1/3,2^{-N_n}}$. Renaming the output length $N_n$ as $n$ gives the claimed hardness of $\PoweredQSDalpha\sbra*{1/3,2^{-n}}$. 
\end{proof}

\subsection{Lower bounds for fixed \texorpdfstring{$\alpha$}{}}
\label{subsec:lower-bounds-fixed}

We first prove quantitative lower bounds for \PureQSDalpha{} with $1 \leq \alpha \leq \infty$:

\begin{theorem}[Quantitative lower bounds for \PureQSDalpha{} when $\alpha \geq 1$]
    \label{thm:lower-bounds-PureQSDalpha}
    For any $1 \leq \alpha \leq \infty$ and $0 < \epsilon < 1/4$, there exist $n$-qubit pure states $\ket{\psi_0}$ and $\ket{\psi_1}$ such that deciding whether $\Talpha(\ketbra{\psi_0}{\psi_0},\ketbra{\psi_1}{\psi_1})$ is at least $\epsilon$ or exactly $0$ requires:
    \begin{enumerate}[label={\upshape(\arabic*)},topsep=0.33em, itemsep=0.33em, parsep=0.33em]
        \item \label{thmitem:easy-regime-queries}\textbf{\emph{Queries}}: In the purified quantum access model, the quantum query complexity is $\Omega(1/\epsilon)$. 
        \item \label{thmitem:easy-regime-samples}\textbf{\emph{Samples}}: The quantum sample complexity is $\Omega(1/\epsilon^2)$. 
    \end{enumerate}
\end{theorem}

\begin{proof}
    To establish the desired quantum query (or sample) lower bound, 
    it suffices to reduce the problem to distinguishing between the cases $\ket{\psi_0} = \ket{\psi_1}$ and $\td(\ketbra{\psi_0}{\psi_0},\ketbra{\psi_1}{\psi_1})\geq 2\epsilon$, as stated in \Cref{lemma:lower-bounds-PureQSD}\ref{thmitem:query-lower-bound-PureQSD} (or \Cref{lemma:lower-bounds-PureQSD}\ref{thmitem:sample-lower-bound-PureQSD}).
    Since the quantum Schatten $\alpha$-norm distance $\Talpha(\cdot,\cdot)$ is a metric, it follows that $\Talpha(\ketbra{\psi_0}{\psi_0},\ketbra{\psi_1}{\psi_1})=0$ when $\ket{\psi_0} = \ket{\psi_1}$. When the pure states $\ket{\psi_0}$ and $\ket{\psi_1}$ are far apart, we use the equality in \Cref{corr:Talpha-vs-T-pure} to complete the proof: 
    \[ \Talpha(\ketbra{\psi_0}{\psi_0},\ketbra{\psi_1}{\psi_1}) = 2^{\frac{1}{\alpha}-1} \cdot \td(\ketbra{\psi_0}{\psi_0},\ketbra{\psi_1}{\psi_1}) \geq 2^{\frac{1}{\alpha}-1} \cdot 2\epsilon \geq \epsilon. \qedhere\]
\end{proof}

Next, we establish the quantitative lower bounds for \PoweredQSDalpha{} with $0<\alpha<1$:

\begin{theorem}[Quantitative lower bounds for \PoweredQSDalpha{} when $0<\alpha<1$]
    \label{thm:lower-bounds-PoweredQSDalpha}
    For any $0 < \alpha < 1$ and $0 < \epsilon \leq 1/4$, there exist $n$-qubit rank-$r$ quantum states $\rho$ and $\rho_\ttU$, where $r\geq 4$ and $\rho_\ttU$ denotes the state that is uniform on the support of $\rho$, such that estimating $\Lalpha\rbra{\rho, \rho_\ttU}$ to within additive error $\epsilon$ requires: 
    \begin{enumerate}[label={\upshape(\arabic*)},topsep=0.33em, itemsep=0.33em, parsep=0.33em]
        \item \label{thmitem:queries-0leAle1}\textbf{\emph{Queries}}: In the purified quantum access model, the quantum query complexity is $\Omega\rbra[\big]{r^{\frac{1}{\alpha}-1}/\epsilon^{\frac{1}{\alpha}}}$. 
        \item \label{thmitem:samples-0leAle1}\textbf{\emph{Samples}}: The quantum sample complexity is $\Omega\rbra[\big]{r^{\frac{2}{\alpha}-1}/\epsilon^{\frac{2}{\alpha}}}$.
    \end{enumerate}
\end{theorem}

\begin{proof}
     Since applying quantum sample-to-query lifting (\Cref{lemma:lifting}) to \Cref{thmitem:samples-0leAle1} directly implies \Cref{thmitem:queries-0leAle1}, it suffices to prove \Cref{thmitem:samples-0leAle1}. To this end, we consider the parameter $\delta$ defined by
    \[
    \delta = \frac{r}{2} \rbra*{\frac{4\epsilon}{r}}^{1/\alpha} \leq \frac{1}{2}.
    \]

    By \Cref{lemma:lower-bounds-QSD}\ref{thmitem:sample-lower-bound-QSD}, for any projector $\Pi_r$ onto a subspace of dimension $r$, distinguishing between the following two cases for an unknown state $\rho$, promised to satisfy $\supp\rbra{\rho} \subseteq \supp\rbra{\Pi_r}$, requires $\Omega\rbra{r/\delta^2}$ samples of $\rho$: either $\rho = \rho_{\ttU}$, or $\rho$ has eigenvalues $(1\pm 2\delta)/r$, each with multiplicity $r/2$. Here, $\rho_{\ttU} = \Pi_r / r$.
    We now evaluate $\Lalpha\rbra{\rho,\rho_{\ttU}}$ in these two cases:
    \begin{enumerate}[label={\upshape(\alph*)}]
        \item If $\rho = \rho_{\ttU}$, then $\Lalpha\rbra{\rho, \rho_{\ttU}} = 0$. 
        \item If $\rho$ has eigenvalues $(1\pm 2\delta)/r$ each with multiplicity $r/2$, then 
        \[
        2\Lalpha(\rho,\rho_{\ttU}) = \tr\rbra*{\abs*{\rho - \rho_{\ttU}}^{\alpha}} = \rbra*{\frac{2\delta}{r}}^{\alpha} r = 4 \epsilon. 
        \]
    \end{enumerate}

    Therefore, any quantum algorithm that estimates $\Lalpha(\rho,\rho_{\ttU})$ to within additive error $\epsilon$ distinguishes the aforementioned two cases and hence solves the promise problem, which requires sample complexity 
    \[
    \Omega\rbra*{\frac{r}{\delta^2}} = \Omega\rbra*{\frac{r^{\frac{2}{\alpha}-1}}{\epsilon^{\frac{2}{\alpha}}}}.
     \qedhere \]
\end{proof}

\section{Quantum Schatten \texorpdfstring{$\alpha$}{}-norm distance estimation for non-fixed \texorpdfstring{$\alpha>1$}{alpha > 1} near \texorpdfstring{$1$}{1}}

In this section, we prove that \QSDalpha{} is \QSZK{}-complete for $\alpha\geq 1$ negligibly close to $1$, extending the prior result that \QSD{} ($\alpha=1$) is \QSZK{}-complete, as shown in~\cite{Wat02}: 

\begin{theorem}[\QSDalpha{} is \QSZK{}-complete for $\alpha > 1$ near $1$]
    \label{thm:QSDalpha-hard-regime-QSZK-complete}
    Let $a(n)$ and $b(n)$ be efficiently computable functions such that $0 \leq b(n) < a(n) \leq 1$. Then, the following statements hold:
    \begin{enumerate}[label={\upshape(\arabic*)},topsep=0.33em, itemsep=0.33em, parsep=0.33em]
        \item \label{thmitem:QSZK-containment}For every negligible function $\mu(n)$ and every $\alpha$ satisfying $1\leq \alpha(n) \leq 1+\mu(n)$, $\QSDalpha[a,b]$ is in \QSZK{} when $a(n)^2-b(n) \geq 1/O(\log{n})$.
        \item For every $1 \leq \alpha(n) \leq 1+\frac{1}{n}$, $\QSDalpha[a,b]$ is \QSZK{}-hard if $a(n) \leq 1/2 - 2^{-n^{\tau}-1}$ and $b(n) \geq 2^{-n^{\tau}-\frac{1}{n+1}}$ for every constant $\tau\in(0,1/2)$ and every sufficiently large integer $n$. 
    \end{enumerate} 
\end{theorem}

The main challenge in proving \Cref{thm:QSDalpha-hard-regime-QSZK-complete} is to establish a \QSZK{} containment of $\QSDalpha$ under the polarizing regime $a(n)^2-b(n) \geq 1/O(\log{n})$.\footnote{Notably, similar to the classical cases in~\cite{BDRV19}, by reducing to the \textsc{Quantum Jensen-Shannon Divergence Problem} (\QJSP{}) or the \textsc{Measured Quantum Triangular Discrimination Problem} (\measQTDP{}) introduced in~\cite{Liu23}, this \QSZK{} containment of \QSDalpha{} can be extended slightly beyond the polarizing regime.}
A direct approach, combining the inequalities between $\td$ and $\Talpha$ (\Cref{corr:Talpha-vs-T-simplified}) with the \QSZK{} containment of \QSD{} from~\cite{Wat02,Wat09}, only yields a \QSZK{} containment of $\QSDalpha[a,b]$ under a \textit{weaker} regime, $a(n)^2/2-b(n)\geq 1/O(\log{n})$. To circumvent this, we provide a (partial) polarization lemma for $\Talpha$ (\Cref{lemma:partial-polarization-Talpha}), which enables us to achieve the desired \QSZK{} containment in \Cref{thm:QSDalpha-hard-regime-QSZK-complete}. 

\vspace{1em}
The remainder of this section establishes the \QSZK{} containment of \QSDalpha{} in \Cref{subsec:QSDalpha-nearOne-in-QSZK} using the partial polarization lemma for $\Talpha$ (\Cref{lemma:partial-polarization-Talpha}). We then show the \QSZK{} hardness of \QSDalpha{} (\Cref{thm:QSDalpha-QSZK-hard}) and derive quantitative lower bounds on query complexity (\Cref{thm:query-lower-bounds-QSDalpha}) and sample complexity (\Cref{thm:sample-lower-bound-QSDalpha}) in \Cref{subsec:QSDalpha-hard-regime-hardness}.

\subsection{\QSZK{} containment via a partial polarization lemma for \texorpdfstring{$\Talpha$}{}}
\label{subsec:QSDalpha-nearOne-in-QSZK}

We now establish a stronger version of \Cref{thm:QSDalpha-hard-regime-QSZK-complete}\ref{thmitem:QSZK-containment}, where the regime of $\alpha$ is extended from $[1,1+\mu(n)]$ for every negligible function $\mu(n)$ to $[1,1+1/p_C(n)]$ for some large polynomial $p_C(n)$, depending on a constant $C$ such that $a(n)^2-b(n) \geq 1/(C \log n)$:

\begin{theorem}[\QSDalpha{} is in \QSZK{}]
    \label{thm:QSDalpha-in-QSZK}
    Let $a(n)$ and $b(n)$ be efficiently computable functions satisfying $0 \leq b(n) < a(n) \leq 1$. Assume that $a(n)^2-b(n) \geq 1/(C\log{n})$ for some constant $C>1$ and all sufficiently large $n$, and define $p_C(n) \coloneqq 11 C^2 n^{2C} \log{n}$.
    Then, the following holds:
    \[ \text{For any } \alpha\in \sbra*{1, 1+\frac{1}{p_C(n)}},\quad \QSDalpha[a,b] \text{ is in } \QSZK. \]
\end{theorem}

To prove \Cref{thm:QSDalpha-in-QSZK}, we establish a key technical tool -- a \textit{partial} polarization lemma for $\Talpha$ that ensures any $(a,b)$ within the polarizing regime becomes \textit{constantly} separated: 

\begin{lemma}[A partial polarization lemma for $\Talpha$]
    \label{lemma:partial-polarization-Talpha}
    Let $Q_0$ and $Q_1$ be quantum circuits that prepare quantum states $\rho_0$ and $\rho_1$, respectively. There exists a deterministic procedure that, given an input $(Q_0,Q_1,a,b,k)$ where $a(n)^2-b(n)\geq 1/ (C\log{n})$ for some constant $C>1$ and all sufficiently large $n$, outputs new quantum circuits $\widetilde{Q}_0$ and $\widetilde{Q}_1$ that prepare the states $\widetilde{\rho}_0$ and $\widetilde{\rho}_1$, respectively. 
    Let $\lambda \coloneqq \min\cbra{a^2/b,2}$, $\ell \coloneqq \ceil*{\log_\lambda(64k)}$, and $m \coloneqq \ceil*{\lambda^\ell/\rbra{32 a^{2\ell}}}$. 
    Then the resulting states satisfy the following:
    \begin{align*}
        \text{For any } \alpha \in\sbra*{1,1+\frac{1}{n\ell m}}, \quad  \Talpha(\rho_0,\rho_1) \geq a &\quad\Longrightarrow\quad \Talpha\rbra*{\widetilde{\rho}_0,\widetilde{\rho}_1} \geq \frac{1}{2} - \frac{1}{2} e^{-k},\\
        \Talpha(\rho_0,\rho_1) \leq b &\quad\Longrightarrow\quad \Talpha(\widetilde{\rho}_0,\widetilde{\rho}_1) \leq \frac{1}{16}.
    \end{align*}
    Here, the states $\widetilde{\rho}_0$ and $\widetilde{\rho}_1$ are defined over $\widetilde{O}\rbra[\Big]{n k^{O\rbra[\big]{\frac{b\ln(2/a^2)}{a^2-b}}}}$ qubits. Moreover, when $k \leq O(1)$ or $a^2-b \geq \Omega(1)$, the time complexity of the procedure is polynomial in the size of $Q_0$ and $Q_1$, $k$, and $\exp\rbra[\big]{\frac{b\ln(1/a^2)}{a^2-b}}$.
\end{lemma}

\begin{proof}[Proof of \Cref{thm:QSDalpha-in-QSZK}]
    Let $(Q_0,Q_1)$ be a $\QSDalpha[a,b]$ instance, with $1 \leq \alpha \leq 1+\frac{1}{p_C(n)}$, satisfying $a(n)^2-b(n) \geq 1/(C\log{n})$. Let $k\coloneqq 2$, $\lambda \coloneqq \min\cbra*{a^2/b,2}$, $\ell \coloneqq \ceil*{\log_\lambda(64k)} = \ceil*{7 \log_\lambda(2)}$, and $m \coloneqq \ceil*{\lambda^\ell/\rbra*{32 a^{2\ell}}}$. We first verify that $n\ell m \leq p_C(n)$, so that $\alpha \leq 1+\frac{1}{p_C(n)} \leq 1+\frac{1}{n\ell m}$ and the partial polarization lemma for $\Talpha$ (\Cref{lemma:partial-polarization-Talpha}) is applicable with the chosen parameters. 

    To this end, we bound $\ell$ and $m$ separately. A direct calculation shows the following bound for $\ell$, which holds for all sufficiently large $n$, 
    \begin{subequations}
    \label{eq:QSZK-containment-ell-bound}
    \begin{align}
        \ell = \ceil*{7\log_\lambda(2)} &\leq \frac{7 \ln{2}}{\ln\rbra*{ \min\cbra*{1+ (a^2-b)/b, 2} }} + 1\\
        &\leq \frac{7 \ln{2}}{\ln\rbra*{ 1 + 1/\rbra*{C\log{n}} }} + 1\\
        &\leq 7 \ln{2} \cdot 2C \log{n} + 1 \leq 11C \log{n}.
    \end{align}    
    \end{subequations}
    Here, the second line uses the bound $a^2-b \geq 1/(C \log n)$, and the last line follows from the facts that $\ln(1+x) \geq x/2$ for all $0\leq x \leq 1$ and $C\log{n} \geq 1$ for sufficiently large $n$. 

    To bound $m$, we first consider the case where $a^2/b > \lambda=2$, then we obtain $\ell=7$ and
    \begin{equation}
        \label{eq:QSZK-containment-m-bound-lambda=2}
        m = \ceil*{\frac{2^\ell}{32a^{2\ell}}}
        = \ceil*{\frac{4}{a^{14}}}
        \leq \frac{8}{a^{14}}
        \leq 8(C\log n)^7
        \leq C n^{2C}.
    \end{equation}
    Here, the first inequality uses the bound $\ceil*{x} \leq 2x$ for all $x \geq 1$, the second inequality follows from $\frac{1}{a^2} = \frac{1}{b+(a^2-b)} \leq \frac{1}{a^2-b} \leq C \log{n}$, and the last inequality holds for all sufficiently large $n$. 

    Next we consider the case where $\lambda=a^2/b \leq 2$, in which case
    \begin{equation}
        \label{eq:QSZK-containment-m-bound-lambda<2-A}
        16 m = 16 \cdot \ceil*{\frac{\lambda^\ell}{32 a^{2\ell}}}
        = 16 \cdot \ceil*{\frac{1}{32 b^\ell}}
        \leq \frac{1}{b^\ell},
    \end{equation}
    Here, the last inequality uses the bound that $\ceil{x}\leq 2x$ for all $x\geq 1$. A direct calculation further implies that
    \begin{subequations}
    \label{eq:QSZK-containment-m-bound-lambda<2-B}
    \begin{align}
        \ell \ln\rbra*{\frac{1}{b}} &\leq \rbra*{\frac{7\ln{2}}{\ln{\lambda}} + 1} \ln\rbra*{\frac{1}{b}}\\
        &\leq 7 \ln{2} \rbra*{\frac{b}{a^2-b}+\frac{1}{2}} \ln\rbra*{\frac{1}{b}} + \ln\rbra*{\frac{1}{b}}\\
        &\leq \rbra*{\frac{7 \ln{2}}{e}} C\log n + \rbra*{1+\frac{7 \ln 2}{2}}\ln(2C\log n)\\
        &\leq \ln\rbra*{16 C n^{2C}}
    \end{align}
    \end{subequations}
    Here, the second line follows from $\frac{1}{\ln{\lambda}} \leq \frac{\lambda+1}{2(\lambda-1)} = \frac{a^2+b}{2(a^2-b)} = \frac{b}{a^2-b}+\frac{1}{2}$, which uses the fact that $\ln{t} \geq \frac{2(t-1)}{t+1}$ for all $t\geq 1$, the third line follows from the facts that $b \ln(1/b) \leq 1/e$ and $b \geq a^2/2 \geq (a^2-b)/2 \geq 1/(2C \log{n})$ since $\lambda = a^2/b \leq 2$, and the last inequality holds for all sufficiently large $n$ since $7 \ln{2}/e \leq 2$. 
    
    Putting \Cref{eq:QSZK-containment-ell-bound,eq:QSZK-containment-m-bound-lambda=2,eq:QSZK-containment-m-bound-lambda<2-A,eq:QSZK-containment-m-bound-lambda<2-B} together, we obtain  
    \[ n\ell m \leq n \cdot 11 C \log n \cdot C n^{2C} = 11 C^2 n^{2C} \log{n} = p_C(n). \]
    
    Consequently, the partial polarization lemma for $\Talpha$ (\Cref{lemma:partial-polarization-Talpha}) allows mapping $(Q_0,Q_1)$ to a $\QSDalpha\sbra[\big]{\frac{1-e^{-2}}{2},\frac{1}{16}}$ instance $(Q'_0,Q'_1)$. Using the inequalities between the trace distance and the quantum Schatten $\alpha$-norm distance (\Cref{corr:Talpha-vs-T-simplified}), we obtain that $(Q'_0,Q'_1)$ also forms a $\QSD\sbra[\big]{\frac{1-e^{-2}}{2},\frac{1}{8}}$ instance. 
    Since $\rbra[\big]{\frac{1-e^{-2}}{2}}^2 > 1/8$, applying the polarization lemma for the trace distance (\Cref{lemma:polarization-trace-distance}) to $(Q'_0,Q'_1,t)$ with $2^{-t(n)}$ negligible produces a $\QSD\sbra[\big]{1-2^{-t(n)},2^{-t(n)}}$ instance $(Q''_0,Q''_1)$. Lastly, following~\cite[Theorem 10]{Wat02}, and particularly the protocol in~\cite[Figure 2]{Wat02}, we conclude that $\QSDalpha\sbra[\big]{a,b}$ is contained in \QSZK{} as desired.
\end{proof}

Analogous to polarization lemmas for various classical~\cite{SV97,CCKV08,BDRV19} and quantum~\cite{Wat02,Liu23} closeness measures, we reduce the errors on both sides of the problem \QSDalpha{} separately, as detailed in \Cref{lemma:XOR-Talpha,lemma:direct-product-Talpha}. 

\begin{lemma}[XOR lemma for $\Talpha$]
    \label{lemma:XOR-Talpha}
    Let $Q_0$ and $Q_1$ be quantum circuits that prepare the quantum states $\rho_0$ and $\rho_1$, respectively. There exists a deterministic procedure that, given $(Q_0,Q_1,\ell)$ as input, produces new quantum circuits $\widetilde{Q}_0$ and $\widetilde{Q}_1$ that prepare the states $\widetilde{\rho}_0$ and $\widetilde{\rho}_1$, respectively. These states are defined as $\widetilde{\rho}_b \coloneqq 2^{-\ell+1} \sum_{b_1\oplus \cdots \oplus b_\ell=b} \rho_{b_1}\otimes \cdots \otimes \rho_{b_\ell}$ for $b\in\binset$, and satisfy:
    \[ \text{For any } \alpha \geq 1, \quad \Talpha\rbra*{\widetilde{\rho}_0,\widetilde{\rho}_1} = \Talpha(\rho_0,\rho_1)^\ell. \]
\end{lemma}

\begin{proof}
    We begin with the case $\ell=2$. Specifically, consider the quantum states 
    \[ \rho'_0 \coloneqq \frac{1}{2} \rho_0\otimes\rho_0 + \frac{1}{2} \rho_1\otimes\rho_1 \text{ and } \rho'_1 \coloneqq \frac{1}{2} \rho_0\otimes\rho_1 + \frac{1}{2} \rho_1\otimes\rho_0. \]
    Then, the following holds: 
    \begin{align*}
        \Talpha(\rho'_0, \rho'_1) &= \frac{1}{2} \rbra*{\tr\abs*{\rho'_0-\rho'_1}^{\alpha}}^{1/\alpha}\\
        &= \frac{1}{2} \rbra*{\tr \abs*{\frac{1}{2} (\rho_0-\rho_1)\otimes (\rho_0-\rho_1)}^{\alpha}}^{1/\alpha}\\
        &= \frac{1}{2} \rbra*{\tr \abs*{\rho_0-\rho_1}^{\alpha}}^{1/\alpha} \cdot \frac{1}{2} \rbra*{\tr\abs*{\rho_0-\rho_1}^{\alpha}}^{1/\alpha}\\
        &= \Talpha(\rho_0,\rho_1)^2.
    \end{align*}
    As a consequence, for the case where $\ell > 2$, we can establish the equality by induction. 
\end{proof}

\begin{lemma}[Direct product lemma for $\Talpha$]
    \label{lemma:direct-product-Talpha}
    Let $Q_0$ and $Q_1$ be quantum circuits that prepare the quantum states $\rho_0$ and $\rho_1$, respectively. There exists a deterministic procedure that, given $(Q_0,Q_1,\ell)$ as input, produces new quantum circuits $\widetilde{Q}_0$ and $\widetilde{Q}_1$ that prepare the states $\widetilde{\rho}_0$ and $\widetilde{\rho}_1$, respectively. These states are defined as $\widetilde{\rho}_b \coloneqq \rho_b^{\otimes \ell}$ for $b\in\binset$, and satisfy:
    \[ \text{For any } \alpha \in\sbra*{1, 1\!+\!\frac{1}{n \ell}},\quad \frac{1}{2} - \frac{1}{2} \exp\rbra*{-\frac{\ell}{2} \cdot \Talpha(\rho_0,\rho_1)^2} \leq \Talpha\rbra*{\rho_0^{\otimes \ell}, \rho_1^{\otimes \ell}} \leq \ell \cdot \Talpha(\rho_0,\rho_1).\]
\end{lemma}

\begin{proof}
    We start by proving the upper bound through a direct calculation, which holds for any $\alpha \geq 1$. For convenience, let $\rho_b^{\otimes 0}$ denote $1$ for $b\in\binset$. Then, it follows that: 
    \begin{align*}
        \Talpha\rbra*{\rho_0^{\otimes \ell}, \rho_1^{\otimes \ell}}
        & \leq \sum_{i\in [\ell]} \Talpha\rbra*{\rho_0^{\otimes(\ell-i)}\otimes \rho_0 \otimes \rho_1^{\otimes(i-1)}, \rho_0^{\otimes(\ell-i)}\otimes \rho_1 \otimes \rho_1^{\otimes(i-1)}} \\
        & = \sum_{i\in [\ell]} \frac{1}{2} \cdot \rbra*{\tr\abs*{\rho_0^{\otimes(\ell-i)}\otimes (\rho_0-\rho_1) \otimes \rho_1^{\otimes(i-1)}}^{\alpha}}^{1/\alpha}\\
        & = \sum_{i\in [\ell]}\rbra*{\tr \rho_0^\alpha}^{(\ell-i)/\alpha} \cdot \frac{1}{2} \rbra*{\tr\abs*{\rho_0-\rho_1}^{\alpha}}^{1/\alpha} \cdot \rbra*{\tr \rho_1^\alpha}^{(i-1)/\alpha}\\
        & \leq \ell \cdot 1 \cdot \frac{1}{2} \rbra*{\tr\abs*{\rho_0-\rho_1}^{\alpha}}^{1/\alpha} \cdot 1
        = \ell \cdot \Talpha(\rho_0,\rho_1).
    \end{align*}
    Here, the first line follows from the triangle inequality for $\Talpha$ (\Cref{lemma:Talpha-triangle-inequality}), the third line owes to the fact that $\tr\abs*{A \otimes B} = \tr\abs*{A} \cdot \tr\abs*{B}$ and $\rho_b$ is positive semi-definite for $b\in\binset$, and the last line uses the fact that $\tr(\rho^\alpha) \leq 1$ for every quantum state $\rho$ and every $\alpha \geq 1$.
    
    \vspace{1em}
    Next, we establish the lower bound. Leveraging the inequalities between the trace distance and the fidelity (\Cref{lemma:traceDist-vs-fidelity}), we can obtain that: 
    \begin{equation}
        \label{eq:traceDistPower}
        \td\rbra*{\rho_0^{\otimes \ell}, \rho_1^{\otimes \ell}} 
        \geq 1- \F\rbra*{\rho_0^{\otimes \ell}, \rho_1^{\otimes \ell}}
        = 1- \F\rbra*{\rho_0, \rho_1}^{\ell}
        \geq 1- \rbra*{1-\td(\rho_0,\rho_1)^2}^{\ell/2}.
    \end{equation}

    Combining the inequalities in \Cref{corr:Talpha-vs-T-simplified,eq:traceDistPower}, it holds that:
    \[ \Talpha\rbra*{\rho_0^{\otimes \ell}, \rho_1^{\otimes \ell}} 
    \geq \frac{1}{2} \cdot \td\rbra*{\rho_0^{\otimes \ell}, \rho_1^{\otimes \ell}} 
    \geq \frac{1}{2} - \frac{1}{2} \rbra*{1-\td(\rho_0,\rho_1)^2}^{\ell/2}
    \geq \frac{1}{2}-\frac{1}{2}\exp\rbra*{-\frac{\ell}{2} \cdot \Talpha(\rho_0,\rho_1)^2}.
    \]
    Here, the first inequality follows from the following ratio upper bound:
    \[ \forall \alpha \in \sbra*{1, 1\!+\!\frac{1}{n\ell}}, \quad \frac{\td\rbra[\big]{\rho_0^{\otimes \ell},\rho_1^{\otimes \ell}}}{\Talpha\rbra[\big]{\rho_0^{\otimes \ell},\rho_1^{\otimes \ell}}} \leq \rbra*{2\max\cbra*{\rank\rbra[\big]{\rho_0^{\otimes \ell}},\rank\rbra[\big]{\rho_1^{\otimes \ell}}}}^{1-\frac{1}{\alpha}} \leq 2^{(n \ell+1) \rbra*{1-\frac{1}{\alpha}}} \leq 2,\]
    and the last inequality holds because $1-x \leq e^{-x}$ for all $x$ and $\Talpha(\rho_0,\rho_1) \leq \td(\rho_0,\rho_1)$. This concludes the proof of the desired bounds. 
\end{proof}

Finally, we combine all the results to establish the partial polarization lemma for $\Talpha$: 

\begin{proof}[Proof of \Cref{lemma:partial-polarization-Talpha}]
    Note that $\lambda = \min\cbra*{a^2/b,2} \in (1,2]$, and $\ell \coloneqq \ceil*{\log_{\lambda}(64k)}$. Applying the XOR lemma for $\Talpha$ (\Cref{lemma:XOR-Talpha}) to $(Q_0, Q_1, \ell)$ yields the circuits $(Q'_0, Q'_1)$, which prepare the corresponding states $\rho'_b \coloneqq 2^{-\ell+1} \sum_{b_1\oplus \cdots \oplus b_\ell=b} \rho_{b_1}\otimes \cdots \otimes \rho_{b_\ell}$ for $b\in\binset$. For every $\alpha \geq 1$, these states satisfy the following implications:
    \begin{align*}
        \Talpha(\rho_0, \rho_1) &\geq a \quad \Longrightarrow \quad \Talpha(\rho'_0, \rho'_1) \geq a^\ell;\\
        \Talpha(\rho_0, \rho_1) &\leq b \quad \Longrightarrow \quad \Talpha(\rho'_0, \rho'_1) \leq b^\ell.
    \end{align*}

    Next, note that $m = \ceil*{ \lambda^\ell/\rbra*{32 a^{2\ell}} } \leq \lambda^\ell/\rbra*{16 a^{2\ell}}$ since $\lambda^\ell \geq 64k$ and $a^{2\ell} \leq 1$ for all $0<a\leq 1$ imply $\lambda^\ell/(32a^{2\ell})\geq 2$. Applying the direct product lemma for $\Talpha$ (\Cref{lemma:direct-product-Talpha}) to $(Q'_0,Q'_1,m)$ gives the circuits $(Q''_0,Q''_1)$, which prepare the corresponding states $\rho''_b \coloneqq (\rho'_b)^{\otimes m}$ for $b\in\binset$. For every $\alpha \in \sbra[\big]{1,1+\frac{1}{n\ell m}}$, these states satisfy the following implications:
    \begin{align*}
        \Talpha(\rho_0,\rho_1) &\geq a \quad \Longrightarrow \quad \Talpha(\rho''_0,\rho''_1) 
        \geq \frac{1}{2} - \frac{1}{2} \exp\rbra*{-\frac{m}{2} \cdot a^{2\ell}} 
        \geq \frac{1}{2} - \frac{1}{2} e^{-k},\\
        \Talpha(\rho_0,\rho_1) &\leq b \quad \Longrightarrow \quad \Talpha(\rho''_0,\rho''_1) \leq m b^\ell \leq \frac{1}{16} \rbra*{\frac{\lambda b}{a^2}}^\ell \leq \frac{1}{16}.
    \end{align*}

    We now analyze the time complexity, focusing on upper bounds for $\ell$ and $m$. Since $\lambda \in (1,2]$, it follows that $\ln(\lambda) = \ln(1+(\lambda-1)) \geq \frac{\lambda-1}{2} \geq \Omega\rbra[\big]{\frac{a^2-b}{b}}$, where the first inequality is because $\ln(1+x) \geq x/2$ for every $x\in[0,1]$. This bound implies $\ell = O\rbra[\big]{\frac{\ln{k}}{\ln{\lambda}}} = O\rbra[\big]{\frac{b\ln k}{a^2-b}}$. Consequently, we bound $m$ as $m \leq \frac{1}{16} \cdot \rbra[\big]{\frac{2}{a^2}}^\ell \leq \exp\rbra*{O\rbra[\big]{\frac{b\ln k}{a^2-b} \cdot \ln\rbra[\big]{\frac{2}{a^2}}}}$.
\end{proof}

\subsection{Computational hardness and lower bounds for \texorpdfstring{$\alpha>1$}{} near \texorpdfstring{$1$}{}}
\label{subsec:QSDalpha-hard-regime-hardness}

\begin{theorem}[\QSDalpha{} is \QSZK{}-hard] 
    \label{thm:QSDalpha-QSZK-hard}
    For any positive constant $\delta > 0$ that can be made arbitrarily small, the following holds for all sufficiently large $n$:
    \begin{enumerate}[label={\upshape(\arabic*)},topsep=0.33em, itemsep=0.33em, parsep=0.33em]
        \item For any $1 \leq \alpha(n) \leq 1+\frac{1}{n^{1+\delta}}$, it holds that \label{thmitem:QSDalpha-hard-regime-neg}
        \begin{align*}
             &\forall \tau \in (0,1/2), \quad \QSDalpha\sbra*{ 1 - \gamma_{\delta,\tau}(n), \gamma'_{\delta,\tau}(n)}  \text{ is } \QSZK\text{-hard}, \\
             &\text{where } \gamma_{\delta,\tau}(n) \coloneqq 1 - 2^{-\frac{n+1}{n^{1+\delta}+1}} + 2^{-n^{\tau}-\frac{n+1}{n^{1+\delta}+1}} \text{ and } \gamma'_{\delta,\tau}(n) \coloneqq 2^{-n^{\tau}-\frac{1}{n^{1+\delta}+1}}.
        \end{align*}
        \item For any $1+\frac{1}{n^{1+\delta}} < \alpha(n) \leq 1+\frac{1}{n}$, it follows that \label{thmitem:QSDalpha-hard-regime-half}
        \[ \forall \tau \in (0,1/2), \quad \QSDalpha\sbra*{ \frac{1}{2} - 2^{-n^{\tau}-1}, 2^{-n^{\tau}-\frac{1}{n+1}} } \text{ is } \QSZK\text{-hard}. \] 
    \end{enumerate}
\end{theorem}

\begin{proof}
    Fix any constant $\tau\in(0,1/2)$, and choose another constant $\tau'\in(\tau,1/2)$.
    By \Cref{lemma:QSD-is-QSZKhard}, it follows that $\QSD\sbra[\big]{1-2^{-n^{\tau'}}, 2^{-n^{\tau'}}}$ is \QSZK{}-hard for any constant $\tau \in (0,1/2)$. Let $Q_0$ and $Q_1$ be the corresponding \QSZK{}-hard instance, where these quantum circuits are of polynomial size and prepare the quantum states $\rho_0$ and $\rho_1$, respectively. 

    It suffices to establish \Cref{thmitem:QSDalpha-hard-regime-neg}, because the same reasoning also works with $\delta=0$, which yields \Cref{thmitem:QSDalpha-hard-regime-half}.  
    Using the same hard instance $(Q_0,Q_1)$, we derive the following by applying the inequalities in \Cref{corr:Talpha-vs-T-simplified} and noting that the rank of any $n$-qubit state is at most $2^n$: 
    \begin{itemize}
        \item For \textit{yes} instances, 
        $\td(\rho_0,\rho_1) \geq 1-2^{-n^{\tau'}}$ yields that for any $\tau \in (0,1/2)$, 
        \begin{align*}
        \Talpha(\rho_0,\rho_1) &\geq 2^{(n+1) \cdot \frac{1-\alpha}{\alpha}} \cdot \td(\rho_0,\rho_1) \\
        &\geq 2^{-\frac{n+1}{n^{1+\delta}+1}} \cdot \rbra*{1 - 2^{-n^{\tau}}}\\
        &= 1 - \rbra*{1 - 2^{-\frac{n+1}{n^{1+\delta}+1}}} - 2^{-\frac{n+1}{n^{1+\delta}+1}} \cdot 2^{-n^{\tau}} \coloneqq 1 - \gamma_{\delta,\tau}(n). 
        \end{align*}
        Here, the second line uses the fact that $\frac{1-\alpha}{\alpha} \geq -\frac{1/n^{1+\delta}}{1+1/n^{1+\delta}} = -\frac{1}{n^{1+\delta}+1}$ when $1\leq \alpha \leq 1+\frac{1}{n^{1+\delta}}$, and $1-2^{-n^{\tau'}} \geq 1-2^{-n^\tau}$ for all $\tau'>\tau>0$.
        \item For \textit{no} instances, 
        $\td(\rho_0,\rho_1) \leq 2^{-n^{\tau'}}$ implies that for any $\tau \in (0,1/2)$, 
        \begin{align*}
        \Talpha(\rho_0,\rho_1) &\leq 2^{\frac{1-\alpha}{\alpha}} \cdot \td(\rho_0,\rho_1)\\ 
        & \leq \td(\rho_0,\rho_1) \\
        & \leq 2^{-n^{\tau'}} \\
        &\leq 2^{-\frac{1}{n^{1+\delta}+1}} \cdot 2^{-n^{\tau}} \coloneqq \gamma'_{\delta,\tau}(n).
        \end{align*}
        Here, the second line uses the fact that $2^{\frac{1-\alpha}{\alpha}} \leq 1$ for $1\leq \alpha \leq 1+\frac{1}{n^{1+\delta}}$, and the last line follows that $n^{\tau'}\geq n^{\tau}+\frac{1}{n^{1+\delta}+1}$ for all sufficiently large $n$ since $\tau'>\tau$.        
    \end{itemize}
    
     Since the functions $\gamma_{\delta,\tau}(n)$ and $\gamma'_{\delta,\tau}(n)$ are both monotonically decreasing and converge to zero as $n \to \infty$, we complete the proof by noticing that the promise gap $1-\gamma_{\delta,\tau}(n)-\gamma'_{\delta,\tau}(n)$ remains at least a positive constant for sufficiently large $n$. 
\end{proof}

For any $n$-qubit quantum state $\rho$ of rank $r$, let $\rho_{\ttU}$ be the corresponding $n$-qubit quantum state whose eigenvalues are uniformly distributed over the support of $\rho$. Next, we can establish the following quantum query and sample complexity lower bounds: 
\begin{theorem}[Query complexity lower bounds for \QSDalpha{}]
    \label{thm:query-lower-bounds-QSDalpha}
    For any $1 \leq \alpha(n) \leq 1 + \frac{1}{n}$ and $0 < \epsilon \leq 1/4$, there exist an $n$-qubit state $\rho$ of rank $r$ and the corresponding state $\rho_{\ttU}$ such that deciding whether $\Talpha(\rho,\rho_{\ttU})$ is at least $\epsilon$ or exactly $0$ requires $\Omega\rbra[\big]{r^{1/2}/\epsilon}$ queries in the purified quantum query access model.
\end{theorem}

\begin{proof}
    Fix an even integer $r\leq 2^n$.
    Applying \Cref{lemma:lower-bounds-QSD}\ref{thmitem:query-lower-bound-QSD} with error parameter $2\epsilon$, it suffices to reduce the problem to distinguishing between the cases $\rho=\rho_{\ttU}$ and $\td(\rho,\rho_{\ttU}) \geq 2\epsilon$. Since the quantum Schatten $\alpha$-norm distance $\Talpha(\cdot,\cdot)$ is a metric, it holds that $\Talpha(\rho,\rho_{\ttU})=0$ when $\rho=\rho_{\ttU}$. When the states $\rho$ and $\rho_{\ttU}$ are far apart, by using the inequalities in \Cref{corr:Talpha-vs-T-simplified}, it follows that: 
    \[ \Talpha(\rho,\rho_{\ttU}) \geq (2r)^{\frac{1}{\alpha}-1} \cdot \td(\rho,\rho_{\ttU}) \geq (2r)^{-\frac{1}{n+1}}\cdot 2\epsilon \geq \epsilon. \]
    Here, the second inequality follows because $\alpha(n) \leq 1+\frac{1}{n}$ and $f(r;\alpha) \coloneqq (2r)^{\frac{1}{\alpha}-1}$ is monotonically decreasing in $\alpha$ for fixed $r$, since $\frac{\partial}{\partial\mathrm{\alpha}} f(r;\alpha) = -\frac{\ln(2r)}{\alpha^2} \cdot (2r)^{\frac{1}{\alpha}-1} < 0$. The last inequality holds because $r \leq 2^n$ and $f(r;\alpha)$ is monotonically decreasing in $r$ for fixed $\alpha$, as $\frac{\partial}{\partial r} f(r;\alpha) = \frac{1-\alpha}{\alpha r} \cdot (2r)^{\frac{1}{\alpha}-1} < 0$. This reduction achieves the desired lower bound. 
\end{proof}

Using the same reduction to prove \Cref{thm:query-lower-bounds-QSDalpha}, the rank-dependent sample complexity lower bound in \Cref{lemma:lower-bounds-QSD}\ref{thmitem:sample-lower-bound-QSD} for estimating the trace distance $\td(\cdot,\cdot)$ can be extended to the quantum Schatten $\alpha$-norm distance $\Talpha(\cdot,\cdot)$ with $1 \leq \alpha \leq 1 + \frac{1}{n}$: 
\begin{theorem}[Sample complexity lower bound for \QSDalpha{}]
    \label{thm:sample-lower-bound-QSDalpha}
    For any $1 \leq \alpha \leq 1 + \frac{1}{n}$ and $0 \leq \epsilon \leq 1/4$, there exists an $n$-qubit state $\rho$ of rank $r$ and the corresponding state $\rho_{\ttU}$ such that deciding whether $\Talpha(\rho,\rho_{\ttU})$ is at least $\epsilon$ or exactly $0$ requires $\Omega\rbra[\big]{r/\epsilon^2}$ samples of $\rho$. 
\end{theorem}

\section*{Acknowledgments}
\noindent
The authors thank Zhan Yu for raising a concern about the asymmetric expression $\lambda_{\max}\rbra[\big]{\frac{\rho_0-\rho_1}{2}}$ in \Cref{subsec:discussion} in a previous version, which was later found to be inconsistent with the quantum $\ell_\infty$ distance $\td_{\infty}(\rho_0,\rho_1)$ implicitly defined in \Cref{def:Lalpha-distance}.

The work of Yupan Liu was supported in part by funding from the Swiss State Secretariat for Education, Research and Innovation (SERI), in part by the Ministry of Education, Culture, Sports, Science and Technology (MEXT) Quantum Leap Flagship Program (Q-LEAP) under Grant \mbox{JPMXS0120319794}, in part by Japan Science and Technology Agency (JST) Support for Pioneering Research Initiated by the Next Generation (SPRING) under Grant \mbox{JPMJSP2125} and ``THERS Make New Standards Program for the Next Generation Researchers'', and in part by the Japan Society for the Promotion of Science (JSPS) Grants-in-Aid for Scientific Research (KAKENHI) under Grant \mbox{JP24H00071}. 
The work of Qisheng Wang was supported in part by startup funding from Shanghai Jiao Tong University and in part by the Engineering and Physical Sciences Research Council under Grant \mbox{EP/X026167/1}. 

\bibliographystyle{alphaurl}
\bibliography{Lalpha}

\appendix

\section{\texorpdfstring{\PurePoweredQSDinfty{}}{PurePoweredQSD} is \texorpdfstring{$\CeP$}{C=P}-hard}
\label{sec:PurePoweredQSD-hardness}

\begin{theorem}
    \label{thm:PurePoweredQSDinfty-CeP-hard}
    For all $n \geq 2$, $\PurePoweredQSDinfty\sbra*{1,0}$ is $\CeP$-hard.
\end{theorem}

\begin{proof}    
    Noting that $\coCeP=\NQP$~\cite{ADH97,YY99}, and hence $\CeP=\coNQP$, where $\NQP$ is a subclass of \PP{} that serves as a precise variant of \BQP{} that always rejects for \textit{no} instances, it suffices to show that \PurePoweredQSDinfty{} is \coNQP{}-hard. For any promise problem $(\calP_{\rm yes},\calP_{\rm no}) \in \coNQP[1,1-a(n)]$ with $a(n)\in(0,1)$, we assume without loss of generality that the \coNQP{} circuit $C_x$ has an output length of $n$. 

    To proceed, we adopt the construction from the proof of \Cref{lemma:PureQSD-is-BQPhard} and define a new circuit with output length $n'= n+1$: $C'_x \coloneqq C_x^{\dagger} \CNOT_{\sfO\to \sfF} C_x$, where both $\sfF$ and $\sfO$ are single-qubit registers. We say that $C'_x$ accepts if all qubits yield measurement outcomes of zero. 
    
    Now consider two pure states associated with $Q_0 = I$ and $Q_1=C'_x$: 
    $\ket{\psi_0} \coloneqq \ket{\bar{0}}\otimes \ket{0}_{\sfF}$ and $\ket{\psi_1} \coloneqq  C'_x (\ket{\bar{0}}\otimes \ket{0}_{\sfF})$. A direct calculation yields that
    \[ \abs*{\innerprod{\psi_0}{\psi_1}}^2 = \Pr\sbra*{C'_x \text{ accepts}} = \rbra*{1-\Pr\sbra*{C_x \text{ accepts}}}^2. \]
    
    As a consequence, we complete the proof by considering the following cases: 
    \begin{itemize}
        \item For \textit{yes} instances, $\abs*{\innerprod{\psi_0}{\psi_1}} = 1-\Pr\sbra{C_x\text{ accepts}} =0$, and hence the pure states $\ket{\psi_0}$ and $\ket{\psi_1}$ are orthogonal. Consequently, $\Lambda_{\infty}(\ketbra{\psi_0}{\psi_0},\ketbra{\psi_1}{\psi_1}) = 1$. 
        \item For \textit{no} instances, we have $\abs*{\innerprod{\psi_0}{\psi_1}} = 1-\Pr\sbra{C_x\text{ accepts}} \geq a(n) > 0$. Thus, the pure states $\ket{\psi_0}$ and $\ket{\psi_1}$ are not orthogonal. Consequently, $\Lambda_{\infty}(\ketbra{\psi_0}{\psi_0},\ketbra{\psi_1}{\psi_1}) = 0$. \qedhere
    \end{itemize}
\end{proof}

\end{document}